\documentclass[%
amsmath,amssymb,
10pt,
notitlepage,
]{revtex4-1}

\usepackage[paperwidth=210mm,paperheight=297mm,centering,hmargin=2.3cm,vmargin=2.7cm]{geometry}
\usepackage{marvosym}
\usepackage{amsfonts}
\usepackage{amsthm}
\usepackage{dcolumn}
\usepackage{bm,bbm}
\usepackage{kpfonts}
\usepackage{braket}
\usepackage{color}

\newcommand{\cC}{{\mathcal C}}
\newcommand{\cD}{{\mathcal D}}

\newcommand{\cH}{{\mathcal H}}

\newcommand{\cK}{{\mathcal K}}
\newcommand{\cL}{{\mathcal L}}
\newcommand{\cM}{{\mathcal M}}
\newcommand{\cN}{{\mathcal N}}

\newcommand{\cP}{{\mathcal P}}
\newcommand{\cQ}{{\mathcal Q}}
\newcommand{\cR}{{\mathcal R}}
\newcommand{\cS}{{\mathcal S}}
\newcommand{\cT}{{\mathcal T}}
\newcommand{\cU}{{\mathcal U}}
\newcommand{\cV}{{\mathcal V}}
\newcommand{\cW}{{\mathcal W}}
\newcommand{\cX}{{\mathcal X}}
\newcommand{\cY}{{\mathcal Y}}
\newcommand{\cZ}{{\mathcal Z}}

\newcommand{\fA}{{\mathfrak A}}

\newcommand{\fI}{{\mathfrak I}}
\newcommand{\fJ}{{\mathfrak J}}

\newcommand{\fP}{{\mathfrak P}}

\newcommand{\fT}{{\mathfrak T}}

\newcommand{\bbmC}{{\mathbbm C}}
\newcommand{\bbmE}{{\mathbbm E}}
\newcommand{\bbmN}{{\mathbbm N}}
\newcommand{\bbmR}{{\mathbbm R}}
\newcommand{\bbmQ}{{\mathbbm Q}}
\newcommand{\bbmZ}{{\mathbbm Z}}
\newcommand{\bbmeins}{{\mathbbm 1}}

\newcommand{\prob}{\mathrm{Pr}}

\newcommand{\tr}{\mathrm{tr}}
\newcommand{\supp}{\mathrm{supp}}

\newcommand{\argmin}{\mathrm{argmin}}
\newcommand{\diam}{\mathrm{diam}}

\newcommand{\bfu}{{\mathbf{u}}}
\newcommand{\bfU}{{\mathbf{U}}}
\newcommand{\bfl}{{\mathbf{l}}}

\newcommand{\cq}{{\mathcal{C}\mathcal{Q}}}

\newtheorem{theorem}{Theorem}

\newtheorem{definition}[theorem]{Definition}
\newtheorem{example}[theorem]{Example}

\newtheorem{lemma}[theorem]{Lemma}

\newtheorem{proposition}[theorem]{Proposition}
\newtheorem{remark}[theorem]{Remark}

\begin{document}
\title{Distillation of secret-key from a class of compound memoryless quantum sources}
\author{H. Boche}%
\email{boche@tum.de}
\affiliation{ 
Lehrstuhl f\"ur Theoretische Informationstechnik, Technische Universit\"at M\"unchen, 80290 M\"unchen, Germany
}%
\author{G. Jan\ss en}
\email{gisbert.janssen@tum.de}
\affiliation{ 
Lehrstuhl f\"ur Theoretische Informationstechnik, Technische Universit\"at M\"unchen, 80290 M\"unchen, Germany
}%

\date{\today}

\begin{abstract}
We consider secret-key distillation from tripartite compound classical-quantum-quantum (cqq) sources with free forward public communication under strong security criterion. 
We design protocols which are universally reliable and secure in this scenario. These are shown to achieve asymptotically optimal rates as long as a certain regularity condition is fulfilled 
by the the set of its generating density matrices. We derive a multi-letter formula which describes the optimal forward secret-key capacity for all compound cqq sources being regular in this sense.
We also determine the forward secret-key distillation capacity for situations, where the legitimate sending party has perfect knowledge of his/her marginal state deriving from the source statistics. 
In this case regularity conditions can be dropped. Our results show that the capacities with and without the mentioned kind of state knowledge are equal as long as the source is generated 
by a regular set of density matrices. We demonstrate that regularity of cqq sources is not only a technical but also an operational issue. For this reason, we 
give an example of a source which has zero secret-key distillation capacity without sender knowledge, while achieving positive rates is possible if sender marginal knowledge is provided.
\end{abstract}
\maketitle
\begin{section}{Introduction}
Common randomness shared by communication parties which is moreover uncorrelated to an eavesdropping third party is well known as a valuable resource in information theory. This fact 
becomes apparent e.g. when legitimate parties use a one-time pad coding \cite{vernam26}
procedure to securely randomize codewords. In this way, they can enhance the security of messages sent over an insecure transmission line. \\
A possible way to generate this resource is, to distill it from potentially noisy and insecure correlations preshared by the parties. Design of protocols for 
secret-key distillation initially has been a domain of complexity theory \cite{diffie76}. The complexity approach to cryptography gains its power from exploiting assumed limited computational capabilities
of eavesdropping parties. In such a setting, cooperation of legitimate parties leads to substantial advantages when applying high-complexity protocols. \\ 
A disadvantage of the mentioned kind of strategies is that security is obtained under presuming of limitation of the eavesdropper's technical abilities. Such assumptions are rather restrictive, since 
eavesdropping parties may improve their abilities by technological progress. \\ 
This paper follows the so-called information-theoretic approach to security, where rather the principal statistical limitations of the eavesdropping parties are utilized to obtain security. 
Initiated by works of Ahlswede and Csisz\'ar \cite{ahlswede93a} and Maurer \cite{maurer93}, this direction was intensively studied in the past decades. Recent developments indicate that integrating security
on the physical layer of communication systems more and more heads towards technological application \cite{wilson07}. \\ 
In this flavour, information-theoretic methods of secret-key distillation were also studied for quantum systems in \cite{devetak05}, where the secret-key distillation of classical-quantum-quantum sources 
and completely quantum sources were studied. \\
However, to obtain the mentioned results, the sources under consideration were assumed to produce memoryless outputs with statistical properties being perfectly known to the legitimate parties. 
Since these conditions are rather restrictive in real-world applications, we drop the second of the above conditions and assume presence of a compound memoryless quantum source. In this
model, the source is memoryless, but the legitimate communication parties are assumed to have only incomplete knowledge of the generating density matrix. Instead of knowing it precisely,
they are only provided with knowledge of a subset of states, containing the source states which possibly occur. \\
Consequently, the legitimate parties have to apply protocols which have the property to simultaneously generate secure common randomness for each of the possible states from the prescribed 
set.\\
The contributions of this paper are the following. We consider presence of a compound memoryless classical-quantum-quantum (cqq) source where one receiver is assumed to receive outputs of classical 
nature, while the remaining communication parties receive quantum systems. The classical systems-receiving party is also allowed for public communication of classical messages to support the distillation process. 
We consider a strong security criterion and design for given compound source universal protocols for secret-key distillation. \\ 
It turns out that some compound cqq sources have rather unsuitable structure for generating secret-keys. This fact leads us to introducing a regularity condition on sets of cqq density matrices.
We demand the possible sets of marginal states of the sender-legitimate receiver and sender-eavesdropper systems only to differ in a controllable amount when the derived marginals on the sender's
systems are close. For the class of density matrices fulfilling the regularity property, the approximation methods we develop lead to a proof of achievability, whose optimal rates are also optimal 
in general, i.e. we obtain a full characterization of the forward secret-key capacity. \\
We also consider the case, where the sending party is equipped with sender marginal information (SMI), i.e. perfect knowledge of the marginal distribution on his/her systems deriving from the source statistics. 
The forward secret-key capacities in both mentioned models are shown to be equal as long as the source is regular. Moreover, the formula derived is shown to be also valid for all irregular sources if SMI is granted. \\
The reader may ask for a general proof of validity of the mentioned capacity formula also for case the without SMI. Regarding this question we prove a disappointing negative 
result. The forward secret-key distillation capacities with and without sender state information differ substantially for some compound cqq sources. \\
Things get even harder. The counterexample we introduce shows that the legitimate parties can achieve strictly positive key rates with zero error and perfect security of the key in case they are provided with SMI, while 
they are unable to achieve any positive rate without this additional knowledge. These findings shed some light on the structure of compound cqq sources. Even if there may be weaker regularity conditions than the one 
presented here which lead to a general capacity formula, the notion of regularity bears an operational core. While perfect knowledge of the sender's marginal state does not help to achieve higher forward 
secret-key distillation rates for regular sources, irregularity of the source can split values of both capacities. \\ 
We conclude the paper with applying notions from the general theory of set-valued functions on the issue of regularity of cqq sources. We show that a weaker condition on the source is sufficient to obtain 
regularity in the sense defined. It turns out that lower hemi-continuity of the set-valued functions mapping each marginal distribution on the sender's systems to the sets of sender-eavesdropper and sender-legitimate receiver 
marginals consistent with this distributions is generally sufficient to prove a capacity formula. 

\begin{subsection}*{Related work}
The information-theoretic approach to security was initiated in classical information theory by the works of Ahlswede and Csisz\'{a}r \cite{ahlswede93a}, where among other results 
the capacity for generation of secret-keys from perfectly known tripartite classical memoryless sources with free one-way public communication were determined. The research area opened by the aforementioned 
work led to many relevant considerations. We mention \cite{csiszar00} where the secret-key distillation capacity was also determined under constraints on the classical forward communication rate. An excellent overview of the research activities done on the field can be found in 
Chapter 17 in \cite{csiszar11}. \\
The first results regarding information-theoretic security were obtained under a rather weak security criterion where the measures of security appeared as quantities regularized in blocklength. These notions were improved
in \cite{csiszar96}, and \cite{maurer00}, where first results using so-called strong security criteria were proven. In this work, we define a quantum version of the strong security index known from classical information theory 
(see \cite{csiszar11} for further information).\\ 
The spirit of \cite{ahlswede93a} was injected to quantum information theory by Devetak and Winter\cite{devetak05} where corresponding results were proven under the more general assumption that the correlation used for secret-key generation 
are obtained from outputs of memoryless sources of possibly non-classical nature, while the statistics of the source are perfectly known to the communication 
parties. The compound source model was hardly considered including security constraints even in classical information theory. A first attempt to generalize some of the aforementioned results to the case of classical compound 
memoryless sources was pursued in \cite{boche13c}, and \cite{tavangaran16} under collaboration of the first author of this paper. In the latter
of the mentioned papers, the forward secret-key distillation capacity of a classical compound source was determined for the case of a finite number of possible marginal states on the sender's systems. Also a lower bound on 
the capacity under restriction of the forward public communication was derived therein. 
\end{subsection}

\begin{subsection}*{Outline}
 In Section \ref{sect:notation} we fix notation and introduce some conventions, we freely use in our considerations. We precisely state the relevant definitions and our main results in Section \ref{sect:basic_definitions}. Therein, we also introduce a certain regularity
 condition for sets of cqq density matrices. This condition is defined in terms of Hausdorff continuity of the map connecting each possible marginal probability distribution on the sender's system with the cqq density matrices from the compound source generating 
 set being consistent with it. \\ 
 In Sect. \ref{sect:proof_without_state_knowledge} we provide a full proof of our main result, where we derive a multi-letter capacity formula for the forward secret-key distillation capacity for compound sources which fulfill the mentioned regularity 
 conditions.\\
 We may drop the regularity condition if we assume the sender to have perfect knowledge of his marginal distribution. In Section \ref{section:secret-key_ssi}, we prove a full coding theorem to determine the secret-key distillation capacity 
 of compound memoryless cqq sources in case of sender marginal information (SMI). It turns out that for regular sources, the secret-key distillation capacities with and without SMI are equal. However, the capacities
 do not match in general, when we are facing an irregular source. This claim is substantiated in Section \ref{sect:regularity_condition_counterexample}, where we provide an example of a compound cqq source with a substantial 
 gap between the forward secret-key distillation capacities with and without SMI. In Section \ref{sect:regularity_condition_hemicontinuity}, we apply the general theory of set-valued maps to derive a weaker regularity
 condition. We show that the definitions given in Section \ref{sect:basic_definitions} implicitly allow to weaken the mentioned condition to lower hemi-continuity. \\ 
 Even in case of a fully classical source, the forward secret-key distillation capacity was yet only determined on a lower level of generality\cite{tavangaran16}. Therefore we indicate that our considerations also include the completely classical 
 setting as a special case. 
 We conclude the paper with some general remarks in Section \ref{sect:conclusion}, where we especially point out the relation of our results to the problem of one-way LOCC entanglement distillation from bipartite compound quantum sources. 
 \end{subsection}
\end{section}

\begin{section}{Notation and preliminary results} \label{sect:notation}
 All Hilbert spaces appearing in this work are considered to be finite dimensional complex vector spaces. 
 $\mathcal{L}(\cH)$ is the set of linear maps and $\cS(\cH)$ the set of states (density matrices) on a Hilbert 
 space $\cH$ in our notation. For a finite alphabet $\cX$ and a Hilbert space $\cK$, we denote by
 $\cq(\cX, \cK)$
 the set of classical-quantum channels, i.e. maps from $\cX$ to $\cS(\cK)$. The set of completely positive 
 and trace preserving (c.p.t.p.) maps from $\cL(\cH)$ to $\cL(\cK)$ is denoted $\cC(\cH,\cK)$.
 \\
 Regarding states on multiparty systems, we freely make use of the following convention for a system consisting
 of some parties $X,Y,Z$, for instance, we denote $\cH_{XYZ} := \cH_{X} \otimes \cH_Y \otimes \cH_Z$, while 
 the marginals are labeled by indices assigned to the corresponding subsystems, i.e. $\sigma_{XZ} := \tr_{\cH_Y}(\sigma)$ 
 for $\sigma \in \cS(\cH_{XYZ})$ and so on. 
 The von Neumann entropy of a quantum state $\rho$ is defined by
 \begin{align*}
  S(\rho) := - \tr(\rho \log \rho),
 \end{align*}
 where we denote by $\log(\cdot)$ and $\exp(\cdot)$ the base two logarithms and exponentials throughout this paper.
  For each cq channel $V \in \cq(\cX, \cK)$ and probability distribution $p$ on $\cX$, we define the Holevo quantity by
  \begin{align*}
   \chi(p,V) := S\left(\sum_{x \in \cX} \ p(x) \ V(x) \right) - \sum_{x \in \cX} \ p(x) \ S(V(x)).
   \end{align*}
  Given a quantum state
 $\rho$ on $\cH_{XY}$, we denote the conditional von Neumann entropy of $\rho$ given $Y$ by
 \begin{align*}
  S(X|Y,\rho) := S(\rho) - S(\rho_Y),
 \end{align*}
 the quantum mutual information by
 \begin{align*}
  I(X;Y,\rho) := S(\rho_X) + S(\rho_Y) - S(\rho).
 \end{align*}
 A convenient way of representing systems which have quantum as well as classical subsystems is by coherifying the classical systems. 
 The density
 matrix
 \begin{align}
  \rho := \sum_{x \in \cX} P_X(x) \ket{x} \bra{x} \otimes \rho_x  \ \in \ \cS(\cH_X \otimes \cK_B) \label{conventions_rho}
 \end{align}
 represents a density matrix of a source, where the statistics of a subsystem is driven by a classical random variable $X$ with values 
 in $\cX$ and $\rho_x \in \cS(\cK_B)$ is a density matrix on $\cK_B$ for each $x \in \cX$. A quantum system with Hilbert space 
 $\cH_X := \bbmC^{|\cX|}$ was introduced where each $x \in \cX$ corresponds to the element $\ket{x}$ of a once and for all fixed orthonormal basis 
 (we may assume that this is for each system introduced the canonical basis).  We set the convention to indicate the quantum systems belonging to 
 coherified classical systems by the corresponding random variable. This convention extends to notation of entropic quantities. E.g. 
 \begin{align}
  I(X;B, \rho) = H(X)  + S(\rho_B) - S(\rho). \label{def:quantum_mutual_information}
 \end{align}
 corresponds to the quantum mutual information of the state $\rho$ in (\ref{conventions_rho}). The conditional quantum mutual information of a density matrix
 $\sigma_{ABX}$ is defined
 \begin{align*}
  I(A;B|X,\sigma) := S(\sigma_{AX}) + S(\sigma_{BX}) - S(\sigma_{ABX}) - S(\sigma_X).
 \end{align*}
 If $X$ belongs to a classical system i.e 
 \begin{align*}
  \sigma = \sum_{x \in \cX} P_X(x) \ \ket{x}\bra{x} \otimes \sigma_{AB,x}
 \end{align*}
 with $\sigma_{AB,x}$ being a bipartite density matrix on the remaining systems Hilbert spaces, it holds
  \begin{align*}
  I(A;B|X,\sigma) = \sum_{x \in \cX} \ P_X(x) \ I(A;B, \rho_{AB,x}).
 \end{align*}
 Whenever informational quantities are evaluated on classical systems, we feel free to express them in terms of the corresponding classical informational quantities
 evaluated on the corresponding probability distributions resp. random variables where we completely adopt the notation and calculation rules as presented in Ref \cite{csiszar11}
 if no further reference is given.\\ 
 From \cite{csiszar11}, we also take the definition and properties of types and typical sequences. For given alphabet $\cX$ and $n \in \bbmN$ (which we always regard being finite) we denote
 the set of probability distributions on $\cX$ as $\fP(\cX)$. We will use $[N]$ use as a shortcut for the set $\{1,\dots,N\}$ for each $N \in \bbmN$. 
 The set of types (i.e. empirical distributions) on $\cX^n$ is denoted by $\fT(n, \cX)$, it holds
 \begin{align*}
  |\fT(n,\cX)| \leq (n+1)^{|\cX|} &&(n \in \bbmN).
 \end{align*}
For given type $\lambda \in \cT(n,\cX)$, we denote the set of $\lambda$-typical words in $\cX^n$ by $T_\lambda^n$. For each $\delta > 0$, $p \in \fP(\cX)$, the set of $\delta$-typical
sequences for $p$ in $\cX^n$ is defined by
\begin{align}
 T_{p,\delta}^n := \left\{x^n \in \cX^n: \forall a \in \cX:  |\tfrac{1}{n}N(a|x^n) - p(a)|\leq \delta \ \wedge  \ p(a)=0 \Rightarrow N(a|x^n) = 0 \right\},
\end{align}
where $N(a|x^n)$ is the number of occurrences of $a$ in $x^n$. Several kinds of bounds are known for these sets, we will explicitly employ the bound 
\begin{align}
 p^n\left((T_{p,\delta}^n)^c\right) \leq 2^{-nc\delta^2} \label{general_type_bound}
\end{align}
on the probability of the complement of a $\delta$-typical set, which holds with $c := \tfrac{2}{\ln2}$ for each $\delta > 0$ and large enough $n$. \\ 
For any two nonempty sets $\fI, \fI'$ of states on a Hilbert space $\cH$, the Hausdorff distance (induced by the trace norm $\|\cdot\|_1$) is defined by
 \begin{align*}
  d_H(\fI,\fI') 
  &:= \max \left\{ \sup_{\sigma \in \fI}\inf_{\sigma' \in \fI'}\|\sigma - \sigma'\|_1, 
  \sup_{\sigma' \in \fI'}\inf_{\sigma \in \fI}\|\sigma - \sigma'\|_1 \right\} \\
  &= \inf\{\epsilon > 0: \ \fI' \subset \fI_{\epsilon} \ \wedge \fI \subset \fI'_\epsilon \}
 \end{align*}
where $A_\epsilon$ denotes the $\epsilon$-blowup of $A$ (with regard to $\|\cdot\|_1$) for each set $A$. On the set of subsets of a bounded set, $d_H$ has only finite values. If the set of compact
subsets of such a bounded set are regarded, $d_H$ becomes a full metric. Several properties of the Hausdorff distance are directly inherited from $1$-norm on the underlying space. We will frequently 
use the triangle inequality
\begin{align}
 d_H(A,C) \leq d_H(A,B) + d_H(B,C)  &&(A,B,C \subset \cS(\cH)) \label{hausdorff_triangle}
\end{align}
and monotonicity of $d_H$ under c.p.t.p. maps, i.e. for each $\cN \in \cC(\cH, \cK)$, $A,B \in \cS(\cH)$, it holds
\begin{align}
 d_H(A,B) \geq d_H(\cN(A), \cN(B)),   
\end{align}
where $\cN(C)$ is the image of each set $C \subset \cL(\cH)$ under $\cN$. 
\end{section}
\begin{section}{Basic definitions and main Result} \label{sect:basic_definitions}
In this section we give precise definitions of the secret-key distillations task and the corresponding capacities of compound memoryless cqq sources with and without assumption of SMI. We also
state the two main results Theorem \ref{secret-key_generation-theorem-regular} and Theorem \ref{secret-key_generation-theorem-ssi}.
\begin{subsection}{Source model}
 A \emph{compound memoryless quantum source} generated by a set of density matrices $\fI \subset \cS(\cK)$ on a Hilbert space $\cK$ is the source described by the set of possible output density matrices
 \begin{align*}
  \fI^{\otimes n} := \{\rho^{\otimes n}: \ \rho \in \fI \}
 \end{align*}
 for each blocklength $n \in \bbmN$. This source definition models a situation, where the source statistics is memoryless, but the generating density matrix is not perfectly known to the communication parties.
 They only can be sure that the output statistics is governed by memoryless extensions of a density matrix from $\fI$ (in our case the density matrix is not known perfectly to the legitimate users of the systems, 
 while the eavesdropper is allowed to have this information). \\  To gain some notational flexibility, we will also write $\fI = \{\rho_s\}_{s \in S}$ 
 where $S$ is a suitable index set. In general we do not place restrictions on the set $\fI$ to be finite or countable. The reader may check that most of the technicalities in our considerations would 
 become obsolete if $\fI$ was regarded  
 to be finite. In this paper, the compound sources considered are generated by tripartite 
 classical-quantum density matrix of the form
 \begin{align*}
  \rho := \sum_{x \in \cX} \ p(x) \ket{x}\bra{x} \otimes \rho_{BE,x} \ \in \ \cS(\cH_A \otimes \cH_B \otimes \cH_E)
  \end{align*}
 which is the coherified way to represent a statistics where $A$ receives outputs of a classical source with distribution $p \in \fP(\cX)$, while $B$, and $E$ receive quantum systems with joint state 
 $\rho_{BE,x} \in \cS(\cH_{BE})$, dependent on the letter $x$. If a system is classical, we regard the basis used for coherifying the systems as fixed once and for all (we fix it to be the canonical basis). \\
 Note that $\rho$ can be alternatively described by the pair $(p,V)$ with $p \in \fP(\cX)$ being a probability distribution on $\cX$ and $V \in \cq(\cX, \cH_{BE})$ with 
 \begin{align*}
  V(x) := \rho_{BE,x},
 \end{align*}
  notations, we will use interchangeably. Note that
  \begin{align*}
   I(X;BE, \rho) = \chi(p,V).
  \end{align*}
  We define the class of density matrices in $\cS(\cH_{ABE})$ with classical $A$-system, $\cH_{A} := \bbmC^{|\cX|}$ by
 \begin{align*}
  \cS_{cqq}(\cH_{ABE}) := \left\{\rho = \sum_{x \in \cX} p(x) \ket{x} \bra{x} \otimes \rho_x: \  p \in \fP(\cX) \ \text{and} \ \rho_x \in \cS(\cH_{BE}) \ (x \in \cX) \right\}.
 \end{align*}
  Since bipartite sources with a classical and a quantum subsystem also occur, we also define  
  \begin{align*}
   \cS_{cq}(\cH_{AB}) := \left\{\rho = \sum_{x \in \cX} p(x) \ket{x} \bra{x} \otimes \rho_x: \  p \in \fP(\cX) \ \text{and} \ \rho_x \in \cS(\cH_{B}) \ (x \in \cX) \right\}.
  \end{align*}
 To increase notational flexibility within our considerations, we define for each given set $\fI \subset \cS_{cqq}(\cH_{ABE})$ 
\begin{align*}
  \cP_\fI &:= \left\{p \in \fP(\cX): \ \exists \rho \in \fI \ \text{with}\ \rho_A = \sum_{x \in \cX} p(x) \ket{x} \bra{x} \right\}, \hspace{.2cm} \text{and} \\
  \fI_p &:= \left\{\rho \in \fI :\ \rho_A := \sum_{x \in \cX} p(x) \ket{x} \bra{x} \right \}   
\end{align*}
for each $p \in \cP_{\fI}$. With these notations, $\cP_\fI$ is the set of marginal probability distributions which can occur at the sender's site, while $\fI_p$ collects for each $p$ all 
cqq density matrices which have $p$ as marginal distribution on the sender's systems. For a more efficient notation of the capacity 
formulas appearing below, we also define the following sets of marginal distributions deriving from states in $\fI_p$ by
 \begin{align*} 
  \fI^{AB}_p &:= \{\rho_{AB}: \ \rho \in \fI_p\},  \hspace{.4cm} \fI^{AE}_p=\{\rho_{AE}: \ \rho \in \fI_p\} 
 \end{align*}
 for each $p \in \cP_{\fI}$.  \\ 
 In this paper, the systems labeled $A$, and $B$ belong to the legitimate communication parties, while $E$ labels the systems of the eavesdropper. The definitions in the next section model a situation, 
 where the legitimate parties do not know, which density matrix from $\fI$ governs the source statistics (except the SMI case, where $A$ knows his/her marginal statistics). The eavesdropper 
 instead, may know the source statistics and the protocols applied by the legitimate users. 
\end{subsection}
\begin{subsection}{Secret-key generation from compound cqq sources: Definitions and results}
For a cqq source with fixed density matrix $\rho \in \cS_{cqq}(\cH_{ABE})$, $\cH_A = \bbmC^{|\cX|}$, a secret-key generation protocol for given blocklength $n$ is performed, informally speaking, as follows. 
The $A$-party generates from his/her source output messages $l$ and $m$ where $m$ is the the key value for $A$ and $l$ is broadcasted to the remaining parties via a noiseless channel. 
The legitimate receiver subsequently determines a key value by applying a quantum measurement, which can be chosen according to the received message $l$. This results in a tuple $(K,K',\Lambda,X^n)$ of random variables, 
where 
$K$ ($K'$) is the key random value of $A$ ($B$), $\Lambda$ the random variable representing the public transmission, and $X^n$ the classical random variable initially received by $A$. 
The formal definition for the described type of protocol is as follows. 
\begin{definition}\label{sk_prot_def}
 An \emph{$(n,M,L)$ (forward) secret-key distillation protocol} for states on $\cS_{cqq}(\cH_{ABE})$ is a pair $(T,D)$, with $T: \cX^n \rightarrow \fP([L]\times[M])$ being a stochastic matrix, and 
 $D = \{D_{lm}\}_{l \in [L], m \in [M]}$ being a set of matrices, $0 \leq D_{lm} \leq \bbmeins_{\cH_B}^{\otimes n}$ such that
 \begin{align*}
  \sum_{m =1}^M D_{lm} = \bbmeins_{\cH_B}^{\otimes n} 
 \end{align*}
 holds for each $l \in [L]$. 
\end{definition}
We will also consider the situation, where the sender has full knowledge of the statistics of his/her part of the source. If this is assumed, the sender can choose the stochastic matrix of a protocol 
according to this knowledge. 
 \begin{definition}\label{sk_prot_def_smi}
 An \emph{$(n,M,L)$ (forward) secret-key distillation protocol with sender marginal information (SMI)} for a set $\fI \subset \cS_{cqq}(\cH_{ABE})$ is a family  $(T_p,D)_{p \in \cP_\fI}$, with $(T_p, D)$ being an $(n,M,L)$ forward secret
 key distillation protocol for states on $\cS_{cqq}(\cH_{ABE})$ for each $p \in \fP_\fI$.
\end{definition}
Next, we define the performance of $(n,L,M)$ forward secret-key distillation protocols with and without SMI performed on a compound source generated by a set $\fI := \{\rho_s\}_{s \in S} \subset \cS_{cqq}(\cH_{ABE})$. 
For a protocol with SMI, $(T_p,D)$ is performed, if the cqq density matrix belongs to $\fI_p$. It is convenient, to express the aftermath in coherified manner by the state
\begin{align}
 \rho_{\Lambda K K' E^n,s}  
 := \sum_{l=1}^L \sum_{m,m'=1}^M \sum_{x^n \in \cX^n} 
 & p^n(x^n)\ T_p(l,m|x^n) \ket{l}\bra{l} \otimes \ket{m}\bra{m} \nonumber \\ 
 &\otimes \ket{m'}\bra{m'} \otimes \tr_{\cH_B^{\otimes n}}((D_{lm'}\otimes \bbmeins_{\cH_E}^{\otimes n})V^{\otimes n}(x^n))
 \label{full_protocol_state}
\end{align}
when the state governing the statistics of the source is 
\begin{align*}
 \rho_s = \sum_{x \in \cX} \ p(x) \ \ket{x} \bra{x} \otimes V(x).
\end{align*}We are especially interested in the corresponding marginal state 
\begin{align*}
  \rho_{\Lambda KE^n,s} := \sum_{l=1}^L \sum_{m=1}^M \sum_{x^n \in \cX^n} p^n(x^n)\ T_p(l,m|x^n) \tr(D_{lm'}V_B^{\otimes n}(x^n)) \  \ket{l}\bra{l} \otimes \ket{m}\bra{m} \otimes V_E^{\otimes n}(x^n),
\end{align*}
and the probability distribution $(K_s,K_s')$ belonging to the key given by
\begin{align*}
 P_{KK',s}(m,m') = \braket{m \otimes m', \rho_{KK',s} \ m \otimes m'} = \sum_{l=1}^L \sum_{x^n \in \cX^n} p^n(x^n)\ T_p(l,m|x^n) \tr(D_{lm'}V_B^{\otimes n}(x^n))  
\end{align*}
for all $m,m' \in [M]$. 
The case of application of a protocol without SMI can be regarded as the special case, where $T_p$ does not depend on $p$. The following definition quantifies the performance of each $(n,M,L)$ forward secret-key
distillation protocol with SMI when performed on a set of cqq density matrices. The corresponding definition for the case without SMI is easily obtained by dropping all text in brackets from the following definitions.
\begin{definition}\label{def:sk_prot_perf}
 An $(n,M,L)$ forward secret-key distillation protocol (with SMI) for a set $\fI:= \{\rho_s\}_{s \in S} \subset \cS_{cqq}(\cH_{ABE})$ is an $(n,M,L, \lambda)$ forward secret-key distribution protocol (with SMI) for $\fI$ if the 
 inequalities 
 \begin{enumerate}
  \item $\prob(K_s \neq K'_s) \leq \lambda$, and  
  \item $\log M - H(K_s) + I(K ; E^n \Lambda, \rho_{\Lambda K E^n,s}) \leq \lambda$ \label{security_index_bd}
 \end{enumerate}
 are valid for all $s \in S$. 
 \end{definition}
The first condition above is a bound on the probability that key values mismatch. The second one guarantees for small $\lambda$ that the key is approximately equidistributed and secure. The left hand side of the inequality in \ref{security_index_bd}. 
of the above definition can be regarded as a quantum version of the so-called \emph{security index} introduced in classical information theory \cite{csiszar04}.
For a pair $(K,Z)$ of classical random variables, the \emph{security index} of $K$ against $Z$ is given by the expression 
\begin{align*}
  S_{SID}(K|Z) := \log \supp(P_K) - H(K) + I(K;Z).
\end{align*}
The security index is well-known as a useful criterion for quantifying equidistribution and the degree of decoupling from the eavesdropper (see e.g. \cite{csiszar11} for more information).  From the classical security index, also the above 
introduced quantum version stems its operational significance. 
If $(K,\Lambda,Z)_s$ is a tuple of random variables with $K_s$ being the key random variable, $\Lambda_s$ belonging to the public message of the protocol and $Z_s$ the classical random variable obtained by any measurement on the eavesdropper's system, 
the second condition in Definition \ref{def:sk_prot_perf} implies $S_{SID}(K_s|\Lambda_sZ_s) \leq \lambda$, because 
\begin{align}
 S_{SID}(K_s|\Lambda_s,Z_s) = \log M - H(K_s) + I(K_s;Z_s, \Lambda_s) \leq \log M - H(K_s) + I(K ; E^n \Lambda, \rho_{\Lambda K E^n, s}) \leq \lambda \label{security_index}
\end{align}
holds by the Holevo bound \cite{holevo73}.
\begin{remark}
 In \cite{devetak05}, Devetak and Winter proposed a slightly stronger security criterion to be satisfied instead of the one used in Definition \ref{def:sk_prot_perf}. 
 The authors of the present paper feel that in general the security 
 criterion therein will hardly be satisfied by universal secret-key distillation protocols in general. 
 However, the results in the subsequent sections applied to the case of a compound source $\fI$ with $|\fI| =1$ show that imposing the weaker 
 criterion from Definition \ref{def:sk_prot_perf} does not lead to higher capacities compared to \cite{devetak05} in case of perfectly known source statistics.
 \end{remark}
  \begin{definition}\label{def:achiev_rates}
 A nonnegative number $R$ is called an \emph{achievable secret-key distillation rate for $\fI$ (with SMI)}, if for each $\epsilon > 0, \delta > 0$ exist numbers $n_0$ and $0 < R_c < \infty$ such for each  
 possible marginal state $\rho_A$ there is an $(n,M,L, \epsilon)$ secret-key distillation protocol for $\fI$ (with SMI), such that
 \begin{align}
  M \geq \exp(n(R - \delta)), \hspace{.2cm} \text{and} \hspace{.3cm} L \leq \exp(n R_c) \label{def:achiev_rates_eq}
 \end{align}
 for each $n > n_0$. We define the \emph{forward secret-key capacity of $\fI$ with SMI} by 
 \begin{align}
  K_{\rightarrow, SMI}(\fI) := \sup\{R \geq 0: \ R \ \text{is an achievable secret-key distillation rate for } \ \fI \ \text{with SMI}\}.
 \end{align}
  and the \emph{forward secret-key capacity of $\fI$ without SMI} by 
  \begin{align*}
  K_{\rightarrow}(\fI) := \sup\{R \geq 0: \ R \ \text{is an achievable secret-key distillation rate for } \ \fI \ \text{without SMI}\}.
 \end{align*}
 \end{definition}
  What we define next, is a regularity condition on generating sets for compound cqq sources.
 \begin{definition}[Regularity Condition] \label{def:regularity_condition}
  We call a set $\fI \subset \cS_{cqq}(\cH_{ABE})$ 
  \begin{itemize}
   \item $\epsilon$-\emph{regular}, if there is a $\delta > 0$ such that the implication 
  \begin{align*}
   \|p-q\|_1 < \delta \hspace{.2cm} \Rightarrow \hspace{.2cm} d_H(\fI_p^{AB}, \fI_q^{AB}) + d_H(\fI_p^{AE}, \fI_q^{AE}) < \epsilon
  \end{align*}
  holds for each pair $p,q \in \cP_{\fI}$, where $d_H$ denotes the Hausdorff distance generated by the trace norm distance on the underlying matrix spaces.
 \item \emph{regular}, if $\fI$ is $\epsilon$-regular for each $\epsilon > 0$.
 \end{itemize}
 \end{definition}
 \begin{remark}
  The regularity condition given above aims to cover an as large as possible class of reasonable sets of cqq density matrices under the condition that general protocol constructions are successful. The reader interested
  in detailed discussion of this condition is referred to section \ref{sect:regularity_condition}. In Section  \ref{sect:regularity_condition_hemicontinuity} we show using results from the theory of set-valued functions, 
  that the above condition of regularity can be weakened somewhat to include even a larger class of cqq sources. 
 \end{remark}
 The following two theorems state the main results proven in this paper.
 \begin{theorem}\label{secret-key_generation-theorem-regular}
  Let $\fI$ be a regular set of cqq density matrices on $\cH_{ABE}$. It holds
  \begin{align}
   K_{\rightarrow}(\fI) = \lim_{k \rightarrow \infty} \frac{1}{k} K_{\rightarrow}^{(1)}(\fI^{\otimes k}), \label{secret-key_generation_theorem-regular_formula}
  \end{align}
   where for a set $\fA := \{\sum_{y \in \cY} p(y) \ \ket{y} \bra{y} \otimes \sigma_y\} \in \cS_{cqq}(\cK_{ABE})$, $\cK_A := \bbmC^{|\cY|}$ of cqq density matrices, 
   \begin{align*}
    K_{\rightarrow}^{(1)}(\fA) := \underset{p \in \cP_\fA}{\inf}\ \underset{\Gamma := T \leftarrow U \leftarrow Y_p}{\sup} \left(\underset{\sigma \in \fA_p}{\inf} I(U;B|T, \sigma_{\Gamma}) 
    - \underset{\sigma \in \fA_p}{\sup} I(U;E|T, \sigma_{\Gamma})\right).
   \end{align*}
   The supremum above is over all Markov chains $T \leftarrow U \leftarrow Y_p$ resulting from application of Markov transition matrices $P_{T|U}$, $P_{U|Y}$ on $p$ for each $p \in p$, and 
   \begin{align*}
    \sigma_{TU} := \sum_{y \in \cY} \sum_{t \in \cT} \sum_{u \in \cU} P_{T|U}(t|u) P_{U|Y}(u|y) p(y) \ \ket{t} \bra{t} \otimes \ket{u} \bra{u} \otimes \sigma_y  
   \end{align*}
   for given transition matrices $P_{T|U}$, $P_{U|Y}$ when 
   \begin{align*}
    \sigma = \sum_{y \in \cY} p(y) \ \ket{y} \bra{y} \otimes \sigma_{y}.
   \end{align*}
 \end{theorem}
  
 \begin{theorem}\label{secret-key_generation-theorem-ssi}
  Let $\fI$ be a set of cqq density matrices on $\cH_{ABE}$. It holds
  \begin{align}
   K_{\rightarrow, SMI}(\fI) = \lim_{k \rightarrow \infty} \frac{1}{k} K_{\rightarrow}^{(1)}(\fI^{\otimes k}), \label{secret-key_generation_theorem-regular}
  \end{align}
  where the function $K^{(1)}$ is defined in the preceding theorem. 
  \end{theorem}
  Notice that the inequality 
  \begin{align}
    K_{\rightarrow}(\fI) \leq  K_{\rightarrow, SMI}(\fI)  \label{with_or_without_knowledge}
  \end{align}
  holds for each $\fI$. This can be directly observed from the definitions of achievable rates given above. 
  The next section is devoted to giving a full argument which justifies the claim of Theorem \ref{secret-key_generation-theorem-regular}. Here we give a short outline of 
 the proof. In a sequence of propositions with increasing level of approximation we prepare ourselves for proving the assertion 
 \begin{align}
   K_{\rightarrow}(\fI) \geq  K_{\rightarrow}^{(1)}(\fI) \label{secret-key_generation_theorem_suboptimal}
 \end{align}
 in Proposition \ref{prop:nearlyend}. For this reason, we first design suitable protocols of suboptimal rate for the special case of a source parameterized by a full
 Cartesian product of probability distributions and cq channels. We improve the bounds in Proposition \ref{prop:cqq_markov}, where we derive protocols suitable for the same 
 type of source, but including sender preprocessing of the source by a fixed Markov chain for optimization. Finally, this result is combined with a fine-grained approximation
 of an arbitrary regular source by a \label{finite} number of sources of type subject to the mentioned propositions. 
 The proof of achievability (i.e. the lower bound in (\ref{secret-key_generation_theorem-regular})) follows almost immediately from 
 (\ref{secret-key_generation_theorem_suboptimal}), since we show that regularity of $\fI$ implies, for each $k \in \bbmN$, regularity of the set $\fI^{\otimes k}$
 of all $k$-fold tensor extensions for states from $\fI$. \\
 In Section \ref{section:secret-key_ssi} we give a full proof of Theorem \ref{secret-key_generation-theorem-ssi}. The achievability part therein is derived from 
 the results gathered in Section \ref{sect:proof_without_state_knowledge}. The protocol construction used for proving achievability in Theorem \ref{secret-key_generation-theorem-regular}
 can be employed in case of SMI. To do so we use a certain type of finite covering on the power set of $\cS_{cqq}(\cH_{ABE})$ to decompose a general set $\fI$ into a finite family of regular sets. Moreover,
 we provide a proof to the corresponding upper bound on the forward key capacity. \\
 The reader may ask, whether Theorem \ref{secret-key_generation-theorem-regular} may hold also without assumption of regularity. We give a negative answer to this question in Section \ref{sect:regularity_condition},
 where an example of a cqq set of density matrices with
 \begin{align*}
  K_{\rightarrow}(\fI) <  K_{\rightarrow, SMI}(\fI)
 \end{align*}
 is established. 
 \end{subsection}
 \end{section}

\begin{section}{Secret-key distillation without state knowledge} \label{sect:proof_without_state_knowledge}
In this chapter, we give a detailed argument to prove Theorem \ref{secret-key_generation-theorem-regular}. The first assertion, we prove is on a restricted type of cqq density matrices. Assume 
$\cQ \subset \fP(\cY)$ be a set of probability distributions and $\cV \subset \cq(\cY,\cK_{BE})$ be a set of cq-channels. We define 
\begin{align*}
 \rho_{(p,V)} := \sum_{y \in \cY} \ p(y) \ \ket{y}\bra{y} \otimes V(y) &&(p \in \cQ,\ V \in \cV),
\end{align*}
and the set 
\begin{align}
 \fJ := \left\{\rho_{(p,V)}\right\}_{(p,V) \in \cQ \times \cV}. \label{full_cartesian_defined}
\end{align}
We set $V_B = \tr_{\cK_B}\circ V$, and $V_E = \tr_{\cK_E}\circ V$ for each $V \in \cV$. 

 \begin{proposition}\label{prop_cqq_dist}
 Let $\fJ$ be the source defined in (\ref{full_cartesian_defined}), and  $\delta > 0$. 
 There is a constant $c_1 > 0$ and a number $n_0$ such that for each $n > n_0$ there is an $(n,M,L,\mu)$ forward secret-key distillation protocol with
 \begin{align*}
  \frac{1}{n}\log M	&\geq \underset{q \in \cQ} {\inf} \left( \underset{V\in \cV}{\inf} \ \chi(q, V_B) - \underset{V\in \cV}{\sup} \ \chi(q, V_E) \right) - \delta \\
  \frac{1}{n}\log L	&\leq \underset{p \in \cP_\fJ}{\sup}\left(H(p) - \inf_{V \in \cV} \chi(p,V_B)\right)+ \delta \\
		  \mu 	&\leq 2^{-\sqrt[16]{n} c_1}.
 \end{align*}
 \end{proposition}
  Within the proof of Proposition \ref{prop_cqq_dist} we will use some auxiliary results, we introduce first. The following lemma states, for given compound memoryless classical-quantum channel (DMcqC) existence 
  of random codes being of constant composition (i.e. all codewords having the very same type) and equidistributed over the typical sets. The assertion is a direct consequence of coding results stated in Appendix A, 
  where also the basic definitions regarding codes for message transmission over compound DMcq channels can be found.
   \begin{lemma} \label{good_codes_0}
    Let $\cV \subset CQ(\cX, \cK)$ be a set of cq channels. For each $\gamma > 0$, there is a number $n_1(\gamma, \cV)$ such that for each $n > n_1(\gamma)$ and each type $\lambda \in \fT(n,\cX)$ the
    following assertion is true. \\
    There exists a random $(n,M_\lambda)$-code
    \begin{align*}
     \cC(U) := (U_m, D_m(U))_{m=1}^{M_\lambda}
    \end{align*}
    fulfilling the following three properties
    \begin{enumerate}
     \item $U = (U_1,\dots, U_{M_\lambda})$ is an i.i.d. sequence of random variables, such that $U_m$ is equidistributed on $T_\lambda^n$ for each $m \in [M_\lambda]$,
     \item $M_\lambda \geq \exp \left(n \left(\underset{V \in \cV}{\inf} \ \chi(\lambda, V) -\gamma \right)\right)$
     \item $\bbmE\left[\underset{V \in \cV}{\sup} \ \overline{e}(\cC(U), V^{\otimes n}) \right] \leq 2^{-\sqrt[16]{n}\hat{c}}$.
    \end{enumerate}
    with a constant $\hat{c}(\gamma, \cV) > 0$ (independent of $\lambda$).
  \end{lemma}
   \begin{proof}
    We need only consider types with 
    \begin{align}
     \underset{V \in \cV}{\inf} \ \chi(\lambda,V) - \gamma > 0, \label{good_codes_0_type_consider}
    \end{align}
   since for all other types, the bounds in the assertion of the lemma can be satisfied by trivial coding. Setting $\delta := \frac{\gamma}{2}$ in Proposition \ref{good_codes_1} in Appendix A, ensures us,
   that for each large enough blocklength $n$ and each type $\lambda \in \fT(n,\cX)$ the hypothesis of Proposition \ref{good_codes_2} is fulfilled with an $M'_\lambda$ which fulfills
   \begin{align}
    M'_\lambda \ \geq \ \exp\left(n\left(\underset{V \in \cV}{\inf} \ \chi(\lambda, V) - \frac{\gamma}{2}\right) \right) \ > \ 2^{n \frac{\gamma}{2}} \label{good_codes_0_rate_bound}
   \end{align}
   and $\mu \leq 2^{-\sqrt[16]{n}c(\frac{\gamma}{2})}$. Note that the rightmost inequality in (\ref{good_codes_0_rate_bound}) is satisfied because we only consider types, which fulfill the condition in 
   (\ref{good_codes_0_type_consider}). Setting $\vartheta := \frac{1}{2}$, we conclude with Proposition \ref{good_codes_2} that we find, for large enough $n  \in \bbmN$ and random $(n,M_\lambda)$ message transmission code
   fulfilling the properties demanded.
   \end{proof}
  The following matrix covering lemma results from the powerful matrix Chernov bound \cite{ahlswede02} and was proven in \cite{devetak05}.
    \begin{proposition}[\cite{devetak05}, Prop. 2.4] \label{prop_quantum_chernov}
     Let $n \in \bbmN$, $W \in CQ(\cX, \cK)$, $\lambda \in \fT(k,\cX)$, $U := (U_1,\dots,U_M)$ an i.i.d. sequence of random variables generically equidistributed on $T_{\lambda}^n$, and
     \begin{align*}
      \sigma_{n,\lambda}(W) := \frac{1}{|T_\lambda^n|} \ \sum_{x^n \in T_{\lambda}^n} W^{\otimes n}(x^n).
     \end{align*}
     For each $\epsilon, \delta > 0$, there is a number $k := k(\epsilon, \delta)$, such that if $n > k$, then
     \begin{align*}
     \Pr \left( \left \|\frac{1}{M} \sum_{m=1}^M W^{\otimes n}(U_m) - \sigma_{n,\lambda}(W)\right\| \geq \epsilon \right) \leq 2\cdot \dim \cK^n \cdot \exp(-M\Delta_n\cdot \epsilon)
     \end{align*}
     holds with
     \begin{align*}
      \Delta_n := - \frac{1}{288 \ln2}\cdot \exp(-n(\chi(\lambda,W)- \delta))
     \end{align*}
     \end{proposition}
     The next assertion will help us to approximate the set $\cV$ assumed in Proposition \ref{prop_cqq_dist} by a finite subset in a suitable way. 
   \begin{lemma}\label{net-approximation}
    Let $\cV \subset \cq(\cX, \cK)$ be a set of classical quantum channels. For each $\alpha \in (0, \tfrac{1}{e})$ exists a subset $\cV_\alpha \subset \cV$, which fulfills the following three conditions.
    \begin{enumerate}
     \item $|\cV_\alpha| \leq \left(\frac{6}{\alpha}\right)^{2|\cX|\dim \cK^2}$
     \item Given any $n \in \bbmN$, to each $V \in \cV$ exists a $W \in \cV_\alpha$, such that
     \begin{align*}
      \|V^{\otimes n}(x^n) - W^{\otimes n}(x^n) \|_1 \leq 2 n \alpha
     \end{align*}
      holds for each $x^n \in \cX^n$.
     \item For each $p \in \fP(\cX)$, it holds
     \begin{align*}
      \left|\min_{W \in \cV_\alpha} \ \chi(p, W) - \inf_{V \in \cV} \ \chi(p,V) \right| \leq 2 \alpha \log\frac{\dim \cK}{2 \alpha}.
     \end{align*}
    \end{enumerate}
    \end{lemma}

\begin{proof}[Proof of Proposition \ref{prop_cqq_dist}]
Set
\begin{align*}
 R := \underset{q \in \cQ} {\inf} \left( \underset{V\in \cV}{\inf} \ \chi(q, V_B) - \underset{V\in \cV}{\sup} \ \chi(q, V_E) \right),
\end{align*}
and let $\delta > 0$ be a number small enough for fulfilling $R - \delta > 0$, otherwise, there is nothing left to prove. 
Let $\frac{1}{2} > \eta > 0$,  be fixed and small enough such that the inequality
\begin{align}
 12\eta + \log \dim \cK_{BE} + 4 h(2\eta) \leq  \frac{\delta}{16}. \label{prop_cqq_dist_fannes_anfang}
\end{align}
is valid. Let $n \in \bbmN$ be large enough to simultaneously satisfy
\begin{align}
  \frac{1}{n} \leq \frac{1}{16}\delta  \hspace{.3cm} \text{and} \hspace{.3cm} \frac{1}{n} \leq 2\eta. \label{prop_cqq_dist_typenet}
\end{align}
Define 
\begin{align*}
 \fT_n := \fT(n, \cY) \cap \cQ_\eta,
\end{align*}
where $\cQ_{\eta} := \{q \in \fP(\cY): \ \exists p \in \cQ: \|p-q\|_1 \leq \eta \}$ is the $\eta$-blowup of $\cQ$ regarding the variational distance. We set 
for each probability distribution $q \in \fP(\cY)$
\begin{align*}
 \chi_{B,q}	&:= \underset{V \in \cV}{\inf} \ \chi(q, V_B), \\
 \chi_{E,q}	&:= \underset{V \in \cV}{\sup} \ \chi(q, V_E), \\
 \chi_q 	&:= \chi_{B,q} - \chi_{E,q},
\end{align*}
and
\begin{align*}
 \chi_n := \underset{\lambda \in \fT_n}{\min} \chi_\lambda.
\end{align*}
Our choice of $\eta$ and $n$ implies
\begin{align}
 d_H(\fT_n, \cQ) \leq d_H(\fT_n, \cQ_\eta) + d_H(\cQ_\eta, \cQ) \leq \frac{1}{2n} + \eta \leq 2 \eta, \label{prop_cqq_dist_hausdorff_bound}
\end{align}
where the first inequality above is the triangle inequality for the Hausdorff distance, and the second is by (\ref{prop_cqq_dist_typenet}). From 
(\ref{prop_cqq_dist_hausdorff_bound}), and (\ref{prop_cqq_dist_fannes_anfang}) together with twofold application of Lemma \ref{lemma:holevo_bound_2}, we infer
\begin{align}
 |\chi_n - R| \leq \frac{\delta}{16}. \label{prop_cqq_dist_type_rate}
\end{align}
Set, for each $\lambda \in \fT_n$ 
\begin{align*}
 L_\lambda	&:= \left\lfloor \exp\left(n(H(\lambda) - \chi_{B,\lambda} + \tfrac{3}{4}\delta)\right)\right\rfloor, \\
 S_\lambda 	&:= \left\lceil \exp\left(n(\chi_{E,\lambda} + \chi_{\lambda} - \chi_n + \tfrac{3}{4}\delta)\right)\right\rceil, \hspace{.2cm} \text{and} \\
 M 		&:= \left\lfloor \exp\left(n(R - \delta) \right) \right\rfloor.
\end{align*}
The above definitions, together with (\ref{prop_cqq_dist_type_rate}) and the second inequality of (\ref{prop_cqq_dist_typenet}) the bounds
\begin{align}
 M \cdot S_\lambda	&\leq \exp\left(n(\chi_{B,\lambda} - \tfrac{7}{8}\delta)\right), \label{prop_cqq_dist_codesize}\\
 M \cdot L_\lambda	&\leq \exp\left(n(H(\lambda) - \chi_{E,\lambda} - \tfrac{7}{8}\delta) \right), \hspace{.2cm} \text{and} \\
 \Gamma_\lambda 	
 := \frac{L_\lambda \cdot S_\lambda \cdot M}{|T_\lambda^n|} &\geq 2^{n\tfrac{11}{8}\delta}. \label{gamma_lambda_bound}
\end{align}
are valid. The strategy for the rest of the proof is the following. We will in a first step, generate a class of one-way-secret-key distribution protocols 
for $\fJ$, and then show that with high probability, the protocols meet the properties demanded. \\ 
Define for each $\lambda \in \fT_n$ a random matrix 
\begin{align*}
 U^{(\lambda)} := (U^\lambda_{lms})_{(l,m,s) \in [L_\lambda]\times[M]\times[S_\lambda]}
\end{align*}
with all entries being independent and generically equidistributed on $T_\lambda^n$. We collect the matrices defined above in an independent family 
\begin{align*}
	 \bfU := \{U^{(\lambda)}\}_{\lambda \in \fT_n}.
	\end{align*}
Define, for each $y^n \in \cY^n$, $\lambda \in \fT_n$ a random set
\begin{align*}
 A_\lambda(y^n,\bfU) := \left\{(\lambda,l,m,s): \ U^\lambda_{lms} = y^n \right\}.
\end{align*}
Obviously, the sets defined above fulfill for each outcome $\bfu$ of $\bfU$, $\lambda \in \fT_n$ the following relations
\begin{align}
 A_\lambda(y^n,\bfu) 								&= \emptyset &&(y^n \notin T_{\lambda}^n),\nonumber  \\
 A_{\lambda}(y^n,\bfu) \cap A_{\lambda}(z^n,\bfu) 				&= \emptyset &&(y^n \neq z^n), \nonumber \\
 \text{and} \hspace{.3cm} \bigcup_{y^n \in T_{\lambda}^n} A_{\lambda}(y^n,\bfu) &= \{\lambda\} \times [L_\lambda] \times [M] \times [S_\lambda]. \label{prop_cqq_dist_proof_24}
\end{align}
We regard, for each $\lambda \in \fT_n$ and $l \in [L_\lambda]$,
      \begin{align*}
        U_{\lambda,l} := (U^\lambda_{lms})_{(m,s) \in [M]\times [S_\lambda]}
        \end{align*}
	as a random i.i.d. constant composition codebook of size $M\cdot S_\lambda$ with codewords equidistributed over $T_\lambda^n$. Since we have the bound in 
	(\ref{prop_cqq_dist_codesize}), we know from Lemma \ref{good_codes_0}  that there is a random $(n, M\cdot S_\lambda)$ constant composition code
	\begin{align*}
	 \cC_{\lambda,l}(U_{\lambda, l}) := (U^{\lambda}_{lms}, D^{\lambda}_{lms})_{(m,s) \in [M] \times [S_\lambda]}
	\end{align*}
	for the compound DMcqC generated by $\cV_{B} := \{V_B: \ V \in \cV\}$ which has expected average error bounded
	\begin{align}
	 \bbmE \left[\underset{V \in \cV_\lambda}{\sup} \ \overline{e}(\cC_{\lambda, l}(U_{\lambda,l}), V_B^{\otimes n})\right] \leq 2^{-\sqrt[16]{n}\hat{c}} =: \beta_{0} \label{prop_cqq_dist_proof_18}
	\end{align}
        with a strictly positive constant $\hat{c}$ independent of $\lambda$. 
       	Define, for $\beta_3 > 0$, $\lambda \in \fT_n$ a random set
	\begin{align*}
        B_\lambda(\bfU,\beta_{3}) := \left\{l \in [L_\lambda]: \ \underset{V \in \cV}{\max} \ \overline{e}\left(\cC_{\lambda, l}(U_{\lambda,l}), V_B^{\otimes n}\right) < \beta_{3} \right\},
        \end{align*}
	which collects all indices $l \in [L_\lambda]$, such that $\cC_{\lambda,l}$ is $\beta_{3}$-good regarding the average error criterion. Define a random stochastic matrix
        \begin{align*}
         T_{\bfU}: \ \cX \rightarrow \fP\left(\fT(n,\cY) \times [L_\lambda] \times [M] \times [S_\lambda]\right)
        \end{align*}
        with entries 
        \begin{align*}
         T_{\bfU}(\lambda, l, m , s | y^n) := \begin{cases}
                                               |A_\lambda(y^n, \bfU)|^{-1} 	& \text{if} \ (\lambda,l,m,s) \in A_\lambda(y^n,\bfU) \\
                                               0				& \text{otherwise}
                                              \end{cases}
        \end{align*}
	for each $\lambda \in \fT_n$. The values of $T_{\bfU}(\lambda, l, m , s | y^n)$ with $\lambda$ being not in $\fT_n$ will be of no special interest for us, so they may be defined in any consistent way. 
	Let, for each $\lambda \in \fT_n$, $V \in \cV$
	\begin{align*}
	 \sigma_\lambda(V) := \frac{1}{|T_\lambda^n|} \sum_{y^n \in T_\lambda^n} V_E^{\otimes n}(y^n).
	\end{align*}
	Note that 
	\begin{align}
	 \sigma_\lambda(V) = \bbmE\left[V_E^{\otimes n}(U^\lambda_{lms})\right] \label{prop_cqq_dist_proof_15}
	\end{align}
	holds for all $(l,m,s) \in [L_\lambda]\times [M] \times [S_\lambda]$. Define, for $\lambda \in \fT_n$, $\beta_{1},\beta_{2},\beta_{3} > 0$ and each outcome $\bfu$ of $\bfU$, the following sets.
	\begin{align*}
	 C^{(1)}_{\lambda}(\beta_{1}) &:= \left\{\bfu: \ \forall y^n \in T_\lambda^n: \ (1-\beta_{1})\Gamma_\lambda \leq |A_\lambda(y^n,\bfu)| \leq (1+\beta_{1} \Gamma_\lambda) \right\} \\
	 C^{(2)}_{\lambda}(\beta_{2}) &:= \left\{\bfu: \ \forall (l,m) \in [L_\lambda]\times [M], V \in \cV: \|\frac{1}{S_\lambda} \sum_{s=1}^{S_\lambda}V_E^{\otimes n}(u^\lambda_{lms}) 
		    - \sigma_\lambda(V)\|_1 \leq \beta_{2}\right\} \\
	 C^{(3)}_{\lambda}(\beta_{3}) &:= \left\{\bfu: \ |B_\lambda(\bfu, \beta_{3})| \geq (1 - 2\beta_{3})L_\lambda \right\}, \hspace{.3cm} \text{and} \\
	 A				&:= \bigcap_{\lambda \in \fT_n}\bigcap_{i=1}^3 \ C^{(i)}_\lambda(\beta_{i}).
	\end{align*}
	Eventually, we will show that if an outcome $\bfu$ of $\bfU$ is an element of $A$, it generates a suitable protocol for our needs. First we make sure that for the right choice of parameters and each large enough
	blocklength, $A$ is actually nonempty, which we do by actually bounding the r.h.s. of
	\begin{align}
	 \Pr(A^c) \leq \sum_{i=1}^3 \sum_{\lambda \in \fT_n} \Pr\left(C_\lambda^{(i)}(\beta_{i})^c\right) \label{prop_cqq_dist_proof_17}
	\end{align}
	away from one. In the following we separately derive a bound on each of the summands on the right hand side of (\ref{prop_cqq_dist_proof_17}).
	Let $\lambda$ be a type from $\cT_n$. Note that
	\begin{align*}
	 |A_\lambda(y^n,\bfu)| = \sum_{l=1}^{L_\lambda} \sum_{m = 1}^M \sum_{s =1}^{S_\lambda} \bbmeins_{y^n}(u^{\lambda}_{lms})
	\end{align*}
	holds, where $\bbmeins$ is the indicator function, therefore, 
	\begin{align*}
	 \bbmE[|A_\lambda(y^n,\bfu)|] = \frac{L_{\lambda}M S_{\lambda}}{|T_{\lambda}^n|} = \Gamma_\lambda \geq 2^{-n\tfrac{11}{8}\delta}
	\end{align*}
	where the rightmost inequality above results from (\ref{gamma_lambda_bound}). We infer 
	\begin{align*}
	 \Pr\left(C_\lambda^{(1)}(\beta_{1})^c\right) 
	 &= \sum_{y^n \in \cY^n} \Pr\left(\{\bfu: \ |A_\lambda(y^n, \bfu)| \notin ((1 - \beta_1)\Gamma_\lambda,(1 + \beta_1)\Gamma_\lambda )\} \right) \\
	 &\leq 2 |\cY|^n \cdot \exp\left(- \Gamma_\lambda \beta_{1}^2/(2\ln 2) \right) \\
	 &\leq 2 \cdot \exp(-2^{n\frac{9}{8}\delta})
	\end{align*}
	for large enough blocklength $n$, where the first inequality above is by Chernov-bounding with Proposition \ref{classical_chernov_bound}, and the second is by (\ref{gamma_lambda_bound}) together with the choice
	\begin{align}
	 \beta_{1} = 2^{-n\tfrac{\delta}{8}}. \label{prop_cqq_dist_beta_1}
	\end{align}
	To bound the summands with $i=2$, we choose an approximating set $\hat{\cV}_n$ for $\cV$ according to Lemma \ref{net-approximation} with parameter 
        \begin{align*}
         \alpha := 2^{-\sqrt[16]{n}\hat{c}/(16 |\cY|\dim\cK_{BE}^2)}.
        \end{align*}
	which is possible with cardinality
	\begin{align*}
	 |\hat{\cV}_n| \leq 2^{-\sqrt[16]{n}\tfrac{\hat{c}}{4}}
	\end{align*}
	as long as $n$ is large enough. Let for given $V \in \cV$, $W \in \hat{\cV}_n$ be a channel, such that $\|V(x) - W(x)\| \leq 2\alpha$. It holds for each $\lambda \in \fT_n$, $l \in [L_\lambda]$, $m \in [M]$
	\begin{align}
	 \|\frac{1}{S_\lambda} \sum_{s =1}^{S_\lambda}  V_E^{\otimes n}(y^n) - \sigma_\lambda(V)\|_1
	 & \leq \|\frac{1}{S_\lambda} \sum_{s =1}^{S_\lambda}  W_E^{\otimes n}(y^n) - \sigma_\lambda(W)\|_1 \\
	 &+ \|\frac{1}{S_\lambda} \sum_{s =1}^{S_\lambda}  V_E^{\otimes n}(y^n) - \frac{1}{S_\lambda} \sum_{s =1}^{S_\lambda}  W_E^{\otimes n}(y^n) \|_1 \nonumber \\
	 &+ \|\sigma_\lambda(V) - \sigma_\lambda(W) \|_1 \nonumber \\
	 &\leq \|\frac{1}{S_\lambda} \sum_{s =1}^{S_\lambda}  W_E^{\otimes n}(y^n) - \sigma_\lambda(W)\|_1 + 2 n \alpha. \label{netz_unterbound}
	\end{align}
	If we now choose 
	\begin{align}
	\beta_2 = 4n\alpha,   \label{prop_cqq_dist_beta_2}
	\end{align}
	we can bound
	\begin{align*}
	 &\Pr\left(C_\lambda^{(2)}(\beta_{2})^c\right) \\
	 & = \Pr\left(\exists(l,m), V \in \cV: \ \left\|\frac{1}{S_\lambda} \sum_{s=1}^{S_\lambda}V_E^{\otimes n}(u^\lambda_{lms}) - \sigma_\lambda(V)\right\|_1 > \beta_{2} \right) \\
	 &\leq \Pr\left(\exists(l,m), V \in \hat{\cV}_{n}: \ \left\|\frac{1}{S_\lambda} \sum_{s=1}^{S_\lambda}V_E^{\otimes n}(u^\lambda_{lms}) - \sigma_\lambda(V)\right\|_1 > \frac{\beta_{2}}{2} \right) \\
	 & \leq 4 \cdot L_\lambda M |\hat{\cV}_{n}|\cdot  (\dim \cK_E)^n \exp\left(-S_\lambda \cdot 2^{-n(\chi_{E,\lambda} - \frac{\delta}{4})}\frac{\beta_{2}}{576\ln 2} \right) \\
	 &\leq \exp\left(2^{-n\frac{\delta}{4}}\right),
	\end{align*}
	where the first inequality is by (\ref{netz_unterbound}), the second by application of Proposition \ref{prop_quantum_chernov}, and the last inequality holds for each large enough blocklength. 
	At last,
	 \begin{align}
	 \Pr\left(C_\lambda^{(3)}(\beta_{3})^c\right) 
	 &= \Pr\left(\frac{1}{L_\lambda} \sum_{l \in L_\lambda} \bbmeins_{B_\lambda(\beta_{3},\bfU)^c}(l) \geq 2 \beta_{3}\right) \nonumber \\
	 &< \frac{\bbmE\left[\bbmeins_{B_\lambda(\beta_{3},\bfU)^c}(l)\right]}{2 \beta_{3}} \nonumber \\
	 &< \frac{\beta_{0}}{2 \beta_{3}^2}. \label{prop_cqq_dist_proof_19}
	 \end{align}
	 The first inequality above is Markov's inequality applied, the second can be justified as follows. It holds
	 \begin{align}
	  \bbmE\left[\bbmeins_{B_\lambda(\beta_{3},\bfU)^c}(l)\right] 
	  &= \Pr\left(B_\lambda(\bfU, \beta_{3})(l)^c)\right) \nonumber \\
	  &= \Pr\left(\underset{V \in \cV}{\max} \ \overline{e}(\cC_{\lambda,l}, V_B^{\otimes n}) \geq \beta_{3} \right) \nonumber \\
	  &\leq \frac{\bbmE\left[\underset{V \in \cV}{\max} \ \overline{e}(\cC_{\lambda,l}, V_B^{\otimes n})\right]}{\beta_{3}} \nonumber \\
	  &\leq \frac{\beta_{0}}{\beta_{3}}. \label{prop_cqq_dist_proof_20}
	 \end{align}
	 The first inequality above is again by Markov-bounding. The second is by (\ref{prop_cqq_dist_proof_18}). By combination of the estimate in (\ref{prop_cqq_dist_proof_19}),
	  and the choice
	  \begin{align}
	   \beta_{3} = 2^{-\sqrt[16]{n}\hat{c}/4} \label{prop_cqq_dist_beta_3}
	  \end{align}
	  we yield, for large enough blocklength 
	  \begin{align}
	   \Pr\left(C_\lambda^{(3)}(\beta_{3})^c\right) \leq 2^{-\sqrt[16]{n}\hat{c}}.  \label{prop_cqq_dist_proof_21}
	  \end{align}
	  Combining all the bounds derived above with (\ref{prop_cqq_dist_proof_17}), choosing the blocklength large enough, we arrive with our choice of the parameters $\beta_{i}$, $1 \leq i \leq 3$ at
	  \begin{align*}
	   \Pr\left(A^c\right) \leq |\fT_n| 2^{-\sqrt[16]{n}\tilde{c}}
	  \end{align*}
	  with a strictly positive constant $\tilde{c}$ if $n$ is large enough. Since $|\fT_n| \leq |\fT(n,\cY)| \leq (n+1)^{|\cY|}$, we have for each large enough blocklength
	  \begin{align*}
	   \Pr\left(A^c\right) \leq \frac{1}{2},
	  \end{align*}
	  which implies that $A$ is nonempty. Define
	  \begin{align*}
	  L := \underset{\lambda \in \fT_n}{\max}\ L_\lambda, \hspace{.3cm} \text{and} \hspace{.3cm} S := \underset{\lambda \in \fT_n}{\max} S_\lambda.
	  \end{align*}
	  We choose any $\bfu \in A$ and define a stochastic map
	  \begin{align*}
	   T&: \ \cY^n \rightarrow \fT(n,\cY) \times [L] \times [M] \\
	   y^n & \ \mapsto \  T(\lambda, l, m|y^n) := \sum_{s=1}^S T_\bfu(\lambda,l,m,s|y^n).
	  \end{align*}
	  and 
	  \begin{align*}
	   D:= (D_{lm}^\lambda)_{\lambda \in \fT(n,\cY), l \in [L], m \in [M]}, 
	  \end{align*}
	  where 
	  \begin{align*}
	   D^\lambda_{lm} := \sum_{s=1}^S D^\lambda_{lms} 	&&(\lambda,l,m) \in \fT(n,\cY) \times [L] \times [M]
	  \end{align*}
          with $D^\lambda_{lms}$ being the decoding set with index $(m,s)$ from the code $\cC_{\lambda,l}(\bfu)$.
	  Note that some entries of $T_\bfu$ as well as some of the decoding matrices are have been not defined yet (e.g. $L > L_\lambda$ may occur for some $\lambda$)), 
	  we populate the undefined entries of $T_\bfu$ with zeros, and add zero matrices, and add arbitrary but consistently, where decoding matrices are undefined. \\ 
	  To fit the above objects to the definition of a one-way secret-key distillation protocol, we consider each public message as a tuple $\bfl = (\lambda, l)$. With the above definitions, 
	  $\cD := (T,D)$ is an $(n, M, L)$ secret-key distillation protocol with 
	  \begin{align}
	   \frac{1}{n} \log M \geq R - \delta \label{prop_cqq_dist_rate_state}
	  \end{align}
	  by definition of $M$. It remains, to show that actually the bound on the performance $\mu$ stated is fulfilled. We fix an arbitrary member
	  \begin{align*}
	   \rho_{t} = \sum_{y \in \cY} p(y) \ket{y}\bra{y} \otimes V(y)
	  \end{align*}
	  from $\fJ$. We first show that 
	  \begin{align}
	   \prob(K_t \neq K'_t) \leq \mu. \label{prop_cqq_dist_performance}	   
	  \end{align}
	  holds. By construction, we have $\fT(n,\cY) \setminus \subset Q_{\eta}^c$, which together with well known type-bounds implies
	  \begin{align*}
	   p^n\left(\bigcup_{\lambda \in \fT(n,\cY) \setminus \fT_n} T_\lambda^n\right) \leq 2^{-nc\eta^2}
	  \end{align*}
	  with a universal constant $c>0$. It holds
	  \begin{align}
	   \prob(K_t \neq K'_t)
	   &= \sum_{y^n \in \cY^n}  \ p^n(y^n) \cdot \Pr\left(K_t \neq K'_t| Y_t^n = y^n \right) \nonumber \\
	   &\leq \sum_{\lambda \in \fT_n} \sum_{y^n \in T_\lambda^n}  \ p^n(y^n) \cdot \Pr\left(K_t \neq K'_t| Y_t^n = y^n \right) + 2^{-nc\eta^2}. 
	   \label{prop_cqq_dist_error_1}
	  \end{align}
	  We upper-bound for each $\lambda \in \fT_n$ the corresponding summand on the r.h.s. of (\ref{prop_cqq_dist_error_1}). For each $y^n \in T_\lambda^n$, 
	  we have
	  \begin{align}
	    p^n(y^n) \cdot \Pr\left(K_t \neq K'_t| Y_t^n = y^n \right) 
	    &= \sum_{m=1}^M \sum_{m' \neq m} \  p^n(y^n) \cdot \Pr\left(K_t = m, K'_t = m' | Y_t^n = y^n \right) \nonumber \\
	    & = \sum_{m=1}^M \sum_{m' \neq m} \sum_{l=1}^{L} \  p^n(y^n) \cdot \Pr\left(K_t = m, K'_t = m', \Lambda_t = (\lambda,l) | Y_s^n = y^n \right) \nonumber \\
	    & = \sum_{m=1}^M \sum_{m' \neq m} \sum_{l=1}^{L_\lambda} \ \tr\left(D_{lm'}^{\lambda}V_B^{\otimes n}(y^n)\right) \cdot p^n(y^n) \cdot T(\lambda,l,m|y^n)  \nonumber \\
	    & = \sum_{m=1}^M \sum_{m' \neq m} \sum_{l=1}^{L_\lambda} \sum_{s,s'=1}^S \ \tr\left(D_{lm's'}^{\lambda}V_B^{\otimes n}(y^n)\right) \cdot p^n(y^n) \cdot T_{\bfu}(\lambda,l,m,s |y^n). \label{prop_cqq_dist_proof_22}
	  \end{align}
	  On the r.h.s. of (\ref{prop_cqq_dist_proof_22}), only the summands survive, where $(\lambda, l,m,s)$ is in $A_\lambda(y^n, \bfu)$ by definition of $T_\bfu$. For the nonzero summands, we can estimate
	  \begin{align}
	   T_{\bfu}(\lambda,l,m,s|y^n) = \frac{1}{|A_\lambda(y^n,\bfu)|} \leq ((1 - \beta_{3})\Gamma_\lambda)^{-1} \leq 2 \left(\frac{|T_\lambda^n|}{L_\lambda \cdot M \cdot S_\lambda}\right). \label{prop_cqq_dist_proof_23}
	  \end{align}
	  The left of the above inequalities is because $\bfu \in A$, the right holds, if $n$ is large enough. We define the abbreviation $A_\lambda(y^n) := A_\lambda(y^n,\bfu)$. 
	  Counting only the nonzero summands, and using the estimate in (\ref{prop_cqq_dist_proof_23}), we yield
	  \begin{align*}
	  p^n(y^n) \cdot \Pr\left(K_t \neq K'_t| Y_t^n = y^n \right)
	   \leq \sum_{\substack{(l,m,s):\\(\lambda,l,m,s)\in A_\lambda(y^n)}} \sum_{m' \neq m} \sum_{s'=1}^S \ \tr\left(D_{lm's'}^{(\lambda)}V_B^{\otimes n}(u_{lms}^{\lambda})
	   \right) \frac{2|T_\lambda^n|\cdot p^n{(u_{lms}^\lambda})}{L_\lambda S_\lambda M } \\
	    \leq \frac{2}{L_\lambda M_\lambda S_\lambda}\sum_{\substack{(l,m,s):\\(\lambda,l,m,s)\in A_\lambda(y^n)}} 
	    \sum_{m' \neq m}  \sum_{s'=1}^S \ \tr\left(D_{lm's'}^{(\lambda)}V_B^{\otimes n}(u_{lms}^{\lambda}) \right),
	  \end{align*}
	  where in the last inequality, we noted, since $u^{\lambda}_{lms}$ is of type $\lambda$, $|T_\lambda^n|\cdot p^n(u_{lms}^\lambda) = p^n(T_\lambda^n) \leq 1$ holds. Since, for each type $\lambda \in \fT_n$ 
	  \begin{align*}
	   \bigcup_{y^n \in T_\lambda^n} A_\lambda(y^n) = \{\lambda\} \times [L_\lambda] \times [M] \times [S_\lambda] 
	  \end{align*}
	  holds by construction (see (\ref{prop_cqq_dist_proof_24})), we have (with some rearrangements of terms)
	  \begin{align}
	   \sum_{y^n \in T_\lambda^n} p^n(y^n) \cdot \Pr\left(K_t \neq K'_t| Y_t^n = y^n \right)
	   &= \frac{2}{L_\lambda}\sum_{l=1}^{L_\lambda}\frac{1}{S_\lambda M_\lambda} \sum_{m=1}^M \sum_{m' \neq m}  \sum_{s,s'=1}^S \ \tr\left(D_{lm's'}^{(\lambda)}V_B^{\otimes n}(u_{lms}^{\lambda}) \right) \nonumber \\
	   &= \frac{2}{L_\lambda}\sum_{l=1}^{L_\lambda} \ \overline{e}(\cC_{\lambda,l}(u_{\lambda,l}), V_B^{\otimes n}) \nonumber \\
	   &\leq \frac{2}{L_\lambda}(L_\lambda \cdot \beta_{3} + L_\lambda \cdot 2 \beta_{3}) = 6 \beta_{3}. \label{prop_cqq_dist_condition_key_error}
	  \end{align}
	  The last inequality holds, because we have chosen our protocol in a way that for at least a fraction of $1 - 2 \beta_{3}$ of the codes are $\beta_{3}$-good regarding the 
	  average error criterion (i.e. $\bfu$ is in $A \subset C^{(3)}_{\lambda}(\beta_{3})$). \\
	  Collecting inequalities, we arrive at
	  \begin{align}
	   \prob(K_t \neq K'_t) \leq (n+1)^{|\cY|}\cdot 6\beta_{3} + 2^{-nc\eta^2} \ \leq 2^{-\sqrt[16]{n}\hat{c}/8} \label{prop_cqq_dist_cond_error}
	  \end{align}
	  for large enough $n$ by (\ref{prop_cqq_dist_beta_3}). Next, we show that the key is almost equidistributed. For each $\lambda \in \fT(n,\cY)$, we
	  consider the probability distribution $P_{K_t,\lambda}$ on $[M]$, given by
	  \begin{align}
	   P_{K_t,\lambda}(m) 
	   &:= \sum_{y^n \in T_\lambda^n} \frac{\prob(K_t = m|Y^n_t = y^n)}{|T_\lambda^n|} \nonumber \\
	   &= \sum_{y^n \in T_\lambda^n}\sum_{l=1}^{L_\lambda} \sum_{s = 1}^{S_\lambda}  \frac{T_u(\lambda,l,m,s|y^n)}{|T_\lambda^n|} \nonumber \\
	   &= \sum_{l=1}^{L_\lambda} \sum_{s = 1}^{S_\lambda}  \frac{T_u(\lambda,l,m,s|u^\lambda_{lms})}{|T_\lambda^n|}  \label{prop_cqq_dist_equi_1}
	  \end{align}
	  Using the properties of the protocol constructed together with (\ref{prop_cqq_dist_equi_1}), we arrive at
	  \begin{align*}
	    \frac{1}{1+\beta_1} \frac{1}{M} \leq P_{K_t,\lambda}(m) \leq \frac{1}{1-\beta_1} \frac{1}{M},
	  \end{align*}
	  for each $\lambda \in \fT_n$, from which we infer that 
	  \begin{align*}
	   \|P_{K,\lambda} - \pi_{[M]} \|_1 \leq 2 \beta_1
	  \end{align*}
	  is true for all $\lambda \in \fT_n$. We conclude
	  \begin{align*}
	   \|P_{K_t} - \pi_{[M]}\|_1 \leq \sum_{\lambda \in \fT_n} p^n(T_\lambda^n)\ \|P_{K_t,\lambda} - \pi_{[M]} \|_1 \leq 2 \beta_1 + 2 \cdot 2^{-nc\eta^2} 
	   \leq 3 \beta_1,
	  \end{align*}
	   where the last inequality holds if $n$ is large enough. This implies
	   \begin{align*}
	    H(K_t) \geq \log M - 3 \beta_1 \log\frac{M}{3 \beta_1} \geq \log M - \frac{\mu}{2}
	   \end{align*}
	   if $n$ is large enough. It remains to bound $I(K;E^n,\Lambda, \rho_{\Lambda K E^n, t})$. We will actually show that $\rho_{\Lambda K E^n, t}$ is close to a state $\gamma_t$
	   whith $I(K;E^n,\Lambda, \gamma_t) = 0$. Define
	   \begin{align*}
	    \gamma_t := \sum_{\lambda \in \fT(n,\cY)} p^n(T_\lambda^n) \ket{\lambda} \bra{\lambda} \otimes \gamma_{t,\lambda} 
	   \end{align*}
	  where we set for each $\lambda \in \fT(n,\cY)$
	  \begin{align*}
	   \gamma_{t,\lambda} := \frac{1}{L_\lambda \cdot M} \sum_{l = 1}^{L_\lambda} \sum_{m=1}^M \ \ket{l \otimes m} \bra{l \otimes m} \otimes \sigma_\lambda(V).
	  \end{align*}
	   We write $\rho_{\Lambda K E^n, t}$ in the form 
	   \begin{align*}
	  \rho_{\Lambda K E^n, t} = \sum_{\lambda \in \fT(n, \cY)} p^n(T_\lambda^n) \ket{\lambda}\bra{\lambda} \otimes \tilde{\rho}_{t, \lambda},
	   \end{align*}
	  where we defined
	  \begin{align*}
	   \tilde{\rho}_{t,\lambda} := \sum_{y^n \in T_\lambda^n} \sum_{l = 1}^{L_\lambda} \sum_{m=1}^M \sum_{s=1}^{S_\lambda} \frac{T_u(\lambda,l,m,s|y^n)}{|T_\lambda^n|} \ket{l \otimes m} \bra{l \otimes m} \otimes V_E^{\otimes n}(y^n).
	  \end{align*}
	  We first consider $\lambda \in \fT_n$. Note that if $(\lambda, l, m , s)$ is a member of $A_\lambda(\bfu, y^n)$,
	  \begin{align*}
	    \frac{1}{1+\beta_1} \Gamma_\lambda^{-1} \leq T_{\bfu}(\lambda, l, m, s|y^n) \leq \frac{1}{1-\beta_1} \Gamma_\lambda^{-1},
	  \end{align*}
	  while being zero otherwise, which implies
	  \begin{align*}
	   \sum_{y^n \in t_\lambda^n} \sum_{l=1}^{L_\lambda} \sum_{m=1}^M \sum_{s=1}^{S_\lambda} \left|\frac{T_{\bfu}(\lambda,l,m,s|y^n)}{|T_\lambda^n|} - \frac{1}{L_\lambda M |T_\lambda^n| S_\lambda} \right| \leq 2 \beta_1.
	  \end{align*}
	  Also, we know that for each $l \in L_\lambda$, $s \in S_\lambda$
	  \begin{align*}
	   \left\|\frac{1}{S_\lambda}\sum_{s=1}^{S_\lambda} V_E^{\otimes n}(y^n) \right\|_1 \leq 2 \beta_2.
	  \end{align*}
	  We define
	  \begin{align*}
	   \tau := \sum_{y^n \in T_\lambda^n} \sum_{l=1}^{L_\lambda} \sum_{m =1}^M \sum_{s=1}^{S_\lambda} \frac{1}{L_\lambda \cdot M \cdot S_\lambda \cdot |T_\lambda^n|} \ket{l \otimes m} \bra{l \otimes m} \otimes V_E^{\otimes n}(y^n),
	  \end{align*}
	  and obtain
	  \begin{align*}
	   \|\tilde{\rho}_{t,\lambda} - \gamma_{t,\lambda}\|_1 
	   &\leq \|\tilde{\rho}_{t,\lambda} - \tau \|_1 + \|\tau -  \gamma_{t,\lambda} \|_1 \\
	   &\leq \sum_{y^n \in T_\lambda^n} \sum_{l=1}^{L_\lambda} \sum_{m=1}^M \sum_{s=1}^{S_\lambda} \left|\frac{T_{\bfu}(\lambda,l,m,s|y^n)}{|T_\lambda^n|} - \frac{1}{L_\lambda M |T_\lambda^n| S_\lambda} \right|
	     + \left\|\sum_{y^n \in T_\lambda^n} V_E^{\otimes n}(y^n) - \sigma_\lambda(V) \right\|_1 \\
	   &\leq 2 (\beta_1 + \beta_2).
	  \end{align*}
	  Therefore, 
	  \begin{align*}
	   \|\rho_{K\Lambda E^n, t} - \gamma_t \|_1 
	   &\leq \sum_{\lambda \in \fT(n,\cY)} p^n(T_\lambda^n) \ \|\tilde{\rho}_{t,\lambda} - \gamma_{t,\lambda}\|_1 \\
	   &\leq 2(\beta_1 + \beta_2) + 2^{-nc\eta^2} \\
	   &\leq \frac{\mu}{12} 
	   \end{align*}
 	  By using the well-known Alicki-Fannes \label{alicki04} bound for the quantum mutual information, we infer
 	  \begin{align*}
 	  I(K;\Lambda,E^n, \rho_{K\Lambda E^n, t}) \leq I(K;\Lambda,E^n, \gamma_{t}) +  \frac{1}{2} \mu \log(L\cdot \dim \cK_E^{\otimes n}) + h\left(\frac{\mu}{12}\right) \leq \mu,
	  \end{align*}
	  where the last inequality is by the fact that $\gamma_t$ is an uncorrelated state, together with a large enough choice of $n$, and $h$ is the binary Shannon entropy. 
	  \end{proof}
	  Next, we prove an achievability result for the same simple kind of compound source as in the previous proposition, but the lower bound on the key rate derived is including possible
	  preprocessing of the source outputs for the sender by stochastic matrices. For each set $A$ of probability distributions on $\cY$, we denote its diameter (regarding the variational distance) by
	 \begin{align*}
	  \diam(A) := \sup\{\|q - q'\|_1:\ q,q' \in A \}.
	 \end{align*}
	 \begin{proposition} \label{prop:cqq_markov}
	  Let $\cP \subset \fP(\cY)$ be a set of probability distributions, $\diam(\cP)\leq \Delta$, $\cV \subset \cq(\cY, \cK_{BE})$, and $\cU, \cT$  finite alphabets. Define
	  \begin{align*}
	   \fI := \left\{\rho_{(p,V)}:= \sum_{y \in \cY} p(y) \ \ket{y} \bra{y} \otimes V(y)\right\}_{(p,V) \in \cP \times \cV}. 
	  \end{align*}
	  For each $P_{T|U}: \cU \rightarrow \fP(\cT)$, $P_{U|Y}:\ \cY \rightarrow \fP(\cU)$ stochastic matrices, and $\delta > 0$, there is a number $n_0$, such that for each $n > n_0$ we 
	  find an $(n,M,L,\mu)$-secret-key distillation protocol for $\fI$ which fulfills
	  \begin{align}
	   \mu			&\leq 2^{-\sqrt[16]{n}c_2}  \nonumber \\
	   \frac{1}{n}\log L	&\leq  \sup_{p \in \cP}\inf_{\rho \in \fI_p} \ S(U|BT,\tilde{\rho}) + \log|\cT| \nonumber \\
	   \frac{1}{n}\log M 	&\geq \underset{p \in \cP}{\inf} \left(\underset{\rho \in \fI_p}{\inf} I(U;B|T,\tilde{\rho}) - \underset{\rho \in \fI_p}{\sup} I(U;E|T,\tilde{\rho}) \right) - 
	    \delta - 12\Delta \log|\cU| - 4 h(\Delta), \label{prop:cqq_markov_3}
	  \end{align}
	  where $h(x) = -x\log x - (1-x) \log (1-x)$, $(x \in (0,1))$ is the binary entropy, and $c_2$ is a strictly positive constant. We used the definition
	  \begin{align*}
	   \tilde{\rho} := \sum_{t \in \cT} \sum_{u \in \cU} \sum_{y \in \cY} \ P_{T|U}(t|u)P_{U|Y}(u|y)p(y) \ \ket{u} \bra{u} \otimes \ket{t} \bra{t} \otimes V(y)
	  \end{align*}
	  for each state 
	  \begin{align*}
	   \rho = \sum_{y \in \cY}\ p(y)\ \ket{y} \bra{y} \otimes V(y) &&(p \in \cP, V \in \cV).
	  \end{align*}
	 \end{proposition}

	 \begin{proof}
	  For the proof, we define a set of effective cqq density matrices $\hat{\fI}$ on which we apply Proposition \ref{prop_cqq_dist} from which we derive existence of a certain forward secret-key distillation protocol for $\hat{\fI}$. Afterwards, we 
	  show that this protocol can be modified to a forward secret-key distillation protocol for $\fI$ which has the stated properties. Define, for each $p \in \cP$ a probability distribution $q_p \in \fP(\cU)$ by
	  \begin{align*}
	   q_p(u) := \sum_{y \in \cY} P_{U|Y}(u|y)p(y) &&(u \in \cU),
	  \end{align*}
	  a stochastic matrix $W_p: \cU \rightarrow \fP(\cY)$ by
	  \begin{align*}
	   W_p(y|u) := 
	   \begin{cases}
	    \frac{p(y)P_{U|Y}(u|y)}{q_p(u)}	& \text{if} \ q_p(u) > 0 \\
	    \frac{1}{|\cY|}			& \text{otherwise},
	   \end{cases}
	  \end{align*}
	  and a classical-quantum channel $\hat{V}_{p}: \ \cU \rightarrow \cS(\cK_{BE})$ by
	  \begin{align*}
	   \hat{V}_{p}(u) := \sum_{y \in \cY} \ W_p(u|y) V(y)	&&(V \in \cV).
	  \end{align*}
	  Moreover, we define, introducing spaces $\cK_{B'} = \cK_{E'} = \bbmC^{|\cT|}$ a classical-quantum channel $\tilde{V}: \ \cU \rightarrow \cS(\cK_{B'} \otimes \cK_{E'})$ by
	  \begin{align}
	   \tilde{V}(u) := \sum_{t \in \cT} P_{T|U}(t|u)\ \ket{t} \bra{t} \otimes \ket{t} \bra{t} &&(u \in \cU). \label{v_tilde_defined}
	  \end{align}
	  Define for each $(p,p',V) \in \cP \times \cP \times \cV$ a state 
	  \begin{align*}
	   \hat{\rho}_{(p,p',V)} := \sum_{u \in \cU} \ q_p(u) \ \ket{u} \bra{u} \otimes \hat{V}_{p'}(u) \otimes \tilde{V}(u).
	  \end{align*}
	  We define a set of classical-quantum channels $\hat{\cV} := \{\hat{V}_{p'} \otimes \tilde{V}: \ p' \in \cP, V \in \cV\}$ and a set
	  \begin{align*}
	   \hat{\fI} := \left\{\rho_{(p,p',V)}: \ p,p' \in \cP, V \in \cV \right \} \subset \cS_{cqq}(\bbmC^{|\cU|} \otimes \cK_{BE} \otimes \cK_{B'E'})
	  \end{align*}
	  of cqq density matrices. Note that $\hat{\fI}$ meets the specifications of Proposition \ref{prop_cqq_dist} (the states in $\hat{\fI}$ are parameterized by $\cP \times \hat{\cV}$).
	  We apply Proposition \ref{prop_cqq_dist} on $\hat{\fI}$ and infer in case of sufficiently large blocklength existence of an $(n,M,\hat{L}, \mu)$ secret-key distillation protocol $\hat{\cD} = (\hat{T}, \hat{D})$ for
	  $\hat{\fI}$ with a stochastic matrix 
	  \begin{align*}
	   \hat{T}: \ \cU \rightarrow [M] \times [\hat{L}]
	  \end{align*}
	  and 
	  \begin{align*}
	   \hat{D} := \left\{\hat{D}_{lm}\right \}_{(l,m)\in [\hat{L}] \times [M]} \subset \cL(\cK_{BB'}^{\otimes n})
	  \end{align*}
	  being a POVM such that for the key rate, the lower bound
	   \begin{align}
	    \frac{1}{n} \log M 
	    &\geq \underset{p \in \cP}{\inf} \left(\underset{(p,p',V) \in \{p\} \times \cP \times \cV}{\inf} \chi(q_p, \hat{V}_{B,p'} \otimes \tilde{V}_{B'}) - 
	    \underset{(p,p',V) \in \{p\} \times \cP \times \cV}{\sup} \chi(q_p, \hat{V}_{E,p'} \otimes \tilde{V}_{E'})\right) - \delta \nonumber \\
	    &\geq \underset{(p,p',V) \in \cP^2 \times \cV}{\inf} \chi(q_p, \hat{V}_{B,p'} \otimes \tilde{V}_{B'}) - \underset{(p,p',V) \in \cP^2 \times \cV}{\sup} \chi(q_p, \hat{V}_{E,p'} \otimes \tilde{V}_{E'}) 
	    - \delta \label{prop:cqq_markov_initial_rate}
	   \end{align}
	 holds, and for each $s:=(p,p',V)$ the inequalities 
	 \begin{align}
	    \prob(\hat{K}_s \neq \hat{K}'_s) &\leq \mu, \ \text{and}  \nonumber \\
	   \log M  - H(\hat{K}_s) + I(\hat{K};E^n{E'}^n\hat{\Lambda}, \hat{\rho}_{\hat{K}\hat{\Lambda} E^nE'^n,s} )  	&\leq \mu    \label{cqq_markov_initial security}
	 \end{align}
	  being satisfied with $\mu = 2^{-\sqrt[16]{n}c_2}$ with a constant $c_2 > 0$. Notice that the cq-channel $\tilde{V}_B$ defined in (\ref{v_tilde_defined}) has classical structure in the sense that all its output quantum states are diagonal in the 
	  orthogonal basis $\{t\}$. Consequently, we can assume that for each $l \in [\hat{L}]$, $m \in [M]$ the corresponding effect $\hat{D}_{lm}$ has the form
	  \begin{align*}
	   \hat{D}_{lm} = \sum_{t^n \in \cT^n} D_{lt^nm} \otimes \ket{t^n} \bra{t^n}.
	  \end{align*}
	  We define the POVM
	  \begin{align*}
	   D := \{D_{lt^nm}\}_{(l,t^n,m) \in [\hat{L}] \times \cT^n \times [M]}
	  \end{align*}
	  and the stochastic matrix $T: \cY \rightarrow [\hat{L}] \times \cT^n \times [M]$ by
	  \begin{align}
	   T(l,t^n,m|y^n) := \sum_{u \in \cU^n} P^n_{T|U}(t^n|u^n)\hat{T}(l,m|u^n)P_{U|Y}^n(u^n|y^n) &&((l,t^n,m,y^n) \in [L] \times \cT^n \times [M] \times \cY^n). \label{cqq_markov_stoch_def}
	  \end{align}
	  With these definitions, $\cD := (T,D)$ is an $(n,M,\hat{L}\cdot|\cT|^n)$ secret-key distillation protocol for $\fI$. It holds for each $s := (p,V) \in \cP \times \cV$, $m,m' \in [M]$
	  \begin{align}
	  P_{KK',s}(m,m') 
	  &= \sum_{l=1}^{\hat{L}} \sum_{t^n \in \cT^n} \sum_{y^n \in \cY^n} p^n(y^n) T(l,t^n,m|y^n) \tr(D_{lt^nm'}V_B^{\otimes n}(y^n)) \label{cqq_markov_key_id_1} \\
	  &= \sum_{u^n \in \cU^n} \sum_{l=1}^{\hat{L}} \sum_{t^n \in \cT^n} \sum_{y^n \in \cY^n} p^n(y^n) P_{T|U}^n(t^n|u^n) \hat{T}(l,m|u^n) P_{U|Y}^n(u^n|y^n) \tr(D_{lt^nm'}V_B^{\otimes n}(y^n)) \label{cqq_markov_key_id_2}\\
	  &= \sum_{u^n \in \cU^n} \sum_{l=1}^{\hat{L}} \sum_{t^n \in \cT^n} \sum_{y^n \in \cY^n} q_p^n(u^n) W^n_p(y^n|u^n)  \nonumber \\
	  &\times \hat{T}(l,m|u^n) \cdot \tr((D_{lt^nm'}) \otimes \ket{t^n}\bra{t^n})(V_B^{\otimes n}(y^n) \otimes \tilde{V}_{B'}^{\otimes n}(u^n))) \label{cqq_markov_key_id_3}\\
	  &= \sum_{u^n  \in \cU^n} \sum_{l=1}^{\hat{L}} q_p^n(u^n) \hat{T}(l,m|u^n) \tr(\hat{D}_{lm} (\hat{V}_{p}^{\otimes n}(u^n) \otimes \tilde{V}_{B'}^{\otimes n}(u^n))) \label{cqq_markov_key_id_4}\\
	  &= P_{\hat{K}\hat{K}',(p,p,V)}(m,m').   \label{effective-equal}
	  \end{align}
	  The equality in (\ref{cqq_markov_key_id_1}) is holds by definition, (\ref{cqq_markov_key_id_2}) is valid by (\ref{cqq_markov_stoch_def}). The equality in  (\ref{cqq_markov_key_id_3}) is justified by definition of $q_p,\ W_p$, and the 
	  fact that 
	  \begin{align*}
	   \tr(\ket{t}\bra{t}\tilde{V}_B(u)) = P_{T|U}(t|u) &&(t \in \cT, u \in \cU)
	  \end{align*}
          holds by definition of $\tilde{V}$. From (\ref{effective-equal}), we directly infer 
	  \begin{align}
	   \prob(K_s \neq K'_s) &= \prob(\hat{K}_{(p,p,V)} \neq \hat{K}'_{(p,p,V)}) \leq \mu, \  \text{and} \label{cqq_markov_sec_ind_0}\\
	   H(K_s) &= H(\hat{K}_{(p,p',V)})   \label{cqq_markov_sec_ind_1}
	  \end{align}
	  Notice that by definition of $\hat{\rho}_{(p,p,V)}$ is (up to unitaries) equal to $\tilde{\rho}_{(p,V)}$, i.e.
	  \begin{align*}
	   \hat{\rho}_{(p,p,V)} 
	   &= \sum_{u \in \cU} \ q_{p}(u) \ \ket{u} \bra{u} \otimes \hat{V}_p(u) \otimes \tilde{V}(u) \\
	   &= \sum_{t \in \cT} \sum_{u \in \cU} \sum_{y \in \cY} P_{T|U}(t|u)P_{U|Y}(u|y)p(y)\  \ket{u} \bra{u} \otimes V(y) \otimes \ket{t} \bra{t} \otimes \ket{t} \bra{t}
	  \end{align*}
	  holds for each $(p,V) \in \cP \times \cV$.  Consequently, it follows 
	  \begin{align}
	     I(K;\hat{\Lambda}T E, \tilde{\rho}_{K\hat {\Lambda} TE^n,s} ) &= I(\hat{K};EE^n\hat{\Lambda}, \hat{\rho}_{\hat{K}\hat{\Lambda} E^nE'^n,(p,p,V)} ). \label{cqq_markov_sec_ind_2}
	  \end{align}
	  The inequalities contained in (\ref{cqq_markov_sec_ind_1}) and (\ref{cqq_markov_sec_ind_2}) together with the one in (\ref{cqq_markov_initial security}) yield
	  \begin{align*}
	   \log M - H(K_s) +  I(K;\hat{\Lambda}T E, \hat{\rho}_{K\hat {\Lambda} TE^n,s} ) \leq \mu
	  \end{align*}
	  for each $s = (p,V)$, which, together with (\ref{cqq_markov_sec_ind_0}) makes $(T,D)$ an $(n,\hat{L}\cdot|\cT|^n, M, \mu)$ forward secret-key distillation protocol for $\fI$. 
	  At last, we have to show that $M$ indeed satisfies the bound stated in (\ref{prop:cqq_markov_3}). We will therefore, lower-bound the right-hand side of (\ref{prop:cqq_markov_initial_rate}).
	  Because Markov-processing does never increase the variational distance, it holds
	  \begin{align}
	   \|q_p - q_{p'}\|_1 \leq \|p - p'\|_1 \leq \diam(\cP) \leq \Delta \label{prop:cqq_markov_proc_dist}
	  \end{align}
	  for each $p,p' \in \cP$, where the rightmost inequality is by assumption. We obtain for each $V \in \cV$
	  \begin{align}
	   \left|\chi\left(q_{p}, \hat{V}_{B,p} \otimes \tilde{V}_{B} \right) - \chi\left(q_{p'}, \hat{V}_{B,p} \otimes \tilde{V}_{B'} \right) \right| 
	   & \leq 6 \|q_p - q_{p'}\|_1 \ \log|\cU| + 2h(\|q_p - q_{p'}\|_1) \nonumber \\
	   &\leq 6\Delta \log|\cU| + 2 h(\Delta). \label{prop:cqq_markov_fannes_1}
	  \end{align}
	  The first equality above is by application of Lemma \ref{lemma:holevo_bound_1} which can be found in Appendix \ref{appendix:continuity_bounds}, the second by (\ref{prop:cqq_markov_proc_dist}). 
	  The bound in (\ref{prop:cqq_markov_fannes_1}) directly implies
	  \begin{align}
	  \underset{(p',V)}{\inf} \chi \left(q_p, \hat{V}_{B,p'} \otimes \tilde{V}_{B'} \right) 
	  &\geq \underset{V}{\inf} \chi \left(q_p, \hat{V}_{B,p} \otimes \tilde{V}_{B'} \right) - \Delta \log|\cU| - 2 h(\Delta) \nonumber \\
	  &= \underset{\rho \in \fI_p}{\inf} \ I(U,BB', \hat{\rho}) - \Delta \log|\cU| - 2 h(\Delta). \label{prop:cqq_markov_bound_1}
	 \end{align}
	  for each $p \in \cP$. The equality in (\ref{prop:cqq_markov_bound_1}) holds by the identity
	  \begin{align*}
	    \chi \left(q_p, \hat{V}_{B,p} \otimes \tilde{V}_{B'} \right) = I(U,BB', \hat{\rho}_{(p,p,V)}). 
	  \end{align*}
	  By similar reasoning, we also yield the bound 
	  \begin{align}
	  \underset{(p,p',V)}{\sup} \chi \left(q_p, \hat{V}_{E,p'} \otimes \tilde{V}_{E'} \right) &\leq \underset{\rho \in \fI_p}{\sup} \ I(U,EE', \hat{\rho})  + 6\Delta \log|\cU| + 2 h(\Delta). \label{prop:cqq_markov_bound_2}
	  \end{align}
	  Combination of (\ref{prop:cqq_markov_bound_1}) and (\ref{prop:cqq_markov_bound_2}) for each $p \in \cP$ ensures us that 
	  \begin{align*}
	   &\underset{(p,p',V) \in \cP^2 \times \cV}{\inf} \chi(q_p, \hat{V}_{B,p'}\otimes \tilde{V}_{B'}) - \underset{(p,p',V) \in \cP^2 \times \cV}{\sup} \chi(q_p, \hat{V}_{E,p'} \otimes \tilde{V}_{E'}) \\
	   &\geq \underset{p \in \cP}{\inf}\left(\underset{\rho \in \fI_p}{\inf} \ I(U,BB', \hat{\rho}) - \underset{\rho \in \fI_p}{\sup} \ I(U,EE', \hat{\rho})\right)  - 12\Delta \log|\cU| - 4 h(\Delta)
	  \end{align*}
	  holds. Note that the identities 
	    \begin{align*}
	     S(\tilde{\rho}_{BT}) = S(\hat{\rho}_{BB'}), \  \text{and} \hspace{.4cm} S(\tilde{\rho}_{ET}) = S(\hat{\rho}_{EE'}),
	    \end{align*}
	  are valid. Moreover, for each $\rho \in \fI$ the equalities
	    \begin{align}
	     I(U;B|T, \tilde{\rho}) &= H(P_{UT}) + S(\tilde{\rho}_{BT}) - S(\tilde{\rho}_{UBT}) - H(P_T), \hspace{.2cm} \text{and} \\
	     I(U;E|T, \tilde{\rho}) &= H(P_{UT}) + S(\tilde{\rho}_{ET}) - S(\tilde{\rho}_{UET}) - H(P_T), 
	    \end{align}
	   hold by definition, where $P_T$, $P_{TU}$ are the distributions for the random variables $T$ and $TU$. It is important, to notice here that the distributions $P_U$ and $P_{TU}$ depend only on the Markov chain on the sender's systems.
	  Therefore, we have for each $p \in \cP$
	  \begin{align}
	   & \underset{\rho \in \fI_p}{\inf} I(U;B|T, \tilde{\rho}) - \underset{\rho \in \fI_p}{\sup} I(U;E|T, \tilde{\rho}) \nonumber\\
	   &= \underset{\rho \in \fI_p}{\inf}\left(S(\tilde{\rho}_{BT}) - S(\tilde{\rho}_{UBT}) \right) - \underset{\rho \in \fI_p}{\sup}\left(S(\tilde{\rho}_{ET}) - S(\tilde{\rho}_{UET}) \right) \nonumber \\
	   &= \underset{\rho \in \fI_p}{\inf}\left(S(\tilde{\rho}_{BB'}) - S(\tilde{\rho}_{UBB'}) \right) - \underset{\rho \in \fI_p}{\sup}\left(S(\tilde{\rho}_{EE'}) - S(\tilde{\rho}_{UEE'}) \right) \nonumber \\
	   &= \underset{\rho \in \fI_p}{\inf}\left(H(q_p) + S(\hat{\rho}_{BB'}) - S(\hat{\rho}_{UBB'}) \right) - \underset{\rho \in \fI_p}{\sup}\left(H(q_p) + S(\hat{\rho}_{EE'}) - S(\hat{\rho}_{UEE'}) \right) \nonumber \\
	   &= \underset{\rho \in \fI_p}{\inf} I(U;BB',\hat{\rho}) - \underset{\rho \in \fI_p}{\sup} I(U;EE',\hat{\rho}). \label{cqq_markov_last_ineq}
	  \end{align}
	  Collecting the bounds obtained, we can prove the desired lower bound on the key rate. It holds
	  \begin{align*}
	   \frac{1}{n} \log M 
	   &\geq \underset{(p,p',V) \in \cP^2 \times \cV}{\inf} \chi(q_p, \hat{V}_{Bp'} \otimes \tilde{V}_{B'}) - \underset{(p,p',V) \in \cP^2 \times \cV}{\sup} \chi(q_p, \hat{V}_{E,p'}\otimes \tilde{V}_{E'}) -\delta \\
	   &= \underset{p \in \cP}{\inf} \left(\underset{\rho \in \fI_p}{\inf} I(U;BB',\hat{\rho}) - \underset{\rho \in \fI_p}{\sup}\ I(U;EE',\hat{\rho}) \right) - \delta - 12\Delta \log|\cU| - 4 h(\Delta) \\
	   &= \underset{p \in \cP}{\inf} \left(\underset{\rho \in \fI_p}{\inf} I(U;B|T,\tilde{\rho}) - \underset{\rho \in \fI_p}{\sup}\ I(U;E|T,\tilde{\rho}) \right) - \delta - 12\Delta \log|\cU| - 4 h(\Delta).
	  \end{align*}
	  The first inequality above is by (\ref{prop:cqq_markov_initial_rate}), the first inequality is the one from (\ref{prop:cqq_markov_bound_1}), while the last inequality is by (\ref{cqq_markov_last_ineq}). We are done.
	  \end{proof}
	  
	  \begin{proposition}\label{prop:nearlyend}
	  Let $\Delta > 0$, and $ \fJ \subset \cS_{cqq}(\cY,  \cK_{BE})$ be a $\Delta$-regular set of cqq density matrices on $\cH_{ABE}$. For all $z, z'  \in \bbmN$, it holds
	  \begin{align*}
	  K_{\rightarrow}( \fJ) \geq \tilde{K}_\rightarrow^{(1)}(\fJ, z, z)  - \delta- f_{reg}(z,\Delta), 
	  \end{align*}
	  with a function $f_{reg}:\ \bbmN \times \bbmR^+ \rightarrow \bbmR^+$ such that $f(r,\Delta) \rightarrow 0 \ (\Delta \rightarrow 0)$. For a set $\fA := \{\sum_{y \in \cY} p(y) \ \ket{y} \bra{y} \otimes \sigma_y\}$ on some space, and $z,z' \in \bbmN$,
	  the function $\tilde{K}_\rightarrow^{(1)}(\fA, z,z')$ is defined
	  \begin{align*}
          \tilde{K}_{\rightarrow}^{(1)}(\fA) := \underset{p \in \cP_\fA}{\inf}\ \underset{\Gamma := T \leftarrow U \leftarrow Y_p}{\sup} \left(\underset{\sigma \in \fA_p}{\inf} I(U;B|T, \sigma_{\Gamma}) 
          - \underset{\sigma \in \fA_p}{\sup} I(U;E|T, \sigma_{\Gamma})\right).
   \end{align*}
   The supremum above is over all Markov chains $T \leftarrow U \leftarrow Y_p$ resulting from application of Markov transition matrices $P_{T|U}: \cU \rightarrow \cT$, $P_{U|Y}: \cY \rightarrow \cU$ on $p$ for each $p \in p$ with $|\cU| = z, \ |\cT| = z'$,
   and 
   \begin{align*}
    \sigma_{TU} := \sum_{y \in \cY} \sum_{t \in \cT} \sum_{u \in \cU} P_{T|U}(t|u) P_{U|Y}(u|y) p(y) \ \ket{t} \bra{t} \otimes \ket{u} \bra{u} \otimes \sigma_y  
   \end{align*}
   for given transition matrices $P_{T|U}$, $P_{U|Y}$ and 
   \begin{align}
    \sigma = \sum_{y \in \cY} p(y) \ \ket{y} \bra{y} \otimes \sigma_{y}.
   \end{align}
   \end{proposition}
   \begin{proof}
 	  Assume the set $ \fJ$ is parameterized such that $\cP \subset \fP(\cY)$ is the set of possible marginal distributions on the sender's system, while to each $p \in \cP$ a set $\cV_p \subset \cq(\cY, \cK_{BE})$ is associated, i.e.
	  \begin{align*}
	    \fJ = \left\{\rho := \sum_{y \in \cY} q(y) \ket{y}\bra{y} \otimes V(y):\ q \in \cP, V \in \cV_q\right\}
	  \end{align*}
	  Let $\delta > 0$, $z,z' \in \bbmN$ be arbitrary but fixed numbers.  We show that 
	  \begin{align*}    
	    \tilde{K}^{(1)}( \fJ, z,z') - \delta - f_{reg}(z,\Delta)
	  \end{align*}
	   with $f_{reg}(z,\Delta)$ being defined
	   \begin{align*}
	    f_{reg}(z,\Delta) := 32 \Delta \log(z \cdot \dim \cK_{BE}) + 24 h(\Delta)
	   \end{align*}
           is an achievable forward secret-key distillation rate for $\fJ$. Note that the function defined above indeed has the properties claimed above. The strategy of proof will be as follows. We will equip $\cP$ with a 
           regular non-intersecting covering, were we utilize the set of types for large enough blocklength 
	   to define such. With the right choice of parameters, we obtain a finite family of sources which approximate $\fJ$ and have the addition property for fulfilling the hypotheses of Proposition \ref{prop:cqq_markov}. Combining the 
	   protocols obtained for each member of the family with an estimation protocol on the first $\sqrt{n}$ letters for blocklength $n$ leads us to a universal protocol for $\fI$. \\ 
	   We begin setting up the covering of $\cP$. Define for each $k,l \in \bbmN$, $\lambda \in \fT(k, \cY)$ a set
	   \begin{align*}
	    \fT_{\lambda, l} := \left\{q \in \fP(\cY): \ \forall y \in \cY: \lambda(y)- \tfrac{l}{2k} < q(y) \leq \lambda(y) + \tfrac{l}{2k} \right\}. 
	   \end{align*}
	   Notice that the diameter of $\fT_{\lambda, l}$ is bounded by
	   \begin{align*}
	    \diam(\fT_{\lambda,l}) \leq \frac{l\cdot |\cY|}{k},
	   \end{align*}
	   and the sets in the family being pairwise non-intersecting for $l =1$. We fix $k$ to be specified later, define $\cP_\lambda := \fT_{\lambda,3} \cap \cP$ for each $\lambda \in \fT(k,\cY)$, and denote by $\hat{\fT}$ the collection of all $\lambda$ with $\cP_\lambda$ being nonempty. We construct 
	   sets of cqq density matrices, which fit the specifications demanded in Proposition \ref{prop:cqq_markov}, we define 
	   \begin{align*}
	     \hat{\cV}_\lambda 	&:= \bigcup_{q \in \cP_\lambda} \ \cV_q, \hspace{.2cm} \text{and} \hspace{.6cm}  
	     \hat{ \fJ}_{p,\lambda} 	:= \left\{\sum_{y \in \cY} p(y) \ket{y} \bra{y} \otimes V(y)\right\}_{V \in \tilde{\cV}_\lambda }
	    \end{align*}
	  for each $\lambda \in \hat{\fT}$. Fix the number $k$ large enough, to ensure us that for each $\lambda \in \tilde{\fT}$, $p,p'$ are in $\cP_\lambda$ implies
	  \begin{align*}
	   d_H( \fJ^{AB}_p, \fJ^{AB}_{p'}) + d_H( \fJ^{AE}_p, \fJ^{AE}_{p'}) \leq \Delta,
	  \end{align*}
	   and, in addition $\diam(\cP_\lambda) \leq \Delta$. Notice that this choice of $k$ is indeed possible because we assumed $\fJ$ to be $\Delta$-regular. We consider the family $\{\hat{\fJ}_\lambda\}_{\lambda\in \hat{\fT}}$, where  
	   for each $\lambda$, $\hat{\fJ}_\lambda$ is the set of density matrices defined by
	   \begin{align*} 
	    \hat{\fJ}_{\lambda} := \bigcup_{p \in \cP_\lambda} \fJ_{p,\lambda} = \left\{\sum_{y \in \cY} p(y) \ket{y} \bra{y} \otimes V(y): \ p \in \cP_\lambda , V \in \hat{\cV}_\lambda \right\}.
	   \end{align*}
	   The rightmost of the above equalities holds by construction. Notice that since $\cP_\lambda$ has diameter bounded, and $\fJ_\lambda$ is parameterized by the full Cartesian product $\cP_\lambda \times \hat{\cV}_\lambda$, Proposition 
	   \ref{prop:cqq_markov} can be applied in each case. Choose for each $\lambda \in \hat{\fT}$, stochastic matrices 
	   $P_{T|U,\lambda} : \cU \rightarrow \fP(\cT)$ and $P_{U|Y, \lambda}: \cY \rightarrow \fP(\cU)$ such that
	  \begin{align}
	 \underset{\rho \in  \fJ_{p}}{\inf} I(U_\lambda;B|T_\lambda, \tilde{\rho}) -  \underset{\rho \in  \fJ_{p}}{\sup} I(U_\lambda;E|T_\lambda, \tilde{\rho})   
	   \geq \underset{T \leftarrow U \leftarrow Y}{\sup} \left( \underset{\rho \in  \fJ_{p}}{\inf} I(U;B|T, \tilde{\rho}) -  \underset{\rho \in  \fJ_{p}}{\sup} I(U;E|T, \tilde{\rho}) \right) - \frac{\delta}{2} \label{rate_subopt_bound_0}
	  \end{align}
	  is fulfilled. Resulting from the choices made, it also holds 
	   \begin{align}
	    &\underset{\rho \in \hat{ \fJ}_{p,\lambda}}{\inf} I(U_\lambda;B|T_\lambda,\tilde{\rho}) - \underset{\rho \in \hat{ \fJ}_{p,\lambda}}{\sup} I(U_\lambda;E|T_\lambda,\tilde{\rho}) \nonumber \\
	    &\geq  
	    \underset{\rho \in  \fJ_p}{\inf} I(U_\lambda;B|T_\lambda,\tilde{\rho}) - \underset{\rho \in  \fJ_p}{\sup} I(U_\lambda;E|T_\lambda,\tilde{\rho}) - f_{reg}(\Delta,z,z') \label{rate_subopt_bound}
	  \end{align}
	  for each $\lambda \in \tilde{\fT}, p \in \cP_\lambda$. The inequality above is by continuity together with properties of our construction and definition of $f_{reg}$. The full argument for justification can be found in 
	  Appendix \ref{ineq_appendix}. Combining the above estimates, we have for each $\lambda \in \hat{\fT}$
	  \begin{align}
	   &\underset{p \in \cP_{\lambda}}{\inf} \left( \underset{\rho \in  \fJ_{p,\lambda}}{\inf} I(U_\lambda;B|T_\lambda, \tilde{\rho}) -  \underset{\rho \in  \fJ_{p,\lambda}}{\sup} I(U_\lambda;E|T_\lambda, \tilde{\rho}) \right)  \nonumber \\
	   & \geq \underset{p \in \cP_{\lambda}}{\inf} \left( \underset{\rho \in  \fJ_{p}}{\inf} I(U_\lambda;B|T_\lambda, \tilde{\rho}) -  \underset{\rho \in  \fJ_{p}}{\sup} I(U_\lambda;E|T_\lambda, \tilde{\rho}) \right) - 
	   f_{reg}(\Delta,z,z') \nonumber \\
	   & \geq \underset{p \in \cP_{\lambda}}{\inf} \underset{T \leftarrow U \leftarrow Y}{\sup} 
	   \left( \underset{\rho \in  \fJ_{p}}{\inf} I(U;B|T, \tilde{\rho}) -  \underset{\rho \in  \fJ_{p}}{\sup} I(U;E|T, \tilde{\rho}) \right) -  \frac{\delta}{2} - \frac{1}{2}f_{reg}(\Delta,z,z')  \nonumber \\
	   &\geq \underset{p \in \cP}{\inf} \underset{T \leftarrow U \leftarrow Y}{\sup} 
	   \left( \underset{\rho \in  \fJ_{p}}{\inf} I(U;B|T, \tilde{\rho}) -  \underset{\rho \in  \fJ_{p}}{\sup} I(U;E|T, \tilde{\rho}) \right) -  \frac{\delta}{2} -  \frac{1}{2}f_{reg}(\Delta,z,z')  \nonumber \\
	   & = \tilde{K}^{(1)}_{\rightarrow}(\fJ,z,z') - \frac{\delta}{2} -  \frac{1}{2} f_{reg}(\Delta,z,z') .  \label{nearlyend_ratebound}
	   \end{align}	
	   fulfilled. The first inequality above holds by (\ref{rate_subopt_bound}), the second is by (\ref{rate_subopt_bound_0}). Let the blocklength $n \in \bbmN$ be fixed. We set $n = a_n + b_n$ with
	   $a_n := \lceil \sqrt{n} \rceil$, $b_n := n - a_n$, and consider the decomposition
	   \begin{align*}
	    \cY^n = \cY^{a_n} \times \cY^{b_n}.  
	   \end{align*} 
	   Applying Proposition \ref{prop:cqq_markov}, to each of the sets $\hat{\fJ}_\lambda$, we infer for each large enough $n$, $\lambda \in \hat{\fT}$ existence of an $(b_n,M,L,\hat{\vartheta})$ secret-key distillation protocol 
	   $(\hat{T}_\lambda, \hat{D}_\lambda)$ for $\hat{\fJ}_\lambda$ with 
	   \begin{align*}
	    \hat{\vartheta} \leq 2^{-\sqrt[16]{b_n} c_\lambda} \leq 2^{\sqrt[16]{b_n}c}
	   \end{align*}
	   with a strictly positive constant $c_\lambda$ and $c := \underset{\lambda \in \tilde{\fT}}{\min} c_\lambda$ and 
	   \begin{align}
	    M = \left\lfloor \exp \left( b_n\left( \tilde{K}_\rightarrow^{(1)}(\fJ)- \frac{3\delta}{4} -  f_{reg}(\Delta,z,z')  \right) \right) \right\rfloor. \label{nearlyend_logm}
	   \end{align}
	   Note that the combination of $M$ and $\theta$ is indeed possible. This is justified by combining the claim of Proposition \ref{prop:cqq_markov} and the bound in (\ref{nearlyend_ratebound}). Next, we define a two-phase protocol,
	   where the first $a_n$ letters from the source observed by the sender are used to estimate $\lambda \in \hat{fT}$, while the protocol $(\hat{D}_\lambda, \hat{T}_\lambda)$ for the estimated parameter $\lambda$ is applied on the remaining
	   $b_n$ outputs of the source. Tor formalize this strategy, we define a stochastic matrix
	   \begin{align*}
	    T :  \cY^n \rightarrow \fP([L] \times \fT(k,\cY) \times [M])
	   \end{align*}
	    with entries \begin{align*}
	     T(l,\theta,m|y^n) := \hat{T}_{\theta}(l,m|y^{b_n})\ \delta_{\theta \xi(y^{a_n})} &&(l,m, \theta, y^n) \in [L]\times [M] \times \hat{\fT} \times \cY^n) 
	    \end{align*}
	    for each $\mu \in \hat{\fT}$, were we defined a function $\xi: \cY^{a_n} \rightarrow \fT(k,\cY)$ which maps each $y^{a_n}$ to the unique member $\lambda = \xi(y^{a_n})$ such that $\fT_{\lambda,1}$ contains the type of $y^{a_n}$. Notice,
	    that some of the entries are undefined, if $\hat{\fT}$ does not contain all elements of $\fT(k,\cY)$. In this case, entries can be defined in any consistent way, because they will be of no further relevance. 
	    Moreover, we 
	    introduce matrices 
	    \begin{align*}
	     D_{lm\theta} := \bbmeins_{\cH_B}^{\otimes {a_n}} \otimes D_{lm}^{\theta} \in \cL(\cH_{B}^{\otimes n}),
	    \end{align*}
	    where $D_{lm}^\theta$ is the corresponding effect from the POVM $D_\theta$ associated to $\theta$. With these definitions, it is clear that $(T,D)$ is an $(n,M,L\cdot |\hat{\fT}|,\vartheta)$ forward secret-key distillation protocol for $\fJ$, 
	    with a 
	    number $\vartheta$ which we
	    lower-bound. Let 
	    \begin{align*}
	    \rho := \sum_{y \in \cY} p(y) \ket{y} \bra{y} \otimes V(y) 
	    \end{align*}
	    be any fixed member of $ \fJ$, and $\lambda_0$ the unique type in $\hat{\fT}$ such that $p \in T_{\lambda_0,1}$. It is important to notice that not only for $\lambda_0$, but also for each 
	    $\theta \in \hat{\fT}$ with $\theta \in T_{\lambda_0,3}$, $\rho$ is also a member of $\hat{\fJ}_\theta$. Assuming application of the protocol to $\rho$, we suppress indicating the chosen member in the following formulas. By definition, 
	    it holds
	    \begin{align*}
	     P_{KK'\Lambda \Theta|Y^n}(m,m',l,\theta|y^n) 
	     & = \hat{T}_{\theta}(m,l|y^{b_n}) \cdot \delta_{\theta, \xi(y^{a_n})} \cdot \tr(D^{\theta}_{lm'} V^{\otimes b_n}(y^{b_n})) \\
	     & = P^\theta_{\hat{K}\hat{K}'\hat{\Lambda}|Y^{b_n}}(m,m',l,\theta|y^{b_n})\cdot \delta_{\theta, \xi(y^{a_n})} , 
	    \end{align*}
	    with $P^\theta_{\hat{K}\hat{K}'\hat{\Lambda}|Y^{b_n}}(m,m',l,\theta|y^{b_n})$ being the conditional distribution generated by $(\hat{T}_\theta, \hat{D}_\theta)$. We define the sets 
	    \begin{align*}
 	     \iota_1 := \{y^{a_n}: \ \xi(y^{a_n}) \in \fT_{\lambda,1}\}, \hspace{.3cm} \text{and} \hspace{.3cm} \iota_3 := \{y^{a_n}: \ \xi(y^{a_n}) \in \fT_{\lambda,3}\}. &&(\lambda \in \fT(k,\cY))
	    \end{align*}
            It holds
	    \begin{align*}
	     P_{KK'}(m,m') 
	     &= \sum_{\theta \in \fT(k,\cY)} \sum_{y^n \in \cY^n} \sum_{l=1}^L \ P_{KK'\Lambda\Theta|\cY^n}(m,m',l,\theta|y^n) p^n(y^n) \\
	     &= \sum_{\theta \in \fT(k,\cY)} \sum_{y^{a_n} \in \fT_{\theta,1}} p^{a_n}(y^{a_n}) \sum_{y^{b_n} \in \cY^{b_n}} \sum_{l=1}^L \ P^{\theta}_{\hat{K}\hat{K}'\hat{\Lambda}|Y^{b_n}}(m,m',l|y^{b_n}) p^{b_n}(y^{b_n}) \\
	     &= \sum_{\theta \in \fT(k,\cY)} p^{a_n}(\iota_{1,\theta}) \ P^{\theta}_{KK'}(m,m')
	    \end{align*}
	    for each $m,m' \in [M]$. We denote the key random variables produced by performing $(\hat{T}_\theta, \hat{D}_\theta)$ on $\rho^{\otimes b_n}$ by $\hat{K}_\theta$ and $\hat{K}'_\theta$. We directly obtain
	    \begin{align}
	     \prob(K \neq K') 
	     &= \sum_{\theta \in \fT(k,\cY)} p^{a_n}(\iota_{\theta,1}) \cdot \prob(\hat{K}_\theta \neq \hat{K}'_{\theta}) \nonumber \\
	     &\leq p^{a_n}(\iota_{\lambda_0,3}) \cdot \hat{\vartheta} +  p^{a_n}(\iota_{\lambda_0,3}^c) \nonumber \\
	     &\leq 2^{\sqrt[16]{b_n}c} +  2^{-a_n \tfrac{c}{k^2}} \nonumber \\
	     &\leq 2 \hat{\vartheta} \label{nearlyend_performance_2}
	    \end{align}
	    The first inequality above is by the fact that the protocol associated to each $\theta \in \hat{\fT} \cap \fT_{\lambda_0,3}$ is $\hat{\vartheta}$-good for $\rho$ by construction. The second inequality is by standard type bounds. Explicitly, 
	    we have by construction 
	    $y^{a_n} \in \iota_{\lambda_0,3}^c$  implying
	    \begin{align*}
	    \left|\frac{1}{a_n} N(e|y^{a_n}) - p(e)\right| > \frac{1}{k} 
	    \end{align*}
	    for all $e \in \cY$, where $N(e|y^{a_n})$ is the number of occurrences of the letter $e$ in $y^{a_n}$. Consequently 
	    \begin{align}
	     p^{a_n}(\iota_{\lambda_0,3}^c) \leq p^{a_n}\left(\left(T_{p, \frac{1}{k}}^{a_n}\right)^c\right) \leq 2^{-a_n \frac{c}{k^2}} \leq \hat{\vartheta},\label{nearlyend_vartheta}
	    \end{align}
	    where $c$ is a universal, strictly positive constant, and the last inequality holds for large enough choice of $n$. Also, it holds
	    \begin{align}
	     H(K) - I(K;\Lambda \Theta E^n, \rho_{K\Lambda \Theta E^n}) 
	     &= H(K) - I(K; \Theta) - I(K; \Lambda E^n|\Theta, \rho_{K\Lambda \Theta E^n})  \nonumber \\
	     &= H(K|\Theta) - I(K; \Lambda E^n|\Theta, \rho_{K\Lambda \Theta E^n}) \nonumber \\
	     &= \sum_{\theta \in \fT} p^{a_n}(\iota_{\theta,1}) \ \left(H(K|\Theta= \theta) - I(K; \Lambda E^n|\Theta = \theta, \rho_{K\Lambda \Theta E^n}) \right) \nonumber \\
	     &= \sum_{\theta \in \fT} p^{a_n}(\iota_{\theta,1}) \ \left(H(\hat{K}^{\theta}) - I(\hat{K}^\theta; \Lambda E^{b_n}, \rho^{\theta}_{\hat{K}\hat{\Lambda} E^{b_n}}) \right), \label{nearlyend_security}
	    \end{align}
	     where the first equality is the chain rule for the quantum mutual information applied, the second holds by definition of the classical mutual information. The third equality results from the fact that if 
	     $\theta$ is the estimate obtained in the first $a_n$ outputs of the source, $(\hat{T}_\theta, \hat{D}_\theta)$ is performed on the remaining $b_n$ outputs, which determines the conditional quantities as generated from 
	     application of the protocol. Therefore, we obtain
	     \begin{align}
	      &\log M -  H(K) + I(K;\Lambda \Theta E^n, \rho_{K\Lambda \Theta E^n}) \nonumber \\
	      &= \sum_{\theta \in \fT} p^{a_n}(\iota_{\theta,1}) \ \left(\log M - H(\hat{K}^{\theta}) + I(\hat{K}^{\theta}; \hat{\Lambda}^{\theta} E^{b_n}, \rho^{\theta}_{\hat{K}\hat{\Lambda} E^{b_n}}) \right) \nonumber \\
	      &\leq p^{a_n}(\iota_{\lambda_0,3}) \cdot \hat{\vartheta} +  p^{a_n}(\iota_{\lambda_0,3}) \cdot (2 \cdot\log M + \log L + b_n \cdot \log \dim \cK_{BE}) \nonumber \\
	      &\leq 2 \hat{\vartheta} \label{nearlyend_performance_1}
	     \end{align}
	      where the equality above follows from (\ref{nearlyend_security}), the first inequality is by the fact that the protocol $(\hat{T}_\theta, \hat{D}_\theta)$ is $\hat{\vartheta}$-good for $\rho$ whenever 
	      $\theta$ is a direct grid point neighbour of $\lambda_0$, i.e. $\fT_{\theta,1} \subset \fT_{\lambda_0,3}$. Moreover we applied the ultimate bound $I(A;B,\rho) \leq 2 \log \dim \cH_A \otimes \cH_B$ which 
	      holds for each state $\sigma$ on any Hilbert space $\cH_A \otimes \cH_B$.  The last inequality holds with a large enough choice of $n$ by application of the bound in (\ref{nearlyend_vartheta}). The bounds 
	      obtained in (\ref{nearlyend_performance_1}) and (\ref{nearlyend_performance_2} show us that $(T,D)$ is actually an $(n,M,L,\vartheta)$ forward secret-key distillation protocol for $\fJ$, with $\vartheta \leq 2 \hat{\vartheta}$,
	      and, since $b_n/n \rightarrow 1$ for $n \rightarrow \infty$, it holds 
	      \begin{align*}
	       \frac{1}{n} \log M 
	       &\geq \frac{b_n}{n}  \tilde{K}_\rightarrow^{(1)}(\fJ)- \frac{3\delta}{4} -  f_{reg}(\Delta,z,z')  \\
	       &\geq \tilde{K}_\rightarrow^{(1)}(\fJ)- \delta -  f_{reg}(\Delta,z,z')			    
	      \end{align*}
              if $n$ is large enough, where the first inequality is  from (\ref{nearlyend_logm}).
	\end{proof}
	To prove achievability of the multi-letter formula claimed in Theorem \ref{secret-key_generation-theorem-regular}, we have to ensure ourselves that regularity conditions do not break down when considering the set 
	$\fI^{\otimes n} := \{\rho^{\otimes n}: \ \rho \in \fI \}$ instead of a set $\fI$ of cqq density matrices. The following two basic lemmas will turn out to be sufficient for our needs. 
	\begin{lemma} \label{hausdorff_lemma_2}
	Let $\fI,\fJ \subset \cL(\cK)$ be any two sets of density matrices. It holds for each $n \in \bbmN$
	\begin{align*}
	  d_H(\fI^{\otimes n}, \fJ^{\otimes n}) \leq n \cdot d_H(\fI, \fJ),
	\end{align*}
	where $d_H$ is the Hausdorff distance induced by the trace norm on the underlying space.
	\end{lemma}
	\begin{proof}
	The inequality 
	\begin{align*}
	  \|a^{\otimes n} - b^{\otimes n}\|_1 \leq n \cdot \|a - b\|_1
	\end{align*}
	valid for any two matrices $a, b \in \cL(\cK)$ inherits to the Hausdorff distance. It holds
	\begin{align*}
	\underset{a \in \fI}{\sup}\ \underset{b \in \fJ}{\inf} \ \|a^{\otimes n} - b^{\otimes n}\|_1 \leq n \cdot  \underset{a \in \fI}{\sup}\ \underset{b \in \fJ}{\inf} \ \|a - b\|_1.
	\end{align*}
	\end{proof}
	\begin{lemma} \label{extension_regularity_lemma}
	Let $\fI$ be a set of cqq density matrices. It holds
	\begin{align*}
	  \fI \ \epsilon-\text{regular} \ \Rightarrow \ \fI^{\otimes k} \ k \cdot \epsilon - \text{regular}.
	\end{align*}
	for each $k \in \bbmN$.
	\end{lemma}
	\begin{proof}
	Is by direct application of Lemma \ref{hausdorff_lemma_2} and the definition of regularity.
	\end{proof}
	We now obtained sufficient preparations to tackle the proof of achievability in Theorem \ref{secret-key_generation-theorem-regular}. 
	Before we head to the proof, we ensure ourselves that the limit in (\ref{secret-key_generation_theorem-regular_formula}) indeed exists. 
	\begin{lemma}\label{lemma:fekete}
	Let $\fI$ be a set of cqq density matrices on $\cH_{ABE}$. It holds
	\begin{align*}
	 \underset{k \in \bbmN}{\sup} \frac{1}{k} K^{(1)}(\fI^{\otimes k}) = \lim_{k \rightarrow \infty} \frac{1}{k}K^{(1)}(\fI^{\otimes k}).
	\end{align*}
	\end{lemma}
	\begin{proof}
	 The assertion of the lemma follows from application of Fekete's lemma \cite{fekete23} on the sequence $K^{(1)}(\fI^{\otimes k})$. We check that the hypotheses of Fekete's lemma are fulfilled. 
	 Clearly, the sequence is bounded. We show that it is also superadditive, i.e. $K^{(1)}(\fI^{\otimes (k+l)}) \geq K^{(1)}(\fI^{\otimes k}) + K^{(1)}(\fI^{\otimes l})$ being valid for all 
	$k,l \in \bbmN$. We can for each $k$ write $K^{(1)}(\fI^{\otimes k})$ in the form
	\begin{align*}
	  K^{(1)}(\fI^{\otimes k})\ = \ \underset{p \in \cP_\fI}{\inf} \ \underset{z,z' \in \bbmN}{\sup} \ \hat{K}^{(1)}(\fI^{\otimes k},p,z,z')
	\end{align*}
	were we defined 
	\begin{align*}
	 \hat{K}^{(1)}(\fI^{\otimes k}, p, z,z') := \underset{T \leftarrow U \leftarrow X_p}{\sup} \left(\underset{\sigma \in \fI_p}{\inf} I(U;B|T,\tilde{\sigma}) - \underset{\sigma \in \fI_p}{\sup} I(U;E|T,\tilde{\sigma}) \right),
	\end{align*}
	with the outer maximization above being over all Markov chains generated by transition matrices $P_{U|X}:\ \cX \rightarrow \fP(\cX)$ and  $P_{T|U}: \ \cU \rightarrow \fP(\cT)$ with alphabets of cardinalities $|\cU| = z$, $|\cT| = z'$,
	and 
	\begin{align*}
	 \tilde{\sigma} := \sum_{t \in \cT} \sum_{u \in \cU} \sum_{x \in \cX} \ P_{T|U}(t|u) \ P_{U|X}(u|x) \ p(x) \ket{u}\bra{u} \ket{t} \bra{t} \otimes V(x)
	\end{align*}
	for
	\begin{align*}
	 \sigma = \sum_{x \in \cX} \ p(x) \ket{x} \bra{x} \otimes V(x). 
	\end{align*}
	Notice that for each $p \in \cP_{\fI}$, $z,z' \in \bbmN$
	\begin{align*}
	 \hat{K}^{(1)}(\fI^{\otimes (k+l)}, p , z ,z') \geq \hat{K}^{(1)}(\fI^{\otimes k}, p, z, z') + \hat{K}^{(1)}(\fI^{\otimes l}, p, z, z'),
	\end{align*}
	and moreover, for each $z_1 \leq z_2$, $z_1' \leq z_2'$,
	\begin{align*}
	 \hat{K}^{(1)}(\fI^{\otimes k}, p, z_2, z_2') \geq  \hat{K}^{(1)}(\fI^{\otimes k}, p, z_1, z_1')
	\end{align*}
	holds for each $k \in \bbmN, \ p \in \cP_{\fI}$. We obtain
	\begin{align*}
	 \hat{K}^{(1)}(\fI^{\otimes (k+l)}, p, z_2, z_2') 
	 &\geq \hat{K}^{(1)}(\fI^{\otimes k}, p, z_2, z_2') + \hat{K}^{(1)}(\fI^{\otimes l}, p, z_2, z_2') \\
	 &\geq \hat{K}^{(1)}(\fI^{\otimes k}, p, z_2, z_2') + \hat{K}^{(1)}(\fI^{\otimes l}, p, z_1, z_1') 
	\end{align*}
	Consequently, it holds
	\begin{align*}
	\underset{z,z' \in \bbmN}{\sup} \hat{K}^{(1)}(\fI^{\otimes (k+l)}, p, z, z') 
	\geq \underset{z,z' \in \bbmN}{\sup} \hat{K}^{(1)}(\fI^{\otimes k}, p, z, z') + \underset{z,z' \in \bbmN}{\sup} \hat{K}^{(1)}(\fI^{\otimes l}, p, z, z') 
	\end{align*}
	for each $p \in \cP{\fI}$. We conclude 
	\begin{align*}
	K^{(1)}(\fI^{\otimes (k+l)}) \geq K^{(1)}(\fI^{\otimes k}) + K^{(1)}(\fI^{\otimes l})
	\end{align*}
	\end{proof}

	\begin{proof}[Proof of Theorem \ref{secret-key_generation-theorem-regular}]
	We first prove achievability, i.e. validity of the inequality
	\begin{align*}
	K_{\rightarrow}(\fI) \geq \lim_{k \rightarrow \infty} \frac{1}{k} K_{\rightarrow}^{(1)}(\fI^{\otimes k})
	\end{align*}
	Let $z.z',k \in \bbmN$ and $\delta>0$ be arbitrary and fixed. We show that 
	\begin{align*}
	 \frac{1}{k}\tilde{K}^{(1)}_{\rightarrow}(\fI^\otimes k) - \delta
	\end{align*}
	 is an achievable forward secret-key distillation rate. We apply Proposition \ref{prop:nearlyend} with $\fJ = \fI^{\otimes k}$, and conclude that for each large enough blocklength $g \in \bbmN$,
	 we find an $(l,M,L,\vartheta)$ forward secret-key distillation protocol for $\fI^{\otimes k}$ with $\vartheta \leq 2^{-\sqrt[16]{g}c_3}$ with a constant $c_3 > 0$, and 
	 \begin{align}
	  \frac{1}{g} \log M \geq \tilde{K}^{(1)}(\fI^{\otimes k},z,z') - \frac{2}{3} \delta \label{reg_achiev_prf_1}
	 \end{align}
	 where we chose $\Delta$ small enough to satisfy $f_{reg}(z,z', \Delta) \leq \tfrac{\delta}{3}$. Since an $(g, K, M ,\vartheta)$ protocol for $\fI^{\otimes k}$ is obviously an $(g\cdot k, M, L, \vartheta)$ protocol
	 for $\fI$, we obtained sufficient protocols for all large enough blocklengths being integer multiples of $k$. We can achieve sufficient protocols also for the remaining blocklengths just by wasting resources. To be explicit,
	 let $n = k \cdot g  + r$ with $0 < r < k$ and assume $(\hat{T}_{gk}, \hat{D}_{gk})$ being an $(g\cdot k, M, L , \mu)$ protocol for $\fI$. Define a protocol $(T_n, D_n)$ for blocklength $n$ by setting 
	 \begin{align*}
	  T_{n}(l,m|x^{n}) = \hat{T}_{gk}(l,m|(x_1,  \dots, x_{g\cdot k})) &&(x^n = (x_1,\dots, x_n) \in \cX^n) 
	 \end{align*}
	 and effects
	 \begin{align*}
	  D_{n,lm} := \hat{D}_{gk,lm} \otimes \bbmeins_{\cH_B}^{\otimes k}
	 \end{align*}
	 for each $l \in [L], m \in [M]$. It is clear that $(T_n, D_n)$ is an $(n,M,L,\mu)$ forward secret-key distillation protocol for $\fI$ with rate
	 \begin{align}
	  \frac{1}{n} \log M = \frac{1}{g \cdot k + r} \log M \geq  \frac{1}{g\cdot k} \log M- \frac{\delta}{3} \label{reg_achiev_prf_2}
	 \end{align}
         if $n$ is large enough. It follows from (\ref{reg_achiev_prf_1}) and (\ref{reg_achiev_prf_2}) that we actually achieve
         \begin{align*}
          \frac{1}{k}\tilde{K}^{(1)}(\fI^{\otimes k},z,z') - \delta
         \end{align*}
	  Since $\delta$, $z,z'$ were arbitrary, we are done. We do not give a detailed argument for the converse inequality here, since the assertion directly follows from 
	  (\ref{with_or_without_knowledge}) together with a converse proof for the case of a source with SMI given in the next section.
	 \end{proof}
	\end{section}
	\begin{section}{Secret-key distillation with sender marginal information (SMI)} \label{section:secret-key_ssi}
	 In this section, we assume that the sender has perfect knowledge of his/her marginal distribution deriving from the source statistics. We will prove the achievability part of Theorem \ref{secret-key_generation-theorem-ssi} by decomposing
	 each compound cqq source into a \emph{finite} collection of regular compound cqq sources. To obtain such an approximation, we need the following basic assertion. For a given set $X$, we use the notation $2^X$ for the power set. 
	 \begin{lemma}\label{hausdorff_finite}
	  Let $d_H$ be the Hausdorff distance on $2^{\bbmR^n}$ generated by the 1-norm distance on $\bbmR^n$. Let $A \subset \bbmR^n$ be a subset of $\bbmR^n$ with $\diam(A) \leq a < \infty$. For each $\Delta > 0$, there exists
	  a family $\cR_A:=\{\tilde{A}_\omega\}_{\omega=1}^\Omega \subset 2^{\bbmR^n} \setminus \ \{\emptyset\}$ with the following properties.
	  \begin{enumerate}
	   \item $\Omega \leq \exp\left(\left(\frac{n \cdot a}{\Delta}\right)^n\right)$.
	   \item For each $B \subset A$ there exists $\omega \in [\Omega]$, such that 
	         \begin{align*}
	          d_H(B,\tilde{A}_\omega) \leq \Delta, \hspace{.2cm} \text{and} \hspace{.2cm} B \subset \tilde{A}_k.
	         \end{align*}
	  \end{enumerate}
	 \end{lemma}
	 \begin{proof}
	  Equip $\bbmR^n$ with the regular pairwise-disjoint covering, generated by the n-dimensional half-open cubes
	  \begin{align*}
	   \left[\left(k_1\frac{\Delta}{n},\dots,k_n\frac{\Delta}{n}\right), \left((k_1+1)\frac{\Delta}{n},\dots,(k_n+1)\frac{\Delta}{n}\right)\right) &&((k_1,\dots,k_n)\in \bbmZ^n).
	  \end{align*}
	 Since $\diam(A) \leq a$ is assumed, we do not need more than $K := \left(\tfrac{n \cdot a}{\Delta}\right)^n$ of these cubes to cover $A$. Let 
	 $
	  \{G_k\}_{k=1}^K
	 $
	 be any parameterization of the family of cubes intersecting with $A$ by $[K] := \{1,\dots K\}$. Define, for each $\omega \subset [K]$
	 \begin{align*} 
	  \tilde{A}_\omega := \underset{k \in \omega}{\bigcup} G_k.
	 \end{align*}
	 We show that
	 \begin{align*}
	  \cR_A := \{\tilde{A}_\omega\}_{\omega=1}^{\Omega}
	 \end{align*}
	 indeed has the properties stated in the lemma. The first property is fulfilled by the bound on $K$ and the fact that there are not more, than $2^K$ different values for $\omega$.
	 The member 
	 \begin{align}
	   \omega := \{k \in [K]: \ G_k \cap A \neq \emptyset\} \label{hausdorff_property_2}
	 \end{align}
	 fulfills the properties demanded for the second property.
	 \end{proof}
	 
	 \begin{proof}[Proof of Theorem \ref{secret-key_generation-theorem-ssi}]
	 For proving achievability, the following strategy is applied. We approximate $\fI$ by a finite family $\{\fI_\omega\}_{\omega \in \Omega}$ and apply Theorem \ref{secret-key_generation-theorem-regular} 
	 for each degree of regularity. 
	 Let $\fI := \{\rho_s\}_{s \in S}$ be a parameterization of $\fI$, and $\{\rho_{A,t}\}_{t \in T}$ be a parameterization of the set of marginal states on $\cH_A$ which derive from members of $\fI$. 
	 Fix an arbitrary $\Delta > 0$ and let $\{\tilde{A}_\omega\}_{\omega \in \Omega}$ be an approximation of $\fI$ with the properties stated in Lemma \ref{hausdorff_finite} with parameter $\lambda$. Note that 
	 by identifying $\bbmC$ to $\bbmR^2$ in the usual way, the approximation satisfies
	 \begin{align*}
	  |\Omega| \leq \exp\left(\left(\frac{4 \dim \cH_{ABE}^2}{\Delta}\right)^{4\dim \cH_{ABE}^2} \right) < \infty,
	 \end{align*}
	 where we only use the fact, the cardinality of $\Omega$ is finite. Let, for each $t \in T$, $\omega(t)$ the element of $\Omega$ as defined in (\ref{hausdorff_property_2}) for $\fI_t$. It holds 
	 \begin{align}
	  \fI_t \subset \tilde{A}_\omega, \hspace{.3cm} \text{and} \hspace{.3cm} d_H(\tilde{A}_\omega, \fI_t) \leq \Delta. \label{hausdorff_approx_property_1}
	 \end{align}
	  Define 
	  \begin{align*}
	   \tilde{\fI}_\alpha := \underset{t:\ \omega(t) = \alpha}{\bigcup} \fI_t &&(\alpha \in \Omega).
	  \end{align*}
	  The family $\{\tilde{\fI}_\alpha\}_{\alpha \in \Omega}$ is decomposition of $\fI$ into a family of pairwise disjoint sets of cqq density matrices with the additional feature that 
	  for each $\alpha \in \Omega$, $\tilde{\fI}_\alpha$ is 
	  $4\Delta$-regular, which can be justified as follows. For each $t, t'$ with $\omega(t) = \omega(t') := \beta$, it holds
	  \begin{align}
	   d_H(\fI_t,\fI_{t'}) \leq d_H(\fI_t, \tilde{A}_\beta) + d_H(\tilde{A}_\beta, \fI_{t'}) \leq 2 \Delta. \label{hausdorff_approx_property_2}
	  \end{align}
	  The left of the above inequalities is the triangle inequality for the Hausdorff distance applied, the right hand inequality is by (\ref{hausdorff_approx_property_1}). Therefore, we infer,
	  using monotonicity of the Hausdorff distance under taking partial traces,
	  \begin{align*}
	   d_H(\fI_t^{AB}, \fI_{t'}^{AB}) + d_H(\fI_t^{AE}, \fI_{t'}^{AE}) \leq 2 \cdot d_H(\fI_t, \fI_{t'}) \leq 4\Delta.
	  \end{align*}
	  From applying Proposition \ref{prop:nearlyend} on each of the sets
	  $\tilde{\fI}_\beta$, $t\in T$, we know that for each given $\delta, \mu >0$, there is a number $k_0(\beta)$, such that we find for each $n > k_0(\beta)$ an $(n, M_\beta, L_\beta, \mu_\beta)$ forward 
	  secret-key distillation protocol $(T^{(\beta)}, D^{(\beta)})$ for $\tilde{\fI}_\beta$, with
	  \begin{align}
	   \log M_\beta 
	   &\geq \tilde{K}^{(1)}(\fI_\beta, z,z') - f_{reg,\beta}(z,\Delta) - \delta  \nonumber \\
	   &\geq \tilde{K}^{(1)}(\fI, z,z') - f_{reg,\beta}(z,\Delta) - \delta \label{m_bound_smi}
	  \end{align}
	  for each $z,z'\in \bbmN$ with a function $f_{reg,\beta}$ as stated in Proposition \ref{prop:nearlyend}. Moreover, we have bounds 
	  \begin{align*}
	   \mu_\beta \leq 2^{-\sqrt[16]{n} c_\beta}, \hspace{.3cm} \text{and} \hspace{.3cm} L_\beta \leq 2^{n R_{c,\beta}}
	  \end{align*}
	  with constants $c_\beta > 0$ and $R_{c,\beta} \in \bbmR^+$ for each $\beta > 0$. We define $c:= \min_{\beta \in \Omega}$, and $c_\beta$, $R_c := \min_{\beta \in \Omega} R_{c,\beta}$, $L = 2^{nR_c}$,
	  $M := \min_{\beta \in \Omega}$. If we define a stochastic matrix $T_t$ with entries
	  \begin{align*}
	   T_t(\beta,l,m|x^n) := T^{(\beta)}(l,m|x^n)  \cdot \delta_{\beta \omega(t)}  &&(\beta \in \Omega, l \in [L], m \in [)
	  \end{align*}
	  and effects
	  \begin{align*}
	   D_{\beta l m} := D^{(\beta)}_{lm} &&(\beta \in \Omega, l \in [L], m \in [M]),
	  \end{align*}
	   Then $((T_t,D))_{t \in T}$ with $D := \{D_{\beta l m} \}_{(\beta,l,m)\in \Omega \times [L] \times [M]}$ is an $(n,M,|\Omega|\cdot L,\mu)$ secret-key distillation protocol for $\fI$ with SMI, such that 
	   \begin{align*}
	    \log M \geq \tilde{K}^{(1)}(\fI, z,z') - f_{reg,\beta}(z,\Delta) - \delta
	   \end{align*}
	  holds by (\ref{m_bound_smi}). Since $\Delta> 0$ was arbitrary, 
	  \begin{align*}
	   \tilde{K}^{(1)}(\fI,z,z') - 2 \delta
	  \end{align*}
	  is achievable for each $z,z' \in \bbmN$. Consequently, it holds 
	  \begin{align*}
	   K_{\rightarrow, SMI} \geq \sup_{z,z' \in \bbmN} \ \tilde{K}^{(1)}(\fI,z,z') - 2 \delta =  \tilde{K}^{(1)}(\fI) - 2 \delta 
	  \end{align*}
	  The same reasoning can be applied for $\fI^{\otimes k}$, $k \in \bbmN$, which implies that
	  \begin{align*}
	   \frac{1}{k} K^{(1)}(\fI^{\otimes k}) 
	  \end{align*}
	  is achievable as well. It remains to prove the converse inequality. Assume $\cP \subset \fP(\cX)$ to be the set of marginal probability distributions on the sender's systems deriving from $\fI$. Define $\cV_p \subset \cq(\cX, \cH_{BE})$ 
	  to be the set of classical-quantum channels associated to each $p \in \cP$. I.e. 
	   Fix $k \in \bbmN$, and assume $(T,D_p)_{p \in \cP}$ to be an $(k,M,L,\mu)$ forward secret-key distillation protocol for the set
	   \begin{align*}
	    \fI_p := \left\{\rho_{p,V} := \sum_{x \in \cX} \ p(x) \ \ket{x}\bra{x} \otimes V(x): \ V \in \cV_p \right\}
	   \end{align*}
	  of density matrices from $\fI$ having sender marginal distribution $p$. Fix any $p \in \cP$ we suppress the index $p$ for the next lines. 
	  Denote by
	  \begin{align*}
	   \rho_{\Lambda K K ' E^n,V} 
	  \end{align*}
          the state resulting from performing $(T_p,D)$ on $\rho_V$ according to (\ref{full_protocol_state}) for each $V  \in \cV_p$. Note that the resulting pair $(\Lambda,K)$ of random variables does not depend on the chosen state $\rho_V$
          since all state in $\fI_p$ have same sender marginal distribution. 
          Since $\log M - H(K)$ and $I(K;\Lambda E^n, \rho_{\Lambda K E^n,V})$ are nonnegative by definition of the protocol and non-negativity of the quantum mutual information, the inequalities 
         \begin{align}
          \log M - H(K) \leq \mu, \hspace{.3cm} \text{and} \hspace{.4cm} \sup_{V \in \cV} I(K;\Lambda E^n, \rho_{\Lambda K E^n, V}) \leq \mu \label{prop:converse_ineqs}
         \end{align}
         are simultaneously fulfilled. Moreover, we have
         \begin{align}
          H(K) 
          &= I(K; K'_V) + H(K|K'_V) \nonumber \\
          &\leq I(K; K'_V) + \mu \log M + h(\mu) \nonumber \\
          &\leq I(K; K'_V \Lambda) + \mu \log M + h(\mu) \nonumber  \\
          &\leq I(K; B^n \Lambda, \rho_{K\Lambda B^n,V}) + \mu \log M + h(\mu) , \label{prop:converse_ineqs_2}
         \end{align}
	 where the first inequality is by Fano's inequality together with the assumption $\prob(K \neq K'_V) \leq \mu$, while the last two inequalities follow from the data processing inequalities for the classical and quantum
	 mutual information. We infer
	 \begin{align}
	  \log M 
	  &\leq H(K) + \mu \nonumber \\
	  &\leq \underset{V \in \cV}{\inf} \ I(K; B^n \Lambda, \rho_{K\Lambda B^n,V}) + \mu + \mu \log M + h(\mu) \nonumber \\
	  &\leq \underset{V \in \cV}{\inf} \ I(K; B^n \Lambda, \rho_{K\Lambda B^n,V}) - \underset{V \in \cV}{\sup} \ I(K; E^n \Lambda, \rho_{K\Lambda E^n,V}) + 2 \mu + \mu \log M + h(\mu) \nonumber \\
	  &\leq \underset{V \in \cV}{\inf} \ I(K; B^n| \Lambda, \rho_{K\Lambda B^n,V}) - \underset{V \in \cV}{\sup} \ I(K; E^n| \Lambda, \rho_{K\Lambda E^n,V}) + 2 \mu + \mu \log M + h(\mu) \nonumber \\
	  &\leq K^{(1)}(\fI_p^{\otimes k}) + 2 \mu + \mu \log M + h(\mu). \label{converse_estimate}
	 \end{align}
         The first and the third of the above inequalities are from (\ref{prop:converse_ineqs}), while the second is from (\ref{prop:converse_ineqs_2}). The fourth is by definition of the quantum mutual information together 
         with the fact that the distribution $(K;\Lambda)$ does not depend on the chosen $V$. The last one results from observing that $X \rightarrow (\Lambda,K) \rightarrow \Lambda$ is a Markov chain. The estimate in 
         (\ref{converse_estimate}) is valid for each $p \in \cP$. Minimization over all $p \in \cP$ leads to 
         \begin{align*}
          \log M 
          &\leq \underset{p \in \cP}{\inf} K^{(1)}(\fI_p^{\otimes k}) + 2 \mu + \mu \log M + h(\mu) \\
	  &= K^{(1)}(\fI^{\otimes k}) + 2 \mu + \mu \log M + h(\mu),	  
         \end{align*}
	 where the equality above is by definition of the function $K_{\rightarrow}^{(1)}$. 
	 \end{proof}
         \end{section}

\begin{section}{Discussion of regularity of compound cqq sources} \label{sect:regularity_condition}
 This section is of twofold purpose. First, we point out that regularity issues have operational significance for forward secret-key distillation from tripartite compound sources. While for regular sources, there is no gap between the forward secret-key 
 key distillation capacities with and without SMI, there may be serious differences in capacities, if the source is not regular. Second, we introduce a weaker notion of regularity than the one introduced in Definition \ref{def:regularity_condition}, where
 we utilize notions from the theory of set-valued functions.
\begin{subsection}{Operational significance of regularity conditions} \label{sect:regularity_condition_counterexample}
 We have seen in the previous section that there is no difference between the forward secret-key distillation capacities with and without SMI, as long
 as the source is regular in the sense of Definition \ref{def:regularity_condition}. We admit that there may be weaker notions of regularity which also exhibit this property (an example of such a condition is introduced in the next section). 
 Regularity conditions seem somewhat technical on a first view. One can easily imagine large classes of sets of cqq density matrices, which are notoriously easy to process even in the case without sender 
 knowledge, while being irregular. This feature is shared in a trivial way by all irregular sources having zero forward secret-key 
 distillation capacity under sender knowledge. The following example depicts the fact that also in nontrivial cases irregularities may be of no consequences for the behaviour of the source regarding forward secret-key distillation. 
 \begin{example}
  Define for a finite alphabet $\cX$, $A := \{p \in \fP(\cX): \ \forall x \in \cX: \ p(x) \in \bbmQ\}$, $V \in \cq(\cX, \cH_{B} \otimes \cH_E)$, and let $\cK_B = \bbmC^{2}$ be the Hilbert space of an additional system assigned to the 
  legitimate receiver. 
  Define states
  \begin{align*}
   \rho_a := \sum_{x \in \cX} a(x) \ket{x}\bra{x} \otimes V_{BE}(x) \otimes \ket{e_a} \bra{e_a}
  \end{align*}
  with $\{e_1,e_2\}$ being an orthonormal basis in $\cK_B$, $e_a := e_1$ if $a \in A$ and $e_a = e_2$ if $a \in A^c$. The source defined by $\fI := \{\rho_a\}_{a \in \fP(\cX)}$ is not regular, but can be easily converted to a regular one by
  just discarding the systems on $\cK_{B}$.
 \end{example}
  Beside the mentioned facts, the question of regularity in principle, bears strong operational significance. The next theorem shows that for irregular sources, the capacities with and without sender marginal
  state knowledge may be substantially different.
 \begin{theorem} \label{counterexample}
  The equality
  \begin{align*}
   K_{\rightarrow,SMI}(\fI) = K_{\rightarrow}(\fI)
  \end{align*}
   does not hold in general.
 \end{theorem}
 \begin{proof}
  We construct an example of a set $\fI$ with 
  \begin{align}
   K_{\rightarrow,SMI}(\fI) = 1 \hspace{.4cm} \text{and} \hspace{.4cm} K_{\rightarrow}(\fI)= 0. \label{counterexample:ineqs}
  \end{align}
  Let $\cX = \cY = \{0,1\}$, and $\cH_{B} = \cH_{E} = \bbmC^2 \otimes \bbmC^2$. We
  introduce classical-quantum channels $W_1, W_2: \ \{0,1\} \rightarrow \cS(\bbmC^2 \otimes \bbmC^2)$ by
  \begin{align*}
   W_1(x,y) &= W_1(x) := \ket{x} \bra{x} \otimes \Pi, \\
   W_2(x,y) &= W_2(y) := \Pi \otimes  \ket{y} \bra{y}  &&((x,y) \in \cX \times \cY),
  \end{align*}
  where $\Pi := \frac{\bbmeins}{2}$ is the flat state on $\bbmC^2$. We set
  \begin{align*}
   V_{1,B} = V_{2,E} = W_1,  \hspace{.3cm} V_{2,B} = V_{1,E} = W_2, 
  \end{align*}
   and define states
  \begin{align*}
   \rho_{p} := 
   \begin{cases} 
   \sum_{x\in \cX} \sum_{y \in \cY} \frac{1}{4} \ket{x \otimes y} \bra{x \otimes y} \otimes V_{1,B}(x) \otimes V_{1,E}(y)  & \hspace{.3cm} \text{if} \ p = \pi \\
   \sum_{x\in \cX} \sum_{y \in \cY} \frac{1}{2}p(y) \ket{x \otimes y} \bra{x \otimes y} \otimes V_{2,B}(y) \otimes V_{2,E}(x) & \hspace{.3cm} \text{otherwise},
   \end{cases}
  \end{align*}
  where $\pi$ denotes the equidistribution on $\{0,1\}$, i.e. $p(0) = p(1) = \tfrac{1}{2}$. Consider the set  $\fI := \{\rho_p: p \in \fP(\cY)\}$. We first show the left 
  equality in (\ref{counterexample:ineqs}). If we define stochastic matrices $P^{(1)}_{U|XY}, \ P^{(2)}_{U|XY}: \ \cX \times \cY \rightarrow \ \cU := \{0,1\}$
  with entries
  \begin{align*}
   P^{(1)}_{U|XY}(u|x,y) := \delta_{xu},  \hspace{.2cm} \text{and} \hspace{.4cm} 
   P^{(2)}_{U|XY}(u|x,y) := \delta_{yu} &&(x \in \cX, y \in \cY, u \in \cU), 
  \end{align*}
  and use the sender's preprocessings $P^{(1)}_{U|XY}$  for $\rho_{\pi}$ and  $P^{(2)}_{U|XY}$ for each $p \neq \pi$, we achieve the maximum in the capacity formula derived in Theorem \ref{secret-key_generation-theorem-ssi}. 
  The corresponding states are
  \begin{align*}
   \hat{\rho}_{\pi} 
   &:= \sum_{u \in \cU} \sum_{x \in \cX} \sum_{y \in \cY} \ P^{(1)}_{U|XY}(u|x,y) \frac{1}{4} \ket{u} \bra{u} \otimes V_{1,B}(x) \otimes V_{1,E}(y) \\
   &= \sum_{u \in \cU} \frac{1}{2} \ket{u} \bra{u} \otimes \ket{u}\bra{u} \otimes \Pi \otimes \Pi \otimes \Pi \\
   \hat{\rho}_{p}     
   &:=  \sum_{u \in \cU} \sum_{x \in \cX} \sum_{y \in \cY} \ P^{(2)}_{U|XY}(u|x,y) p(x) \frac{1}{2} \ket{u} \bra{u} \otimes V_{2,B}(y) \otimes V_{2,E}(x) \\
   &= \sum_{u \in \cU} \frac{1}{2} \ket{u} \bra{u} \otimes  \Pi \otimes \ket{u}\bra{u} \otimes \sigma_{p,E} &&(p \in \fP(\cY) \setminus \{\pi\}),
  \end{align*}
  where $\sigma_{p,E} := \sum_{x \in \cX} p(x) V_{2,E}(x)$. Note that both of the above states contain perfect common randomness between the legitimate users without sharing any correlations 
  to the eavesdropper, which is the optimum they can achieve, as is easily observed. It therefore holds
  \begin{align*}
    K_{\rightarrow,SMI}(\fI) = \log 2 = 1. 
  \end{align*}
  The situation is completely different, if no SMI is present. Let $\mu > 0$ be fixed and $(T,D)$ an arbitrary $(n,M,L,\mu)$ forward secret-key distillation protocol for $\fI$ without SMI. I.e. the inequalities
  \begin{align}
   \prob(K_p \neq K'_p) \leq \mu \label{counterexample_mu_1}
  \end{align}
  and 
  \begin{align}
   \log M - H(K_p) + I(K,\Lambda E^n, \rho_{K\Lambda E^n,p}) \leq \mu \label{counterexample_mu_2}
  \end{align}
  are satisfied for each $p \in \fP(\cY)$. If we define the states
  \begin{align*}
   \tilde{\rho}_p := \sum_{x \in \cX} \sum_{y \in \cY} p(x) \frac{1}{2} \ket{x \otimes y}\bra{x \otimes y} \otimes V_{1,B}(x) \otimes V_{2,E}(x) &&(p \in \fP(\cY) \setminus \{\pi\}),
  \end{align*}
  the identity
  \begin{align}
   I(K;\Lambda E^n, \rho_{K\Lambda E^n,p}) = I(K;\Lambda B^n, \tilde{\rho}_{K\Lambda B^n,p}) \label{counterexample_mutual_identity}
  \end{align}
  is fulfilled by symmetry. Moreover, it holds
  \begin{align}
   \left\|\rho_{K\Lambda B^n, \pi}- \tilde{\rho}_{K\Lambda B^n,p}\right\|_1 \leq \|\tr_{\cH_E^{\otimes n}}\rho_{\pi}^{\otimes n} - \tr_{\cH_E^{\otimes n}}\tilde{\rho}_{p}^{\otimes n} \|_1 = \|p^n - \pi^n \|_1 \leq n \|p-\pi \|_1\label{counterexample_distance_bound}
  \end{align}
  where the first inequality is by c.p.t.p. monotonicity of the trace norm distance, and the second is by construction. Combining (\ref{counterexample_mutual_identity}) and (\ref{counterexample_distance_bound}) with Fannes' inequality for the 
  quantum mutual information, we obtain
  \begin{align}
    I(K;\Lambda E^n, \rho_{K\Lambda E^n,p}) 
    &= I(K;\Lambda B^n, \tilde{\rho}_{K\Lambda B^n,p}) \nonumber \\
    &\leq I(K;\Lambda B^n, \rho_{K\Lambda B^n,\pi)}) + f(n\|p - \pi \|_1) \label{counterexample_mutual_bound}
  \end{align}
   for each $p \in \fP(\cY) \setminus \{\pi\}$, where $f$ is a function with $f(a) > 0$ for all $a > 0$, and $f(a) \rightarrow 0, (a \rightarrow 0)$. 
   Therefore, we have for each $p \neq \pi$
   \begin{align}
    \log M \ 
    &\leq \ H(K_{p}) - I(K; \Lambda E^n, \rho_{K\Lambda E^n, p}) + \mu \nonumber \\
    &\leq \ H(K_{p}) - I(K; \Lambda B^n, \rho_{K\Lambda B^n, \pi}) + \mu + f(n\|p - \pi \|_1) \nonumber \\
    &\leq \ H(K_{p}) - I\left(K_{\pi};K'_{\pi} \right) + \mu + f(n\|p - \pi \|_1 ) \nonumber \\
    &\leq \ H(K_{p}) - H(K_{\pi}) + H\left(K_{\pi}|K'_{\pi} \right) + \mu + f(n\|p - \pi \|_1 ) \nonumber \\
    &\leq \ H(K_{p}) - H(K_{\pi}) + \mu \log M + h(\mu) + \mu + f(n\|p - \pi \|_1), \nonumber \\
    &\leq  \mu \log M + h(\mu) + \mu + 2f(n\|p - \pi \|_1), \label{counterexample_error}
  \end{align}
   where the first inequality is (\ref{counterexample_mu_2}),  the second is by (\ref{counterexample_mutual_bound}), the third is by the quantum data processing inequality, the fifth by Fano's inequality together with (\ref{counterexample_mu_1}),
   and the last is by Fannes' inequality. By taking the infimum over all $p$ in the above inequality arrive at
   \begin{align}
    \log M \ \leq  \mu \log M + h(\mu) + \mu.  \label{counterexample_error_2}
   \end{align}
   We conclude that $R = 0$ is the only achievable forward secret-key distillation rate.
  \end{proof}
  \begin{remark}
  The lack of sender knowledge can have worst consequences. A closer look at the example introduced to prove the preceding theorem shows, how different the situations with and without sender knowledge can be. With sender knowledge, we
  achieve capacity with zero error and security index for each blocklength, while all public forward communication needed is the information whether $\pi$ is present or not. On the other hand, the bound (\ref{counterexample_error_2}) 
  reveals that no nonzero forward secret-key rate can be achieved even if nonzero asymptotically performance $\mu$ is allowed asymptotically!
  \end{remark}
\end{subsection}
\begin{subsection}{A weaker notion of regularity} \label{sect:regularity_condition_hemicontinuity}
  In this section we show that a slightly weaker condition on the set $\fI$ of cqq density matrices generating the outputs of the compound source is sufficient for proving a version of Theorem \ref{secret-key_generation-theorem-regular}.
  To formulate the corresponding assertion, we introduce some notions from the theory of set-valued maps, where we take the corresponding definitions from Chapter 11 in \cite{border85}. In the following we denote for 
  each given set $\Omega$ the power set of $\Omega$ (i.e. the family of subsets of $\Omega$) by $2^{\Omega}$. \\ 
  Let $f:\ \Theta \rightarrow 2^{\Omega}$ be a set-valued map. We define for each $E \subset \Omega$
  \begin{align}
   f^+(E) := \{\theta \in \Theta: \ f(\theta) \subset E \}, \hspace{.2cm} \text{and} \hspace{.4cm} f^-(E) := \{\theta \in \Theta: \ f(\theta)\cap E \neq \emptyset \}.
  \end{align}
  \begin{definition}
   We call a set-valued map $f: \ \Theta \rightarrow 2^{\Omega}$ 
   \begin{enumerate}
    \item \emph{upper hemi-continuous}, if for each $\theta \in \Theta$ the following is true. Whenever $\theta \in f^+(E)$ for an open set $E$, there is a neighbourhood $U(\theta)$ of $\theta$ with $U(\theta) \subset f^+(E)$.
    \item \emph{lower hemi-continuous}, if for each $\theta \in \Theta$ the following is true. Whenever $\theta \in f^-(E)$ for an open set $E$, there is a neighbourhood $U(\theta)$ of $\theta$ with $U(\theta) \subset f^-(E)$.
    \item \emph{continuous}, if $f$ is both upper and lower hemi-continuous.
   \end{enumerate}
  \end{definition}
  We will always regard $\Theta$ and $\Omega$ being finite-dimensional. In this case, we obtain sequential characterizations of upper and lower hemi-continuity, if we assume the set-valued function to have only compact values.
  \begin{proposition}
   Let $f: \Theta \rightarrow 2^{\Omega}$ be a set-valued map with $\Theta \subset \bbmR^m$, $\Omega \subset \bbmR^k$, and $f(\theta)$ compact for each  $\theta \in \Theta$. It holds
   \begin{enumerate}
    \item $f$ is \emph{upper hemi-continuous} if and only if for each $\theta \in \Theta$, every sequence $(\theta_n)_{n \in \bbmN}$ with $\theta_n \rightarrow \theta (n \rightarrow \infty)$ and $\omega_n \in f(\theta_n)$, $n \in \bbmN$ 
    there is a subsequence $\{\omega_{n_k}\}_{k \in \bbmN}$ with $\lim_{k \rightarrow \infty} \omega_k \in f(\theta)$.
    \item \emph{lower hemi-continuous}, if for each $\theta \in \Theta$ and sequence $\{\theta_n\}_{n\in \bbmN} \subset \Theta$, and $\omega \in f(\theta)$ from $\lim_{n \rightarrow \infty} \theta_n = \theta$ it follows,
    that there is a sequence $\{\omega_n \}_{n \in \bbmN}$ with $\omega_n \in f(\theta_n)$, $n \in \bbmN$ and $\lim_{n \rightarrow \infty} \omega_n = \omega$.
    \end{enumerate}
  \end{proposition}
   \begin{proof}
    See \cite{border85}, Proposition 11.11.
   \end{proof}
  For our considerations the closed-graph characterization of upper hemi-continuity will be of utility. 
  \begin{theorem} \label{thm:closed_graph}
   Let $\Theta \subset \bbmR^m$, $\Omega \subset \bbmR^k$, $f: \Theta \rightarrow 2^{\Omega}$ be a set-valued map with $\Omega$ being compact. If the graph of $f$, i.e. the set
   \begin{align}
    Gr f := \{(\theta, \omega) \in \Theta \times \Omega: \ \omega \in f(\theta) \}
   \end{align}
   is closed, then $f$ is upper hemi-continuous.
  \end{theorem}
  \begin{proof}
   See for example Proposition 11.9 in \cite{border85}
  \end{proof}
  We need the following basic Lemma. 
  \begin{lemma}\label{lemma:closed_lower_semi}
   If a set-valued function is lower hemi-continuous, then its closure $\overline{f}$ (i.e. the function defined by closing the graph of $f$) is lower hemi-continuous as well.
  \end{lemma} 
  \begin{proof}
   Assume that there is a sequence $\{\theta_n\}_{n \in \bbmN}$ with $\theta_n \rightarrow \theta$ and $\omega \in \overline{f}(\theta)$, 
   such that no sequence $\{\omega_n\}_{n \in \bbmN}$ exists with $\omega_n \in \overline{f}(\theta_n)$ for all $n \in \bbmN$ and $\omega_n \rightarrow \omega$. If $\omega$ 
   is in $f(\theta)$, such a sequence always exists by lower hemi-continuity of $f$. If $\omega$ is in $\overline{f}(\theta)  \setminus f(\theta)$ the hypothesis is only 
   true if $\omega$ is no point of accumulation of $\overline{f}(\theta)$, which contradicts the definition of $\overline{f}$.
  \end{proof}
  \begin{definition}\label{def:weak_regularity}
   We call a set $\fI \subset \cS_{cqq}(\cH_{ABE})$ \emph{weakly regular}, if the set-valued map
   \begin{align}
    f_{AX}&: \cP_{\fI} \rightarrow 2^{\cS_{cq}(\cH_{AX})} \\
	 p&\mapsto \fI^{AX}_p
   \end{align}
   is lower hemi-continuous for $X = B,E$. 
  \end{definition}
  
  \begin{proposition} \label{approximating_set_hemicontinuity}
   Let $\fI \subset \cS_{cqq}(\cH_{ABE})$ be a weakly regular set of cqq density matrices. There exists a regular set $\hat{\fI} \subset \cS_{cqq}(\cH_{ABE})$ with
   \begin{enumerate}
    \item $\fI \subset \hat{\fI}$ 
    \item $K_{\rightarrow}(\fI) \leq K_{\rightarrow}(\hat{\fI})$.
    \item $\hat{\fI}$ is regular
   \end{enumerate}
   \end{proposition}
   \begin{proof}
    Assume $\fI$ being parameterized as 
    \begin{align}
     \fI := \left\{\rho_{(p,V)}: \ \sum_{x \in \cX} \ p(x) \ket{x} \bra{x} \otimes V(x) \right\}_{(p,V) \in S}
    \end{align}
    with 
    \begin{align}
     S := \bigcup_{p \in \fP_{\fI}} \{p\} \times \cV_p
    \end{align}
    with sets $\cP_\fI \in \fP(\cX)$, $\cV_p  \subset \cq(\cX, \cH_{BE})$, $p \in \cP_\fI$. 
    We define $\hat{\fI}$ as the closure of $\fI$. Obviously, the first condition $\fI \subset \hat{\fI}$ stated in the proposition is fulfilled. We show that the two remaining conditions 
    are also fulfilled. 
    Assume that $(T,D)$ is an $(n,M,L,\mu)$ forward secret-key distillation protocol for $\fI$. Since the performance and security criteria in Definition \ref{def:sk_prot_perf} 
    are defined in terms of functions being continuously dependent on the cqq density matrix, it is clear that $(T,D)$ is an $(n,M,L,\mu)$ forward secret-key distillation protocol 
    also for $\hat{\fI}$, which directly implies that also the second claim of the proposition is satisfied. \\ 
    For validating the third claim we notice that since $\hat{\fI}$ is closed, the corresponding set-valued functions $f_{AB}$ and $f_{AE}$ have closed graphs. Therefore both maps 
    are upper hemi-continuous by Theorem \ref{thm:closed_graph}. The hypothesis of $\fI$ being weakly regular, together with Lemma \ref{lemma:closed_lower_semi} ensures us that
    $\overline{f}_{AB}$ and $\overline{f}_{AE}$ are also lower hemi-continuous. Therefore, they are continuous. Since the set of sender marginal distributions $\cP_{\hat{\fI}}$ deriving from $\hat{\fI}$ 
    is a compact set, we infer that $f_{AB}$, $f_{AE}$ are uniformly continuous, which implies that for each $\epsilon > 0$ we find a $\delta > 0$,
    such that the implication
    \begin{align}
     \|p -q\|_1 < \delta \ \Rightarrow \  d_H(f_{AB}(p), f_{AB}(q)) + d_H(f_{AE}(p), f_{AE}(q)) < \epsilon
    \end{align}
    for each $p, q \in \cP_{\hat{\fI}}$. Since 
    \begin{align}
     \ d_H(\hat{\fI}^{AB}_p, \hat{\fI}^{AB}_q) + d_H(\hat{\fI}^{AE}_p, \hat{\fI}^{AE}_q) \ = \ d_H(f_{AB}(p), f_{AB}(q)) + d_H(f_{AE}(p), f_{AE}(q))
    \end{align} 
    holds by definition, $\hat{\fI}$ is regular. 
    \end{proof}
   \begin{theorem}
    Let $\fI$ be a weakly regular set of cqq density matrices on $\cH_{ABE}$. It holds
   \begin{align}
    K_{\rightarrow}(\fI) = \lim_{k \rightarrow \infty} \frac{1}{k} K_{\rightarrow}^{(1)}(\fI^{\otimes k}), 
   \end{align}
   \end{theorem}
   \begin{proof}
    We approximate $\fI$ by the set $\hat{\fI}$ as defined in the proof of Proposition \ref{approximating_set_hemicontinuity}.
    The first and second property of $\hat{\fI}$ in Proposition \ref{approximating_set_hemicontinuity} together imply that 
    \begin{align}
     K_{\rightarrow}(\fI) = K_{\rightarrow}(\hat{\fI}) = \lim_{k \rightarrow \infty} \frac{1}{k} K_{\rightarrow}^{(1)}(\hat{\fI}^{\otimes k})
    \end{align}
    holds. The rightmost of the above inequalities is by application of Theorem \ref{secret-key_generation-theorem-regular} on $\hat{\fI}$, which is possible, 
    because $\hat{\fI}$ is regular by Proposition \ref{approximating_set_hemicontinuity}. In fact, $d_H(\fI, \hat{\fI}) = 0$, and consequently  
    $d_H(\fI^{\otimes k}, \hat{\fI}^{\otimes k}) = 0$ holds for each $k \in \bbmN$. Therefore
    \begin{align}
     K_{\rightarrow}^{(1)}(\fI^{\otimes k})  = K_{\rightarrow}^{(1)}(\hat{\fI}^{\otimes k})
    \end{align}
    holds for each $k \in \bbmN$ by continuity of $K^{(1)}$. We are done. 
    \end{proof}
\end{subsection}
  \end{section}
\begin{section}{Special case: Forward secret-key distillation capacity of a classical tripartite compound sources}
 Our results also cover the case of a completely classical tripartite compound source. Let $(X,Y,Z)$ be a triple of classical random variables with distribution $P_{XYZ} \in \fP(\cX \times \cY \times \cZ)$. The state of this classical 
 system coherified to a Hilbert space $\cH_{X} \otimes \cH_{Y} \otimes \cH_Z$ is represented by the density matrix
 \begin{align}
  \rho := \sum_{(x,y,z) \in \cX \times \cY \times \cZ} P_{XYZ}(x,y,z) \ket{x}\bra{x} \otimes \ket{y} \bra{y} \otimes \ket{z} \bra{z}. \label{classical_density_matrix}
 \end{align}
  Forward secret-key distillation for this kind of classical compound memoryless source was considered in \cite{tavangaran16} done under collaboration of one of the authors of the present paper. Among other results obtained therein, it was
  derived a capacity formula for the case without sender marginal knowledge in case that the set of sender marginal distributions deriving from the source is finite. Our results extend the capacity description also to the case of an arbitrary 
  regular tripartite classical 
  source. By coherifying each set $(X_s,Y_s,Z_s)_{s \in S}$ of triples into a mutually commuting set of density matrices as in (\ref{classical_density_matrix}), Theorem \ref{secret-key_generation-theorem-regular} directly leads to a version
  of Theorem 2 in \cite{tavangaran16} where the assumption on the set of sender marginal states to be finite can be replaced by the substantially weaker assumption of regularity of the source. Theorem \ref{secret-key_generation-theorem-ssi} in the same way provides a 
  capacity formula for the case where the sender party perfectly knows the distribution of his/her part of the source. The reader may reply that the definition of a (deterministic) classical decoding procedure is more restricting than that of a 
  POVM in quantum theory, since the decoding sets are demanded to be pairwise non-intersecting. Our reasoning is not affected by this fact, because in case of pairwise commuting density operators, optimal decoding can always be achieved by using 
  projection valued measures. Alternatively, one could replace Lemma \ref{good_codes_0} by a completely classical version.\\
  We point out that the need for a regularity condition as well as differences between the capacities with and without sender's knowledge of the marginal state are not effects of the quantum nature of the sources considered in this paper. The reader 
  may notice that the example given to prove Theorem \ref{counterexample} is essentially classical, since all density matrices involved pairwise commute. \\
  Since there are stronger preresults available in classical information theory (especially regarding error exponents for coding of classical compound channels), a classical method of proof might lead to faster decrease of error with a potentially 
  simplified proof. 
 \end{section}
\begin{section}{Conclusion} \label{sect:conclusion}
We have considered the task of secret-key distillation under free forward classical communication for compound memoryless sources with classical legitimate sender and quantum legal receiver and eavesdropper outputs. We derived a capacity formula
for all sources of this class which additionally exhibit a certain regularity condition. \\ 
We also discussed the situation, where the legitimate sender has perfect knowledge of the probability distribution governing his/her outputs. In this 
case, we were able to derive a capacity formula which equals the one given for the case without sender marginal information for regular sources, and moreover does hold for all non-regular sources. \\ 
As we have also seen, the capacities with and without sender marginal information differ at least for some non-regular sources. We admit that the regularity conditions assumed in this paper may be somewhat weakened to determine the forward secret-key 
generation capacity without sender knowledge for a larger class of sources. We provided a further step in this direction by applying the general theory of set-valued maps to derive a slightly broader class of compound cqq sources with a general
capacity description.\\ 
We leave open the more general case of proving a capacity theorem for forward secret-key distillation from compound sources where the generating set of density matrices may contain members being not in the class of cqq density matrices. 
In \cite{devetak05} such a theorem were proven in case of a perfectly known source without restriction on the legitimate sender to be classical. \\ 
The strategy therein to prove a coding theorem was, to combine an achievability 
result for cqq sources with an optimization over instruments dephasing the sender's system to a classical one. Notice that such a strategy in general does not apply to compound sources in a direct way as it did in \cite{devetak05} 
(at least in case that the sender does not have 
perfect marginal knowledge).  In general,  there is no control whether or not a dephasing operation leads to a non-regular compound cqq source. \\ 
However, the approximation techniques presented here may lead to a better understanding of the secret-key distillation
task even for tripartite quantum compound sources.\\ 
Another astonishing result from Ref.\citenum{devetak05} is the close correspondence between the forward secret-key distillation and one-way LOCC entanglement distillation tasks. It was demonstrated that modifying forward secret-key distillation protocols leads to 
one-way local operations and classical communications protocols suitable for proving the so-called quantum hashing bound which is used to determinine the one-way entanglement distillation capacity of bipartite memoryless quantum sources.  \\
The authors of this paper are of the opinion that following a similar strategy to derive the entanglement distillation capacity of bipartite compound memoryless quantum sources may be not successful in the same way as it is with perfect knowledge of the source.
A closer look to the corresponding considerations in \cite{devetak05} may underpin this opinion. Therein, an important part of the coding strategy was to apply a nondestructive measurement on the sender's marginal of the bipartite quantum state subject to 
entanglement distillation -- a strongly state-dependent task, which can hardly be performed without sender marginal knowledge. 
For a generalization to the case of bipartite compound quantum sources under assumption of SMI, the strategy may be feasible. We did not 
pursue this path, because the one-way entanglement distillation capacity of compound quantum sources is already known from Refs. \citenum{bjelakovic13_b}, and \citenum{boche14}. \\ 
This parallels a similar observation made for channel coding from \cite{boche14_b}. Therein, it was argued that the ingenuous way to derive entanglement generation codes for quantum channels from codes for secret message transmission over 
classical-quantum wiretap channels used in \cite{devetak05_b} for proving the quantum coding theorem leads to suboptimal results for compound channels if no sender state knowledge is assumed.
\end{section}
\begin{acknowledgments}
 The authors would like to thank their colleagues Andrea Grigorescu Vlass, Nima Tavangaran, Sebastian Baur, and Harout Aydinian for their encouragement and interest in this research. This work was supported by the BMBF via grant 01BQ1050.
\end{acknowledgments}
\begin{appendix}
 \begin{section}{Universal random constant composition codes for compound DMcq Channels}
 In this section, we state and prove some results on compound DMcq channels we need within the proof of Lemma \ref{good_codes_0}. For convenience of the reader, we first provide definitions, 
 necessary to understand the subsequent arguments. Let $\cV \subset CQ(\cX,\cK)$
 be a set of cq channels mapping a finite alphabet $\cX$ to the set of density matrices on a Hilbert space $\cK$. The \emph{compound discrete memoryless classical quantum channel generated by $\cV$} (the \emph{DMcqC $\cV$} 
 for short) is given by the family $\{V^{\otimes n}: \ V \in \cV\}_{n \in \bbmN}$ of possible outputs. To catch up with the notation from \cite{bjelakovic09}, we sometimes write $\cV = \{V_s\}_{s \in S}$ assuming a 
 suitable parameterization of $\cV$ by an index set $S$. For given blocklength $n \in \bbmN$, $M \in \bbmN$, an \emph{$(n, M)$-code for transmission of classical messages over $\cV$} is a family 
 $\cC := (u_m, D_m)_{m=1}^M$ with $u_m \in \cX^n$, $D_m \in \cL(\cK^{\otimes n})$ for each $m  \in [M]$, with the additional property that for all $m \in [M]$
 \begin{align*}
  0 \leq D_m \leq \bbmeins_{\cK}^{\otimes n}, \hspace{.3cm}\text{and}\hspace{.3cm} \sum_{m=1}^M D_m \leq \bbmeins_{\cK}^{\otimes n}
 \end{align*}
 holds. For given $(n,M)$-code $\cC$, $s \in S$, we define the \emph{average error of transmission} by
 \begin{align*}
  \overline{e}(\cC, V_s^{\otimes n}) := \frac{1}{M} \sum_{m=1}^M \ \tr(D_m^{\perp} V^{\otimes n}(u_m) ),
 \end{align*}
 where we allow ourselves to define $A^\perp := \bbmeins_{\cK}^{\otimes n} - A$ (even if $A$ is not a projector). The following two Propositions combined, immediately imply the proof of Lemma \ref{good_codes_0} stated in 
 the text. The claim of the 
 first one follows by careful modification of the proof of Theorem 5.10 from \cite{bjelakovic09} and delivers for each large enough blocklength and each type of sequences a suitable random compound transmission code  
 having super-polynomially decrease of error, universal regarding the channels in the compound set, as well as the appearing types.
 \begin{proposition}\label{good_codes_1}
 Let $\cW \subset \cq(\cX, \cK)$ be an arbitrary set of cq channels. For each $\delta > 0$, there is a number $n_2 := n_2(\delta)$ such that for each $n> n_2$ and each type $\lambda \in \fT(n,\cX)$
 there exists a random $(n,M_\lambda)$-code $\cC_\lambda(U) := (U_m, D_m(U))_{m=1}^{M_\lambda}$ for $\cW$ with the following properties
 \begin{enumerate}
  \item $U := (U_1,...,U_{M_\lambda})$ is an i.i.d. sequence of random variables, each with values in $\cX^n$ and distribution $\lambda^{\otimes n}$.
  \item $M_\lambda \geq \exp\left\{n\left(\underset{W \in \cW}{\inf}\ \chi(\lambda, W) -\delta\right)\right\}$.
  \item $\bbmE\left[\underset{W \in \cW}{\sup}\ \overline{e}(\cC_\lambda, W^{\otimes n})\right] \leq 2^{-\sqrt[16]{n}c(\delta)}$ 
 \end{enumerate}
 with a constant $c(\delta) > 0$. 
\end{proposition}

\begin{proof}
 The assertion is basically contained in the proof of Theorem 5.10 in \cite{bjelakovic09} and follows by minor modifications of the argument given there. We assume the reader's familiarity with the arguments presented in
 \cite{bjelakovic09} and restrict ourselves to indicate the steps of modification necessary to justify 
 our claim. \\
 We write $\cW := \{W_t\}_{t \in T}$ with a suitable index set $T$. Let $n \in \bbmN$, and consider a type $\lambda \in \fT(n,\cX)$. Assume that 
 \begin{align}
  \underset{t \in T}{\inf}\ \chi(\lambda, W_t) - \delta > 0, \label{type_0}
 \end{align}
 holds, since otherwise the claim is trivially fulfilled for $\lambda$. Choose an approximating set $\cW_n := \{W'_t\}_{t \in T_n}$ for $\cW$ as used in \cite{bjelakovic09}. We will execute the suggestions given in 
 Remark 5.11 therein, and therefore choose
 the diameter of the partition of output states used to define $\cW_n$ to be $2^{-\sqrt[16]{n}}$ instead of $\frac{1}{n^2}$. Define 
 \begin{align*}
  \Omega_{\lambda,n} := \{\rho'_{t,\lambda} := \sum_{x \in \cX} \lambda(x) \ket{x}\bra{x} \otimes W'_t(x)\ : \ t \in T_n \},
 \end{align*}
 \begin{align*}
  \lambda := \sum_{x \in \cX} \lambda(x) \ket{x}\bra{x}, \hspace{.2cm} \text{and} \hspace{.2cm}  \sigma'_{t,\lambda} := \sum_{x \in \cX} \lambda(x) \sigma'_t(x) &&(t \in T_n).
 \end{align*}
 Note that the properties of the set $\cW_n$ in \cite{bjelakovic09} (see Lemma 5.6 therein) together with the above definitions, implies the bound
 \begin{align}
  \lambda_{\min}(\lambda \otimes \sigma'_{t,\lambda}) 
  &\geq \underset{x \in \supp(\lambda)}{\min} \lambda(x)\cdot \frac{1}{d} \cdot 2^{-\sqrt[16]{n}} \nonumber \\
  &\geq \frac{2^{-\sqrt[16]{n}}}{n\cdot d} \label{type_1}
 \end{align}
 on the minimal eigenvalue $\lambda_{\min}(\lambda \otimes \sigma'_{t,\lambda})$ of $\lambda \otimes \sigma'_{t,\lambda}$. The last estimate above follows from the fact that $\lambda$ is a 
 type of sequences in $\cX^n$. Closely following the argument given in \cite{bjelakovic09}, we are ensured that choosing $l_n = [\sqrt{n}]$, 
 we find $a_n,b_n \in \bbmN$, $0 \leq b_n < l_n$, with 
 \begin{align*}
  n = a_n l_n + b_n,
 \end{align*}
 and a projection-valued measure $\cM_{l_n,\lambda} := \{P_{1,l_n,\lambda}, P_{2,l_n,\lambda}\}$, such that for all $t,s \in T_n$
 \begin{align}
  S_{M_{l_n,\lambda}}({\rho'}_{t,\lambda}^{\otimes {l_n}}||(\lambda \otimes \sigma_{\lambda,s})^{\otimes l_n}) 
  & \geq l_n(S(\Omega_{\lambda,n}||\lambda \otimes \sigma'_{\lambda,s}) - \xi_{l_n}(\lambda \otimes \sigma'_{\lambda,s}) \nonumber \\
  & \geq l_n(\min_{t \in T_n} \chi(\lambda, W'_t) - \xi_{l_n}(\lambda \otimes \sigma'_{\lambda,s})) \label{type_2}
 \end{align}
 holds. Careful investigation of the function $\xi_{l_n}$ in \cite{bjelakovic09} using the type-independent bound in (\ref{type_1}) shows that
 \begin{align}
 \lim_{n \rightarrow \infty} \underset{\lambda \in \fT(n, \cX)}{\max} \underset{s \in T_n}{\max}\ \xi_{l_n}(\lambda \otimes \sigma'_{\lambda,s}) = 0  \label{type_3}
 \end{align}
 holds. Introducing the refinement $\cQ_{l_n,\lambda}$ of the projection-valued measure $\cM_{l_n,\lambda}$, and the stochastic matrices $V_{t,\lambda}$, $t \in T_n$ generated by $Q_{l_n,\lambda}$ as in \cite{bjelakovic09} 
 (note that these, also may depend on the chosen type $\lambda$), it holds
 \begin{align*}
  \frac{1}{l_n} \underset{t \in T_n}{\min}\ I(\lambda^{\otimes l_n}, V_{t,\lambda}) \geq \underset{t \in T}{\inf}\ \chi(\lambda, W_t) - n\cdot 2^{-\sqrt[16]{n}}C(d) - \xi_{l_n,\max},
 \end{align*}
 where we used $\xi_{l_n,\max}$  defined by
 \begin{align*} 
  \xi_{l_n,\max} := \underset{\lambda \in \fT(n, \cX)}{\max} \underset{s \in T_n}{\max}\ \xi_{l_n}(\lambda \otimes \sigma'_{\lambda,s}).
 \end{align*}
Since (\ref{type_3}) holds, we find for each $0 < \eta < \delta$ a number $n_2(\eta) \in \bbmN$ (independent of $\lambda$) such that
 \begin{align}
  \frac{1}{l_n}\ \underset{t \in T_n}{\min}\ I(\lambda^{\otimes l_n}, V_{t,\lambda}) \geq \underset{t \in T}{\inf} \ \chi(\lambda, W_t) - \eta > 0 \label{type_4}
 \end{align}
 is fulfilled for all $n > n_2$. Note that the last inequality holds by (\ref{type_0}). The bound above differs a bit from the one given in Eq. (30) in \cite{bjelakovic09}. 
 However, it will be sufficient for the following argument. 
 Let
 \begin{align}
  \Theta := \left\{\theta \in \bbmR: \ 0 < \theta < \frac{\eta}{4  } \right\}, \label{type_5}
 \end{align}
 and
 \begin{align*}
  I_{n,\lambda} 
  &:= \underset{t \in T_n}{\min}\ I(\lambda^{\otimes l_n}, V_{t,\lambda}). 
 \end{align*}
  Following the lines of \cite{bjelakovic09} (always respecting dependencies on $\lambda$), we yield the bound 
 \begin{align*}
  \Pr(i_{\lambda}^{a_n} \ \leq I_{n,\lambda} - 2l_n\theta) \ \leq \frac{1}{|T_n|} \sum_{t \in T_n} \Pr_{t,\lambda}(i_{t,\lambda}^{a_n} \leq I_n - l_n \theta) + |T_n|2^{-a_n l_n \theta}.
 \end{align*}
  Since $i^{a_n}_{t,\lambda}$ is a sum of i.i.d. random variables each with values in the interval 
  \begin{align*}
   [-l_nd \sqrt[16]{n}, l_n d \sqrt[16]{n} ],
  \end{align*}
   and
  \begin{align*}
   I_{n,\lambda} \leq \bbmE_{t,\lambda}(i^{a_n}_{t,\lambda})
  \end{align*} 
  holds for each $t \in T_n$, our counterpart to Eq. (34) in \cite{bjelakovic09},
  \begin{align*}
   \Pr(i^{a_n}_{t,\lambda} \leq I_n - l \theta) \leq e^{\tfrac{a_n \theta^2}{16 \cdot \sqrt[16]{n}}}
  \end{align*}
  is valid. By closely following the lines of \cite{bjelakovic09} (having in mind our bounds) together with the choice $\theta = \frac{\eta}{4}$, we know that there is a projection $P_{n,\lambda,\theta}$ with
  \begin{align*}
   \tr(\rho^{(n)}_{\lambda}P_{n,\lambda,\theta}) \geq 1 - e^{\tfrac{a_n \eta^2}{16^2 \sqrt[16]{n}}} - |T_n| 2^{-a_nl_n \frac{\eta}{4}}
  \end{align*}
  and 
  \begin{align}
   \tr((\lambda^{\otimes n} \otimes \sigma^{(n)}_\lambda)P_{n,\lambda,\theta}) 
   &\leq 2^{-a_n(I_n - 2 l_n \theta)} \nonumber \\
   &\leq 2^{-a_nl_n(\underset{t \in T}{\inf} \ \chi(\lambda, W_t) - \frac{3}{2}\eta )}, \label{type_6}
  \end{align}
  where the last of the above inequalities above holds by (\ref{type_4}). Since $b < \sqrt{n}$, 
  \begin{align}
    &\tr((\lambda^{\otimes n} \otimes \sigma^{(n)}_\lambda)P_{n,\lambda,\theta}) \leq 2^{-n(\underset{t \in T}{\inf} \ \chi(\lambda, W_t) - 2\eta )} \label{type_7}
  \end{align}
  holds for each $n > n_2$, if $n_2$ is chosen large enough. Setting $\eta := \frac{\delta}{2}$, and using the bounds in (\ref{type_6}) and (\ref{type_7}) together with Theorem 1.1. stated in the Appendix of \cite{bjelakovic09} (and leaving out the derandomization step leading to a deterministic 
  code in the Proof of Theorem 5.10 in \cite{bjelakovic09}), we conclude that there is an $(n,M'_\lambda)$ random code $\cC_{\lambda}(U) := (U_m, D_m(U))_{m=1}^{M'}$ for the average channel 
  $W^{\otimes n} := \frac{1}{|T_n|} \sum_{t \in T_n} {W'}_t^{\otimes n}$ with $U = (U_1,\dots,U_{M'})$ being a sequence of i.i.d. random variables each with values in $\cX^n$ and generic distribution $\lambda^{\otimes n}$,
  such that
  \begin{enumerate}
   \item $M \geq \left\lfloor \exp \left\{ n \left( \underset{t \in T}{\inf}\ \chi(\lambda W_t) - \frac{3}{2} \delta \right)\right\} \right\rfloor$.
   \item $\bbmE\left[\overline{e}(\cC_\lambda, W^{\otimes n})\right] \leq 2^{-\sqrt[16]{n} \tilde{c}(\delta)}$  
  \end{enumerate}
  with a positive constant $\tilde{c}(\delta)> 0$. From this, we can conclude that there is a number $n_0(\delta)$, such that for each $n > n_0(\delta)$ (independent of $\lambda$) $\cC_\lambda$ is an $(n,M)$ random code 
  of the above stated properties, with 2. above replaced by
  \begin{align*}
  \bbmE\left[\underset{W \in \cW}{\sup}\ \overline{e}(\cC_\lambda, W^{\otimes n})\right] \leq 2^{-\sqrt[16]{n}c(\eta)},  
  \end{align*}
  where $c(\delta) := \frac{1}{2}\tilde{c}(\delta)$. 
\end{proof}
The following proposition is a compound version of Proposition 2.5 in \cite{devetak05}. It is proven by exactly the same strategy replacing the Holevo-Schumacher-Westmoreland codes for $DMcqC$ with perfectly known 
generic cq channel by the codes constructed in \cite{bjelakovic09} together with the modifications done in Proposition \ref{good_codes_1} above.
\begin{proposition}[cf. \cite{devetak05}, Prop. 2.5]\label{good_codes_2}
 Let $\cW \subset CQ(\cX, \cK)$ be an arbitrary set of cq channels, $n \in \bbmN$ and $\lambda \in \fT(n,\cX)$ a type of sequences in $\cX^n$. If there exists a random $(n, M')$-classical message transmission code 
 \begin{align*}
  \cC'(U) = (U_m, D_m(U))_{m=1}^{M'}
 \end{align*}
 for the DMcqC $\cW$ which has the properties 
 \begin{enumerate}
  \item $U := (U_1,\dots,U_{M'})$ is an i.i.d. sequence of random variables with values in $\cX^n$ with generic distribution $\lambda^n$
  \item $\bbmE\left[\underset{W \in \cW}{\sup}\ \overline{e}(\cC, W^{\otimes n})\right] \leq \mu$
 \end{enumerate}
 with $\mu \in (0,1)$, then there exists for each given $\vartheta \in (0,1)$ a random $(n,M)$-message transmission code 
 \begin{align*} 
 \cC(V):= (V_m, D_m(V))_{m=1}^M  
 \end{align*}
 having the properties
 \begin{enumerate}
  \item $V := (V_1,...,V_M)$ is an i.i.d. sequence of random variables, each equidistributed on $T_\lambda^n$,
  \item $M = \lfloor \vartheta \cdot (n+1)^{-|\cX|} \cdot M' \rfloor$, 
  \item and 
  \begin{align*}
  \bbmE\left[\underset{W \in \cW}{\sup}\ \overline{e}(\cC', W^{\otimes n})\right] \leq 
  \frac{2}{\vartheta}(n+1)^{|\cX|}\mu + 2^{-M'(1-\vartheta)^2(n+1)^{-|\cX|}/\ln 2} .  
  \end{align*}
 \end{enumerate}
\end{proposition}
For proving the above assertion, we will make use of the following variant of the Chernov bound
\begin{proposition}\label{classical_chernov_bound}
 Let $n \in \bbmN$, $\delta > 0$ and $X = (X_1,\dots,X_n)$ be an i.i.d. sequence of random variables with 
 $0 \leq X_1,\dots, X_n \leq 1$ and $\bbmE[X_i] = E$, $i \in [n]$, then
 \begin{align*}
  \Pr\left(\frac{1}{n}\sum_{i=1}^n X_i \leq (1 - \delta) E\right) \leq 2^{-n\delta^2 E^2/\ln 2}
 \end{align*}
\end{proposition}

\begin{proof}[Proof of Proposition \ref{good_codes_2}]
 Define the event that at least $M$ codewords of a codebook $u$ are $\lambda$-typical sequences by
 \begin{align*}
  A(M) := \left\{u=(u_1,\dots,u_{M'}) \in \cX^{nM'}:\ |\{m:\ u_m \in T_{\lambda}^n\}| \geq M\right\}.
 \end{align*}
 For an i.i.d. sequence $U = (U_1,\dots, U_{M'})$ as in the hypotheses of the proposition, it holds 
 \begin{align*}
  \Pr\left(A(M)^c \right) = \Pr\left(\sum_{i=1}^{M'} \bbmeins_{T_{\lambda}^n}(U_m) < M\right).
 \end{align*}
 Notice that
 \begin{align*}
 \bbmE[\bbmeins_{T_{\lambda}^n}(U_m)] = \lambda^n(T_{\lambda}^n) \geq (n+1)^{-|\cX|}
 \end{align*}
 holds. The rightmost inequality above holds by the fact that among all typical sets of words in $\cX^n$, 
 $T_{\lambda}^n$ has the largest probability w.r.t. $\lambda^n$. \\
 We fix $\vartheta \in (0,1)$ and set $M := \lfloor\vartheta(n+1)^{-|\cX|}M'\rfloor$. We obtain
 \begin{align}
  \Pr\left(A(M)^c \right) \ 
  &\leq \ \Pr\left(\sum_{i=1}^{M'} \bbmeins_{T_{\lambda}^n}(U_m) < M\right) \nonumber \\
  &\leq \ \Pr\left(\frac{1}{M'}\sum_{i=1}^{M'} \bbmeins_{T_{\lambda}^n}(U_m) < \vartheta(n+1)^{-|\cX|}\right) \nonumber \\
  &\leq \ \Pr\left(\frac{1}{M'}\sum_{i=1}^{M'} \bbmeins_{T_{\lambda}^n}(U_m)< \vartheta \bbmE[U_1]\right) \nonumber \\
  &\leq \ \exp\left\{-M'(1-\vartheta)^2\bbmE[U_1]/\ln 2)\right\} \label{good_codes_2_1} \\
  &\leq \ \exp\left\{-M'(1-\vartheta)^2(n+1)^{-|\cX|})/\ln 2\right\} =: \tau. \nonumber
  \end{align}
  Except the one in Eq. (\ref{good_codes_2_1}), which is by application of the bound in Proposition 
  \ref{classical_chernov_bound}, all of the above estimates are by the preceding definitions and 
  bounds. \\ 
  Next, define a function $\varphi: \cX^{nM'} \rightarrow (T_{\lambda}^n)^M$ by
  \begin{align*}
   \varphi(u) := \begin{cases} 
		  (v_1,\dots, v_M) = v    & \text{if}\ u \in A(M)\\
                  (\tilde{v},\dots \tilde{v})    & \text{otherwise},
                 \end{cases}
  \end{align*}
  where $\tilde{v}$ is any word from $T_{\lambda}^n$. By symmetry, it holds
  \begin{align*}
   \lambda^{nM'}(\varphi^{-1}(v)) = \frac{\lambda^{nM'}(A(M))}{|T_{\lambda}^n|},
  \end{align*}
  i.e. the push forward measure $\lambda^{nM'}\circ \varphi^{-1}$ of $\lambda^{nM'}$ under $\varphi$ is 
  nearly equidistributed. Explicitly, 
  \begin{align*}
   \left\| \pi_{T_\lambda^n}^{\otimes M} - \lambda^{nM'}\circ \varphi^{-1} \right\|_1 
   = |1 - \lambda^{nM'}(A(M))| 
   = \lambda^{nM'}(A(M)^c) < \tau.
  \end{align*}
  For each outcome $v = (v_1,\dots,v_{M}) \in T_{\lambda}^{nM}$, we define an $(n,M)$-message transmission code
  $\cC(v) := (v_{m'},D_{m'})_{m'=1}^M$ as follows. Let 
  \begin{align}
   u_v := \underset{u \in \varphi^{-1}(v)}{\argmin}\ \sup_{t\in T} \ \overline{e}(\cC'(u), W_{t}^{\otimes n}), \label{good_codes_2_2}
  \end{align} 
  and $D_{m'} := D_m(u_v)$ for the respective $m$ (i.e. $(u_{v,m},D_m(u_v))$ is a pair of codeword and decoding set
  in $\cC'(u_v).)$ Then
  \begin{align*}
   \underset{t \in T}{\sup} \ \overline{e}(\cC(v), W_t^{\otimes n})
   & = \  \underset{t \in T}{\sup} \ \frac{1}{M} \sum_{m = 1}^M \tr\left(D_m(v)^\perp W_t^{\otimes n} (v_m)\right) \\
   & \ \leq \ \underset{t \in T}{\sup} \ \frac{1}{M} \sum_{m = 1}^{M'} \tr\left(D_m^\perp(u_v) W_t^{\otimes n} (u_{v,m})\right) \\
   & \ \leq \frac{M'}{M} \ \underset{t \in T}{\sup} \ \overline{e}(\cC'(u_v), W_t^{\otimes n}) \\
   &= \ \frac{M'}{M}  \underset{u \in \varphi^{-1}(v)}{\min} \ \underset{t \in T}{\sup}\ \overline{e}(\cC'(u_v), W_t^{\otimes n}). 
  \end{align*}
  The last equality above is by our code definition from Eq. (\ref{good_codes_2_2}). To each $u \in A(M)^c$, we assign some valid $(n,M)$-code $\cC(\varphi(u)) := (v_{0,m},D_m)_{m=1}^{M}$ being of 
  no further interest. Let $\hat{V} = \varphi(U)$ (which is not i.i.d. so far!), then $\cC(\hat{V}_m,D_m(\hat{V}))_{m=1}^M$ is a random constant composition $(n,M)$-code with
  \begin{align*}
   \bbmE\left[\underset{t \in T}{\sup} \ \overline{e}(\cC(\hat{V}), W_t^{\otimes n}) \right]
    & = \ \sum_{v \in \varphi(A(M))} \lambda^{nM'}(\varphi^{-1}(v)) \ \underset{t \in T}{\sup} \ \overline{e}(\cC(v),W_t^{\otimes n}) \\
    & \  + \sum_{v \in \varphi(A(M)^c)} \lambda^{nM'}(\varphi^{-1}(v)) \ \underset{t \in T}{\sup} \ \overline{e}(\cC(v),W_t^{\otimes n}) \\
    & < \  \sum_{v \in \varphi(A(M))} \lambda^{nM'}(\varphi^{-1}(v)) \ \underset{t \in T}{\sup} \ \overline{e}(\cC(v),W_t^{\otimes n}) + \tau \\
    & \leq \ \frac{M'}{M} \sum_{u \in A(M)} \lambda^{nM'}(u) \  \underset{t \in T}{\sup} \  \overline{e}(\cC'(u),W_t^{\otimes n}) + \tau \\
    & \leq \ \frac{M'}{M}  \bbmE\left[\underset{t \in T}{\sup}\  \overline{e}(\cC'(u),W_t^{\otimes n})\right] + \tau \\
    & \leq \frac{M'}{M} \mu + \tau.
  \end{align*}
  Now, let $V = (V_1,\dots,V_{M})$ be a sequence of i.i.d. random variables each equidistributed on $T_{\lambda}^n$. And $\cC(V) := (V_m,D_m(V))_{m=1}^M$. It holds
  \begin{align*}
   \frac{M'}{M}  \bbmE\left[\underset{t \in T}{\sup}\ \overline{e}(\cC'(u),W_t^{\otimes n})\right] 
   & \ \leq \frac{M'}{M}  \bbmE\left[\underset{t \in T}{\sup} \overline{e}(\cC'(u),W_t^{\otimes n})\right] \\
   & + \left\|\pi_{T_\lambda^n}^{\otimes M} - \lambda^{nM'}\circ \varphi^{-1} \right\|_1 \\
   & < \ \frac{M'}{M} \mu +2 \tau.
  \end{align*}
  Since 
  \begin{align*}
   \frac{M'}{M} \leq \frac{2}{\vartheta}(n+1)^{|\cX|},
  \end{align*}
  we are done.
 \end{proof}

\end{section}
\begin{section}{Continuity bounds} \label{appendix:continuity_bounds}
For convenience of the reader, we state and prove several continuity properties of entropic quantities. Most of them follow straightforwardly from the Alicki-Fannes bound \cite{alicki04} for von Neumann entropies.
\begin{theorem}[\cite{alicki04}] \label{alicki_fannes}
 Let $\rho, \sigma \in \cS(\cK_A \otimes \cK_B)$ be states on $\cK_A \otimes \cK_B$ with $\|\rho - \sigma\| \leq \epsilon$. It holds
 \begin{align*}
  |S(A|B,\rho) - S(A|B,\sigma)| \leq 4\epsilon \log \dim \cK_A + 2 h(\epsilon),
 \end{align*}
where $h(x) := -x \log x - (1-x)\log (1-x)$ is the binary Shannon entropy of $(x, 1-x)$.
\end{theorem}
The following bound is easily derived from Theorem \ref{alicki_fannes}.
\begin{lemma} \label{lemma:holevo_bound_1}
 Let $p,q \in \fP(\cY)$ be probability distributions with $\|p-q\|_1 \leq \epsilon$. For each cq-channel $V \in \cq(\cY, \cK)$, it holds
 \begin{align*}
  |\chi(p,V)- \chi(q,V) | \leq 6 \epsilon \log\dim \cK + 2 h(\epsilon).
 \end{align*}
\end{lemma}

\begin{lemma} \label{lemma:holevo_bound_2}
Let $\cQ,\cQ' \subset \fP(\cY)$ be probability distributions with $d_H(\cQ,\cQ') \leq \epsilon$. For each set $\cV \subset \cq(\cY,\cK_B \otimes \cK_E)$
 \begin{align*}
  &\left| \underset{q \in \cQ}{\inf}\left(\underset{V \in \cV}{\inf} \chi(p,V_B)- \underset{V \in \cV}{\sup} \chi(q,V_E) \right)
  - \underset{q \in \cQ'}{\inf} \left(\underset{V \in \cV}{\inf}  \chi(p,V_B)- \underset{V \in \cV}{\sup} \chi(q,V_E)\right) \right| \\
  &\leq 6 \epsilon \log \dim \cK_{BE} + 4 h(\epsilon).
 \end{align*}
\end{lemma}

\begin{lemma}\label{lemma:fannes_cont_mut}
 Let $\fI, \fJ \subset \cS_{cqq}(\cY, \cK_X)$ of cqq density matrices and $d_H(\fI,\fJ) \leq \Gamma$, and with stochastic matrices $P_{U|Y}: \ \cY \rightarrow \fP(\cU)$ and $P_{T|U}:\ \cU \rightarrow \fP(\cT)$
 \begin{align*}
  \tilde{\rho} := \sum_{t \in \cT} \sum_{u \in \cU} \sum_{y \in \cY} \ P_{T|U}(t|u) P_{U|Y}(u|y) p(y) \ket{u}\bra{u} \otimes  \ket{t}\bra{t} \otimes V(y)
 \end{align*}
 if
 \begin{align*}
  \rho := \sum_{y \in \cY} p(y) \ket{y} \bra{y} \otimes V(y).
 \end{align*}
Then, the inequalities
\begin{align*}
  \underset{\rho \in \fI}{\inf} \ I(U;X|T, \tilde{\rho}) &\geq \underset{\rho \in \fJ}{\inf} \ I(U;X|T, \tilde{\rho}) - 8 \Delta \log(|\cU| \cdot \dim \cK) - 6 h(\Delta)   \\
  \underset{\rho \in \fI}{\sup} \ I(U;X|T, \tilde{\rho}) &\leq \underset{\rho \in \fJ}{\sup} \ I(U;X|T, \tilde{\rho}) +  8 \Delta \log(|\cU| \cdot \dim \cK) + 6 h(\Delta)
\end{align*}
are valid.
\end{lemma}
\begin{proof}
 It holds for any two states $\rho, \sigma \in \cS_{cqq}(\cY, \cK_X)$
 \begin{align}
  \|\tilde{\rho} - \tilde{\sigma}\|_1 \leq \|\rho - \sigma\|_1. \label{fannes_cond_mut_mon}
 \end{align}
 Note that for each cqq density matrix $\rho$, it holds
 \begin{align*}
  I(U;X|T, \tilde{\rho}) = S(U|T, \tilde{\rho}) + S(X|T, \tilde{\rho}) - S(UX|\tilde{\rho}).
 \end{align*}
 by definition of the quantum mutual information. If $\rho, \sigma$ fulfill $\|\rho - \sigma \|_1 \leq \delta$, then
 \begin{align}
  |I(U;X|T, \tilde{\rho}) - I(U;X|T, \tilde{\sigma})| \leq 8 \delta \log(|\cU| \cdot \dim \cK) + 6 h(\delta) \label{fannes_cond_mut}
 \end{align}
 holds by (\ref{fannes_cond_mut_mon}) and application of Lemma \ref{alicki_fannes}. From (\ref{fannes_cond_mut}) and the assumptions, we directly infer the claims.
\end{proof}
\end{section}
\begin{section}{Proof of Eq. (\ref{rate_subopt_bound})} \label{ineq_appendix}
 Let $\lambda \in \hat{\fT}$ and $p \in \cP_\lambda$. We define for each $q \in \cP_\lambda$ the set
 \begin{align*}
 \fI^{B}_{(p,q)} := \left\{\sum_{y \in \cY} p(y) \ket{y} \bra{y} \otimes V_B: \ V \in \cV_q \right\}.  
 \end{align*}
 Observe the relations
 \begin{align*}
  \hat{\fI}^B_{p,\lambda} = \underset{q \in \cP_\lambda}{\bigcup} \fI^B_{(p,q)}, \hspace{.2cm} \text{and} \hspace{.4cm} \fI^B_p = \fI^B_{(p,p)}.
 \end{align*}
 assume $V \in \cV_p$, $V' \in \cV_q$. It holds
  \begin{align*}
  \left\|\sum_{y \in \cY} p(y) \ket{y}\bra{y} \otimes V(y) - \sum_{y \in \cY} p(y) \ket{y}\bra{y} \otimes V'(y)  \right\|_1
  &= \sum_{y \in \cY} \|p(y)V(y) - p(y) V'(y) \|_1 \\
  &\leq \sum_{y \in \cY} \|p(y)V(y) - q(y) V'(y) \|_1  + \|p-q\|_1,
 \end{align*}
 from which we directly infer
 \begin{align*}
  d_H(\fI^B_p, \fI_{(p,q)}) \leq d_H(\fI^B_p,\fI^B_q) + \|p-q\|_1.
 \end{align*}
 Using these facts together with the first claim of Lemma \ref{lemma:fannes_cont_mut}, we obtain
 \begin{align}
  \underset{\rho \in \hat{\fI}_{p,\lambda}}{\inf} I(U_\lambda;B|T_\lambda,\tilde{\rho})
  &= \underset{\sigma \in \hat{\fI}^B_{p,\lambda}}{\inf} I(U_\lambda;B|T_\lambda,\tilde{\sigma}) \nonumber \\
  &= \underset{q\in \cP_\lambda}{\inf} \underset{\sigma \in \fI^B_{(p,q)}}{\inf} I(U_\lambda;B|T_\lambda,\tilde{\sigma}) \nonumber \\
  &\geq  \underset{\sigma \in \fI^B_{p}}{\inf} I(U_\lambda;B|T_\lambda,\tilde{\sigma}) - 8 \Delta \log(|\cU| \cdot \dim \cK_B) - 6 h(\Delta) \nonumber \\
  &=  \underset{\rho \in \fI_{p}}{\inf} I(U_\lambda;B|T_\lambda,\tilde{\rho}) - 8 \Delta \log(|\cU| \cdot \dim \cK_B) - 6 h(\Delta) \label{fannes_calc_1}
\end{align}
 Applying a similar reasoning leads us to 
 \begin{align}
 \underset{\rho \in \hat{\fI}_{p,\lambda}}{\sup} I(U_\lambda;E|T_\lambda,\tilde{\rho}) \leq \underset{\rho \in \fI_{p}}{\sup} I(U_\lambda;E|T_\lambda,\tilde{\rho}) - 8 \Delta \log(|\cU| \cdot \dim \cK_E) - 6 h(\Delta). \label{fannes_calc_2}
 \end{align}
 Combination of (\ref{fannes_calc_1}) and (\ref{fannes_calc_2}) for all $p \in \cP_\lambda$ yields the desired bound 
  \begin{align*}
   &\underset{\rho \in \hat{\fI}_{p,\lambda}}{\inf} I(U_\lambda;B|T_\lambda,\tilde{\rho}) - \underset{\rho \in \hat{\fI}_{p,\lambda}}{\sup} I(U_\lambda;E|T_\lambda,\tilde{\rho}) 
    \\ &\geq \underset{\rho \in \fI_p}{\inf} I(U_\lambda;B|T_\lambda,\tilde{\rho}) - \underset{\rho \in \fI_p}{\sup} I(U_\lambda;E|T_\lambda,\tilde{\rho}) - 16 \Delta \log(|\cU| \cdot \dim \cK_{BE}) - 12 h(\Delta). 
  \end{align*}

\end{section}

\end{appendix}

\end{document}